\documentclass[a4paper,11pt]{article}
\usepackage[utf8]{inputenc}
\usepackage[T1]{fontenc}
\usepackage{graphicx}
\usepackage{amsmath,amsfonts,amssymb,amsthm}
\usepackage{mathtools}
\usepackage{mathrsfs}
\usepackage{bm}
\usepackage{xcolor}%
\usepackage{fullpage}

\definecolor{myred}{rgb}{0.77, 0.0, 0.1}
\definecolor{crimson}{rgb}{0.86, 0.08, 0.24}
\definecolor{awesome}{rgb}{1.0, 0.13, 0.32}
\definecolor{newgreen}{rgb}{0.0,0.6,0.0}
\definecolor{trueblue}{rgb}{0.0, 0.45, 0.81}
\definecolor{softblue}{RGB}{6,90,255}
\definecolor{bluemod}{RGB}{0,0,205}

\usepackage{hyperref}
\hypersetup{
  colorlinks, %
  linkcolor=blue, %
  citecolor=blue, %
  linktoc=page, %
  urlcolor=bluemod
}

\renewcommand{\leq}{\leqslant}
\renewcommand{\geq}{\geqslant}

\newcommand{\di}{\mathrm{d}}

\newcommand{\eps}{\varepsilon}

\newcommand{\wt}{\widetilde}

\newcommand{\ol}{\overline}
\newcommand{\pp}{\, : \;}

\newcommand{\N}{\mathbf N}
\newcommand{\Z}{\mathbf Z}

\newcommand{\R}{\mathbf R}

\usepackage{xspace}
\newcommand{\ie}{i.e.\@\xspace} 
\newcommand{\eg}{e.g.\@\xspace}
\newcommand{\iid}{i.i.d.\@\xspace}

\newcommand{\ceil}[1]{\lceil #1 \rceil}

\newcommand{\innerp}[2]{\langle #1 , #2 \rangle}

\newcommand{\tr}{\mathrm{Tr}}

\renewcommand{\O}{\mathsf{O}}

\newcommand{\norm}[1]{\|#1\|}

\newcommand{\diam}{\mathop{\mathrm{diam}}}
\DeclareMathOperator{\vol}{Vol}
\newcommand{\dist}{\mathrm{dist}}
\newcommand{\conv}{\mathrm{conv}}

\newcommand{\E}{\mathbb E}
\renewcommand{\P}{\mathbb P}

\newcommand{\probas}{\mathcal{P}}
\newcommand{\indic}[1]{\bm 1 ( #1 )}
\newcommand{\kl}{\mathrm{KL}}%
\newcommand{\kll}[2]{\kl(#1, #2)}%

\newcommand{\tv}[1]{\norm{#1}_{\mathrm{TV}}}%

\newcommand{\gaussdist}{\mathcal{N}}

\newcommand{\F}{\mathcal{F}}

\newcommand{\Y}{\mathcal{Y}}
\newcommand{\Zs}{\mathcal{Z}}

\newcommand{\A}{\mathcal{A}}

\newcommand{\regret}{\mathrm{Regret}}

\newcommand{\mmr}{\mathsf{R}^*}
\newcommand{\mred}{\mathsf{Red}}%
\newcommand{\corr}{\mathsf{Corr}}
\newcommand{\Ts}{\mathcal{T}}
\newcommand{\pr}{\mathsf{p}}
\newcommand{\fs}{\mathsf{f}}

\newtheorem{proposition}{Proposition}[section]%
\newtheorem{theorem}{Theorem}[section]
\newtheorem{lemma}{Lemma}[section]
\newtheorem{corollary}{Corollary}[section]

\theoremstyle{definition}
\newtheorem{definition}{Definition}[section]

\theoremstyle{remark}
\newtheorem{remark}{Remark}[section]

\numberwithin{equation}{section}

\title{Universal coding, intrinsic volumes, and metric complexity%
}
\author{Jaouad Mourtada\thanks{%
    CREST, ENSAE,
    Institut polytechnique de Paris; \href{mailto:jaouad.mourtada@ensae.fr}{jaouad.mourtada@ensae.fr}}
}
\date{\today%
}

\begin{document}
\maketitle

\begin{abstract}
  We study sequential probability assignment in the Gaussian setting, where the goal is to predict, or equivalently compress, a sequence of real-valued observations almost as well as the best Gaussian distribution with mean constrained to a given subset of $\R^n$.
  First, in the case of a convex constraint set $K$, we express the hardness of the prediction problem (the minimax regret) in terms of the intrinsic volumes of $K$; specifically, it equals the logarithm of the Wills functional from convex geometry.
  We then establish a comparison inequality for the Wills functional in the general nonconvex case, which underlines the metric nature of this quantity and generalizes the Slepian-Sudakov-Fernique comparison principle for the Gaussian width.
  Motivated by this inequality, we characterize the exact order of magnitude of the considered functional for a general nonconvex set, in terms of global covering numbers and local Gaussian widths.
  This implies sharp estimates, of metric nature, on the log-Laplace transform of the intrinsic volume sequence of a convex body.
  As part of our analysis, we also characterize the minimax redundancy for a general constraint set.
  We finally relate and contrast our findings with classical asymptotic results in information theory.
\end{abstract}


\tableofcontents

\section{Introduction}
\label{sec:introduction}

In this work we study universal coding in a Gaussian setting, %
and relate the hardness of this prediction problem to the geometry of the ``model''.
This information-theoretic problem will be first connected to convex geometry, and then studied from a distinct metric perspective inspired by the theory of Gaussian processes.
In an attempt to make this text accessible to readers from various fields, we have included some background material in the first two sections.

\paragraph{Sequential prediction.}

Sequential probability assignment is a basic problem in information theory, statistics and learning theory, with connections to
sequential prediction~\cite{cesabianchi2006plg,vovk1998mixability}, %
lossless coding and data compression~\cite{cover2006elements,csiszar2004information,merhav1998universal,gassiat2018universal,polyanskiy2022information},
minimum description length~\cite{rissanen1985mdl,barron1998minimum,grunwald2007mdl}, as well as aggregation and estimation theory in statistics~\cite{catoni2004statistical,yang1999information,yang2000mixing}.
Informally, it can be formulated as follows: given a sequence of observations unknown a priori, the goal is to assign a joint probability as large as possible to this sequence.

Specifically, let $y_{1:n} = (y_1, \dots, y_n) \in \Y^n$ be a sequence of {observations} taking value in some measurable space $\Y$, which %
we aim to predict.
Here, the forecast is not ``punctual'' (that is, consisting of another sequence that should be ``close'' to the true one in a suitable sense) but rather ``probabilistic'': it consists in assigning probabilities to all possible sequences, with the aim of attributing a probability as large as possible to the actual data sequence $y_{1:n}$.
Given a reference measure $\mu$ on $\Y$, a prediction then consists of a positive probability density $p : \Y^n \to \R^+$ with respect to $\mu^n$, %
whose error on the sequence $y_{1:n}$ is measured by the \emph{logarithmic loss} (or \emph{negative log-likelihood})
\begin{equation}
  \label{eq:log-loss}
  \ell (p, y_{1:n})
  = - \log p (y_1, \dots, y_n)
  \, .
\end{equation}

A few comments on the definition and interpretation of the problem are in order.

First, the loss function~\eqref{eq:log-loss} has a precise and classical interpretation in terms of coding and data compression, see \eg \cite[Chapter~5]{cover2006elements},~\cite[Chapter~1]{catoni2004statistical},~\cite[Chapter~1]{gassiat2018universal} and~\cite[Chapter~13]{polyanskiy2022information}. %
Roughly speaking, if $\Y$ is finite and $\mu$ is the counting measure, there is an almost exact correspondence between probability densities $p$ on $\Y^n$ on the one hand, and on the other hand \emph{codes} $\phi : \Y^n \to \bigcup_{k \geq 1} \{0, 1\}^k$ compressing each sequence $y_{1:n}$ into a string of bits $\phi (y_{1:n})$, with a suitable form of injectivity allowing for decoding.
In addition, under this correspondence, the length %
of the code $\phi (y_{1:n})$ associated to the sequence $y_{1:n}$ coincides (up to an integer rounding error) to the logarithmic loss $- \log_2 p (y_{1}, \dots, y_n)$ in base $2$.
(The precise statement underpinning this correspondence, namely \emph{Kraft's inequality}, is
included in Section~\ref{sec:coding-minimax-regret} and 
in the aforementioned references. 
In addition, the correspondence extends from finite sets to continuous spaces $\Y$ by discretization using an increasing sequence of finite partitions.)
Because of this, in what follows we will sometimes refer to probability densities $p$ as ``codes'', and to the log-likelihood $\ell (p, y_{1:n})$ as the ``code-length''.

Second, the problem admits an equivalent %
``sequential'' formulation.
Indeed, in the formulation above, a code $p$
on $\Y^n$ must be chosen ahead of time, prior to seeing any observation.
A perhaps more appealing formulation %
is in terms of \emph{sequential prediction} (see~\cite{cesabianchi2006plg} for a thorough exposition of this topic).
Here, observations %
are revealed one at a time; at each step, one wishes to issue a probabilistic prediction of the next observation %
knowing the previous ones.
In this context, a \emph{sequential prediction strategy}, or sequential coding scheme, is a sequence 
$(p_1, \dots, p_n)$, where %
$p_i$ is a function associating a density $p_i (\cdot | y_1, \dots, y_{i-1})$ on $\Y$ to each $(y_1, \dots, y_{i-1}) \in \Y^{i-1}$.
Now, each such $p_i$ can be seen as a conditional density on $y_i$ given $y_{1}, \dots, y_{i-1}$, hence a sequential prediction strategy $(p_1, \dots, p_n)$ amounts to a family of such conditional densities, which in turns amounts to a single joint density $p$ on $\Y^n$ defined by $p (y_1, \dots, y_n) = p_1 (y_1) p_2 (y_2|y_1) \cdots p_n (y_n | y_1, \dots, y_{n-1})$.
In addition, the cumulative loss (total code-length) of the sequential prediction scheme associated to $p$ coincides with the loss~\eqref{eq:log-loss}:
\begin{equation*}
  - \sum_{i=1}^n \log p_i (y_i | y_1, \dots, y_{i-1})
  = - \log p (y_1, \dots, y_n)
  = \ell (p, y_{1:n})
  \, .
\end{equation*}
This remarkable tensorization property %
shows that
the previously defined problem is in fact equivalent to sequential prediction.

\paragraph{Statistical model and minimax regret.}

As formulated, the problem of minimizing the loss~\eqref{eq:log-loss} is ill-defined.
Indeed, the normalization constraint---that $p$ must be a probability density on $\Y^n$---implies that no density can assign a high probability to all sequences.
This expresses the fact that no code can efficiently compress all sequences.
It follows from this obstruction that one should specify on which ``structured'' sequences a small loss should be incurred, and quantify what constitutes a good performance.

A common way to achieve this goal is to resort to statistical modeling of the data sequence.
Concretely, a \emph{statistical model} $\probas$ is a set of probability densities on $\Y^n$.
Perhaps the most classical way to use a model $\probas$ is to assume that the sequence of observations $Y_{1:n}$ is random, with density $p^*$ for some (unknown) $p^* \in \probas$.
In this case, the code with smallest expected code-length is the true density $p^*$, and if $p$ is another density then $\E [ \ell (p, Y_{1:n}) - \ell (p^*, Y_{1:n}) ] = \kll{p^*}{p}$, the Kullback-Leibler divergence between $p^*$ and $p$ \cite{kullback1951information,cover2006elements} (see Section~\ref{sec:coding-minimax-regret} for more details).

An alternative approach (see~\cite{cesabianchi2006plg} for a historical account on this perspective) %
is to avoid making explicit assumptions on the data sequence $y_{1:n} \in \Y^n$, which may be arbitrary, and instead use the statistical model $\probas$ as a benchmark.
Specifically, define the \emph{regret} (sometimes also called \emph{redundancy}) of a density $q$ (not necessarily belonging to $\probas$) with respect to the model $\probas$ over the sequence $y_{1:n} \in \Y^n$ as 
\begin{equation}
  \label{eq:regret}
  \regret (q, \probas, y_{1:n})
  = \ell (q, y_{1:n}) - \inf_{p \in \probas} \ell (p, y_{1:n})
    \, .
\end{equation}
That is, the regret is the difference in code-length between the code $q$ and the best code (a posteriori) in the model $\probas$ for the sequence $y_{1:n}$.
Since no assumption is made on the sequence $y_{1:n}$, one can evaluate the performance of a distribution $q$ relative to a class $\probas$ by the worst-case (maximum) regret over all sequences.
This leads to considering the \emph{minimax regret}, which is the best regret guarantee (with respect to $\probas$) one can ensure uniformly over all sequences:
\begin{equation}
  \label{eq:minimax-regret}
  \mmr (\probas)
  = \inf_{q} \sup_{y_{1:n} \in \Y^n} \regret (q, \probas, y_{1:n})
  \, ,
\end{equation}
where the infimum is taken over all probability densities $q$ on $\Y^n$.
Clearly, the richer the model $\probas$, the harder it is to compete against it, and hence the larger $\mmr (\probas)$ is.
As such, the minimax regret measures the complexity of the model,
from the perspective of worst-case sequential prediction.

\paragraph{Gaussian setting.}

In this work, we consider sequential probability assignment in a Gaussian setting, where the observations $y_1, \dots, y_n$ are real-valued and modeled by a subset of the Gaussian sequence model.
This problem can be seen as a coding-theoretic analogue of Gaussian regression.
Specifically, for $\theta \in \R^n$, denote by $p_{\theta} : y \mapsto (2\pi)^{-n/2} e^{-\norm{y-\theta}^2/2}$ the density of the Gaussian measure $\gaussdist (\theta, I_n)$ with respect to the Lebesgue measure $\mu (\di y) = \di y$ on $\R^n$.
To any nonempty subset $A \subset \R^n$ (called model or constraint set), one can associate the corresponding subset of the Gaussian model and its corresponding regret:
\begin{equation}
  \label{eq:gaussian-model-subset}
  \probas_A
  = \big\{ p_\theta : \theta \in A \big\}
  \, ;
  \qquad
  \mmr (A)
  = \mmr (\probas_A)
  \, .
\end{equation}
With these definitions at hand, the general question we consider is the following:
\begin{quote}
  \emph{How does the complexity $\mmr (A)$ of sequential prediction depend on the geometry of the constraint set $A \subset \R^n$?}
\end{quote}
Part of the results we will obtain will be specific to the case where $A$ is a convex body, but most of the main results remain in fact valid for general (possibly non-convex) sets.
In addition, while our primary focus will be on the worst-case deterministic setting (where hardness is measured by the minimax regret $\mmr (A)$), we will also study the aforementioned statistical setting.

\paragraph{Summary of contributions.}

In order to convey the main thread of ideas in this work, 
we provide a sample of our %
results.
We refer to the corresponding sections for more complete statements and discussion.

The first result states that, in the case where $A = K$ is a convex body in $\R^n$, the minimax regret $\mmr (K)$ can be expressed in terms of its \emph{intrinsic volumes} $V_j (K)$, $j = 0, \dots, n$.
We refer to Section~\ref{sec:convex-bodies} for a definition of the intrinsic volumes of a convex body.

\begin{theorem}[see Theorem~\ref{thm:minimax-wills-volumes}]
  \label{thm:short-minimax-volumes}
  If $K \subset \R^n$ is a convex body, the minimax regret $\mmr (K)$ equals
  \begin{equation*}
    \mmr (K)
    = \log \bigg( \sum_{j=0}^n V_j %
    \big( K / \sqrt{2\pi} \big)
    \bigg)
    \, .
  \end{equation*}
\end{theorem}

This result connects universal coding to classical convex geometry~\cite{schneider2013convex}.
In particular, it shows that for convex bodies, the minimax regret can be understood in terms of the volume behavior of its random projections.
This connection is however specific to convex bodies, as it involves quantities such as intrinsic volumes, which are naturally defined on convex sets.

A rather different perspective, which we describe next, extends to general nonconvex sets.
An important property of the functional $\mmr$ is that it is a ``metric'' quantity, which can be understood %
in terms of distances.
This is suggested by the following comparison inequality, which is analogous to the Slepian-Sudakov-Fernique comparison inequality for suprema of Gaussian processes (see Section~\ref{sec:comp-theor-wills} for additional discussion).
Here, $\norm{\cdot}$ denotes the Euclidean norm.

\begin{theorem}[see Theorem~\ref{thm:comparison-wills}]
  \label{thm:short-comparison}
  Let $A, B \subset \R^n$ be two nonempty sets.
  If there exists a function $\varphi : A \to B$ such that
  $\varphi (A) = B$ and $\norm{\varphi (x) - \varphi (y)} \leq \norm{x - y}$ for every $x, y \in A$,
  then $\mmr (A) \leq \mmr (B)$.
\end{theorem}

Theorem~\ref{thm:short-comparison} suggests that the minimax regret $\mmr (A)$ can be understood in terms of distances alone, though it does not provide an explicit metric characterization.
An explicit characterization can be obtained, relating (up to universal constant factors) the minimax regret $\mmr (A)$ to local Gaussian widths and global covering numbers.
Below, for $r > 0$ we denote by $N (A, r)$ the smallest number of closed balls of radius $r$ that cover $A$, by $B (\theta, r)$ the closed Euclidean ball of radius $r$ centered at $\theta \in \R^n$, and by $w (A) = \E \sup_{\theta \in A} \sum_{i=1}^n \theta_i X_i$ the Gaussian width of $A$, where $X_1, \dots, X_n$ are independent standard Gaussian variables.

\begin{theorem}[see Theorem~\ref{thm:minimax-regret-metric}]
  \label{thm:short-metric-characterization}
  For any subset $A \subset \R^n$, one has
  \begin{equation*}
    \mmr (A)
    \asymp \inf_{r \geq 0} \Big\{ \sup_{\theta \in A} w \big( A \cap B (\theta, r) \big) + \log N (A, r) \Big\}
    \, .
  \end{equation*}  
\end{theorem}

Here and throughout this article, given two quantities $f$ and $g$, we use the notation $f \asymp g$ to mean that there exist universal constants $c_1, c_2 > 0$ such that $c_1 f \leq g \leq c_2 f$.

Using Talagrand and Fernique's Majorizing Measure theorem~\cite{talagrand1987regularity,talagrand2021upper}, a characterization of $\mmr (A)$ involving only distances can be deduced from Theorem~\ref{thm:short-metric-characterization} (Corollary~\ref{cor:majorizing-measure-regret}).
This characterization features a certain truncated version of %
the
generic chaining functional.

Theorems~\ref{thm:short-minimax-volumes} and~\ref{thm:short-metric-characterization} provide two different characterizations of the minimax regret over a convex body:
the first involving quantities from classical convexity (intrinsic volumes), the second in terms of metric quantities that appear in high-dimensional geometry, probability and statistics.
Relating these two characterizations leads to the following result, 
which provides precise quantitative estimates on the intrinsic volume sequence of a convex body.

\begin{corollary}[see Corollary~\ref{cor:metric-intrinsic-volumes}]
  \label{cor:short-isomorphic-volumes}
  For any convex body $K \subset \R^n$ and $\lambda \geq 0$, one has
  \begin{equation*}
    \log \bigg( \sum_{j=0}^n V_j (K) \lambda^j \bigg)
    \asymp \inf_{r \geq 0} \Big\{ \lambda \sup_{\theta \in K} w \big( K \cap B (\theta, r) \big) + \log N (K, r) \Big\}
    \, .
  \end{equation*}
\end{corollary}

Let us mention that the sum of intrinsic volumes of a convex body is known in convex geometry as the Wills functional~\cite{wills1973gitterpunktanzahl,hadwiger1975wills}. %
Corollary~\ref{cor:short-isomorphic-volumes} then corresponds to an ``isomorphic'' characterization of the Wills functional, as discussed in Section~\ref{sec:isom-char-wills}.

\paragraph{Organization.}

This paper is organized as follows.
After some additional background on sequential probability assignment and on intrinsic volumes, Section~\ref{sec:minimax-wills} is devoted to the relation between the two in the case of convex bodies.
Section~\ref{sec:infin-dimens-vers} is concerned with a qualitative study of the functional $\mmr$ in the general nonconvex case, including its basic properties, infinite-dimensional version and its finiteness, the link with Gaussian width as well as small and large-scale asymptotics.
Then, Section~\ref{sec:comp-metric-estimates} is devoted the quantitative study of $\mmr$ from a metric perspective; as part of the analysis, we also determine the hardness of sequential prediction in the ``statistical'' or ``average-case'' setting.
The quantitative characterizations of Section~\ref{sec:comp-metric-estimates} hold up to (universal) constant factors.
In Section~\ref{sec:conc-prop-noise}, we complement them by considering some exact structural properties of the regret, involving its interaction with the additive structure of the space.
This allows in particular to establish a form of ``strong monotonicity'' with respect to the noise level.
In Section~\ref{sec:non-asympt-param}, we relate and contrast our findings to classical asymptotic results from the literature on universal coding and minimum description length.
Finally, we consider some concrete examples in Section~\ref{sec:examples}, including the case of ellipsoids.

\paragraph{Related work.}

Sequential prediction is a topic at the intersection of Game theory, Statistics, Information theory and Learning theory.
We refer to the classic textbook~\cite{cesabianchi2006plg} for an exposition of the many facets of the subject.
We point out some results in sequential prediction which, although not directly related to the problem at hand, %
also provide precise information in their considered settings.
First, in the context of sequential prediction of a binary sequence with absolute error loss (see~\cite[Chapter~8]{cesabianchi2006plg}), it is shown by Cesa-Bianchi and Lugosi~\cite{cesabianchi1999prediction} that the minimax regret is equal to the Rademacher average of the class of experts.
In addition, a precise description of achievable error bounds for prediction of binary sequences has been obtained by Cover~\cite{cover1967behavior}.
In the context of prediction of bounded sequences under square loss, Rakhlin and Sridharan~\cite{rakhlin2014online_nonparametric} use certain sequential centered offset Rademacher complexities to sharply control the minimax regret.

A canonical instance of sequential prediction is sequential probability assignment,
see~\cite[Chapter~9]{cesabianchi2006plg} for an introduction to this problem.
As discussed previously, it admits a precise interpretation in terms of universal coding~\cite{cover2006elements,catoni2004statistical,gassiat2018universal,polyanskiy2022information} in information theory.
The seminal work of Shtarkov~\cite{shtarkov1987universal} identifies the general properties of the problem,
and provides an integral representation of the minimax regret which will be reviewed in Section~\ref{sec:coding-minimax-regret}.
In addition, a classical result of Rissanen~\cite{rissanen1996fisher} provides an asymptotic expansion of the regret for a fixed finite-dimensional model in the limit of large ``sample size'',
whose relation to the present work we discuss in more detail in Section~\ref{sec:non-asympt-param}.
Finally, Opper and Haussler~\cite{opper1999worst} obtain upper bounds on the minimax regret in terms of certain covering numbers of the class;
these bounds have been refined in a series of works~\cite{cesabianchi2001logarithmic,rakhlin2015sequential,grunwald2019tight,bilodeau2020tight}.
These upper bounds do not apply to the Gaussian setting we consider, but they are
nonetheless related to the upper bound in Corollary~\ref{cor:majorizing-measure-regret}.
However, they do not come with matching lower bounds.

Intrinsic volumes are equivalent (up to multiplicative constants) to quermassintegrals of convex bodies, themselves a special case of mixed volumes.
These quantities are at the core of classical convex geometry, the so-called ``Brunn-Minkowski theory''; we refer to the classic book of Schneider~\cite{schneider2013convex} for a thorough reference on this topic.
The sum of intrinsic volumes, known as \emph{Wills functional}, has been considered by Wills~\cite{wills1973gitterpunktanzahl} in the context of lattice point enumeration.
It has been used in several contexts, including the study of Gaussian processes~\cite{vitale1996wills} and in establishing concentration of the sequence of intrinsic volumes~\cite{lotz2020concentration}.
Several notable inequalities on this functional have been obtained
by Hadwiger~\cite{hadwiger1975wills}, McMullen~\cite{mcmullen1991inequalities}, and more recently Alonso-Gutiérrez, Hernández Cifre, and Yepes Nicolás~\cite{alonso2021further}.
These references are highly relevant to the present work, and we will discuss some of their results in the following sections.
We contribute to the study of the Wills functional in two primary ways, first by %
connecting it to sequential coding, and second by
characterizing it (under a suitable normalization) up to universal constant factors.

A singular aspect of Theorem~\ref{thm:short-minimax-volumes} is that it features intrinsic volumes, which are %
uncommon complexity measures in high-dimensional statistics (apart from the first intrinsic volume, namely the Gaussian mean width).
A notable exception is the important %
work~\cite{amelunxen2014living}, which characterizes phase transitions in random convex programs in terms of conic (spherical) intrinsic volumes.
Despite the appearance of certain intrinsic volumes in both~\cite{amelunxen2014living} and the present work, the approaches are different.
First, the considered problems differ:
\cite{amelunxen2014living} deals with linear inverse problems with random measurements
(see~\cite{candes2006near,rudelson2008sparse,donoho2009counting,stojnic2009various,chandrasekaran2012convex,foucart2013compressive} and references therein for more information%
), while we consider universal coding.
Second, \cite{amelunxen2014living} features \emph{conic} intrinsic volumes (associated to certain sub-differential cones on the edge of the convex domain), in contrast with the \emph{Euclidean} intrinsic volumes that appear in our context.
These quantities are of a different character: conic intrinsic volumes of sub-differential cones are scale-free and reflect the differential and infinitesimal structure of the boundary of the domain, while Euclidean intrinsic volumes are scale-sensitive global complexity measures.
This qualitative difference 
reflects the fact that~\cite{amelunxen2014living} consider either noiseless exact recovery or noisy recovery in the low-noise limit, while in our context the noise level is fixed.

\paragraph{Notation.}

Throughout this article, the standard scalar product on the Euclidean space $\R^n$ or on the sequence space $\ell^2 = \{ (x_i)_{i \geq 1} \in \R^{\N^*} : \sum_{i \geq 1} x_i^2 < +\infty \}$ is denoted by $\langle \cdot, \cdot \rangle$; we also set $\norm{x} = \langle x, x\rangle^{1/2}$ to be the associated $\ell^2$ norm.
We denote by $B_2^n = \{ x \in \R^n : \norm{x} \leq 1 \}$ the closed unit $\ell^2$ ball in $\R^n$, and let $B_2$ be either $B_2^n$ or the unit ball in $\ell^2$ depending on the ambient space; likewise, for $\theta$ in $\R^n$ or in $\ell^2$ and for $r \geq 0$, we let $B (\theta, r) = \{ x : \norm{x-\theta} \leq r \} = \theta + r B_2$ denote the closed ball of radius $r$ centered at $\theta$.
The $n \times n$ identity matrix is denoted $I_n \in \R^{n \times n}$.
For a measurable subset $A \subset \R^n$, we let $\vol_n (A)$ denote its volume, namely its $n$-dimensional Lebesgue measure.
In addition, the diameter of $A \subset \R^n$ is $\diam (A) = \sup_{x, y \in A} \norm{x - y}$, while the distance of $x \in \R^n$ to $A$ is denoted by $\dist (x, A) = \inf_{\theta \in A} \norm{\theta - x}$.
Given two sets $A, B \subset \R^n$, their \emph{Minkowski sum} is the set $A + B = \{ x + y : x \in A, y \in B \}$; also, $A-B = A+ (-B)$.
In particular, for $r \geq 0$ the \emph{parallel set} of $A$ with radius $r$ is $\{ x \in \R^n : \dist (x, A) \leq r \} = \overline{A} + r B_2$, where $\ol A$ is the closure of $A$.

Additional notation and definitions
associated with convex bodies are introduced in Section~\ref{sec:convex-bodies} and throughout the text.
We use the letters $A, B$ to refer to general (possibly nonconvex) subsets of $\R^n$ or $\ell^2$, and the letters $K, L$ for convex sets (which are often also assumed to be compact).
General subsets $A, B$ will typically be assumed nonempty, though this is not always stated explicitly for conciseness.
Finally, for two functions $f, g \geq 0$, the notation $f \lesssim g$ means that there exists a constant $C > 0$ such that $f \leq C g$ everywhere, while $f \asymp g$ means that $f \lesssim g$ and $g \lesssim f$.
This notation will only be used when the underlying constants are universal.

\section{Minimax regret,
  intrinsic volumes and Wills functional}
\label{sec:minimax-wills}

In this section, we provide additional background on sequential probability assignment (Section~\ref{sec:coding-minimax-regret}) and on %
intrinsic volumes (Section~\ref{sec:convex-bodies}).
The link between
the two is then %
established in Section~\ref{sec:regret-intrinsic}.

\subsection{Sequential probability assignment, coding, and minimax regret}
\label{sec:coding-minimax-regret}

We start by briefly describing the statement underpinning the connection between sequential probability assignment and coding.
Assume that the space $\Y$ of observations is countable and that $\mu$ is the counting measure on $\Y$.
A \emph{prefix code} on $\Y^n$ is a function $\phi : \Y^n \to \bigcup_{k \geq 1} \{ 0, 1\}^k$ such that for every pair of sequences $y, y' \in \Y^n$, the associated encoding $\phi (y)$ (a string of bits) is not a prefix of $\phi (y')$.
(The prefix property is useful to decode concatenated sequences,
and %
relaxing it to the minimal %
assumption of injectivity of $\phi$ does not lead to sizeable gains in code-length.%
)
The associated code-length function is $c_\phi (y) = |\phi (y)|$, the length of the string $\phi(y)$.
Then, Kraft's inequality (\eg, \cite[Theorem~5.2.2 p.~109]{cover2006elements}) states that a function $c : \Y^n \to \N^*$ is the code-length function %
of some prefix code %
if and only if
\begin{equation}
  \label{eq:kraft}
  \sum_{y \in \Y^n} 2^{-c(y)}
  \leq 1
  \, .
\end{equation}
Hence, if $\phi$ is a prefix code on $\Y^n$, the positive
probability density $p_\phi$ on $\Y^n$ defined by 
\begin{equation*}
  p_\phi (y)
  = 2^{- |\phi (y)|} \Big/ \Big(\sum_{y' \in \Y^n} 2^{- |\phi (y')|} \Big) 
\end{equation*}
satisfies $- \log_2 p_\phi (y) \leq | \phi (y) |$ for any sequence $y \in \Y^n$.
Conversely, if $q$ is a positive probability density on $\Y^n$, the function $c : \Y^n \to \N^*$ defined by $c (y) = \ceil{- \log_2 q (y)}$ satisfies $\sum_{y} 2^{-c (y)} \leq \sum_{y} q (y) = 1$, so by Kraft's inequality there exists a prefix code $\phi_q$ such that $|\phi_q (y)| = c (y) = \ceil{- \log_2 q (y)} \leq - \log_2 q (y) + 1$ for any $y \in \Y^n$.
Hence, up to a rounding error (which plays a negligible role for long codes, for instance if the sample size $n$ is large) and a $\log 2$ factor, probability distributions correspond to codes,
and the log-loss~\eqref{eq:log-loss} can be interpreted as the associated code-length for the sequence $y_{1:n}$.
General observation spaces $\Y$ may be approximated by countable ones using (a sequence of increasingly fine) countable partitions; for instance, the real line $\Y = \R$ may be discretized by the partitions $\{ [l/2^k, (l+1)/2^k ) : l \in \Z \}$ for $k \geq 1$.

\paragraph{Minimax redundancy and regret.}

This point now clarified, we come back to the %
formulation described in the introduction.
Let $\Y$ be a measurable observation space endowed with a base measure $\mu$,
and $\probas$ a \emph{model} on $\Y^n$, namely a set of probability densities with respect to $\mu^n$.

The two objectives we consider %
correspond to some form of relative accuracy (in prediction or compression) with respect to the model $\probas$.
In the terminology ``universal coding'', universality refers to the fact that
a good accuracy relative to \emph{all} densities $p \in \probas$ is sought.

In the statistical %
setting, the data sequence $Y_{1:n}$ is assumed to be random with density $p \in \probas$, where $p$ is unknown.
In this case, for any other density $q$, the difference in average codelength is %
\begin{equation}
  \label{eq:excess-risk-kl}
  \E \big[ \ell (q, Y_{1:n}) - \ell (p, Y_{1:n}) \big]
  = \int_{\Y^n} \log \Big( \frac{p (y)}{q (y)} \Big) p (y) \mu^n (\di y)
  = \kll{p}{q}
  \geq 0 \,,
\end{equation}
where $\kll{p}{q} = \int p \log (p/q) \di \mu^n$ is the \emph{Kullback-Leibler divergence} or \emph{relative entropy}~\cite{kullback1951information,cover2006elements}.
This additional expected codelength %
is sometimes called \emph{redundancy}. It
does not depend on the choice of dominating measure $\mu$, but only on the distributions $P = p\cdot \mu$ and $Q = q \cdot \mu$.
Recalling that the distribution $p \in \probas$ is unknown, a natural measure of the complexity of the problem in the statistical setting is the \emph{minimax redundancy}
\begin{equation}
  \label{eq:def-redundancy}
  \mred (\probas)
  = \inf_{q} \sup_{p \in \probas} \E \big[ \ell (q, Y_{1:n}) - \ell (p, Y_{1:n}) \big]
  = \inf_{q} \sup_{p \in \probas} \kll{p}{q}
  \, ,
\end{equation}
where the infimum is with respect to all densities $q$ (not necessarily belonging to $\probas$).
This quantity is intrinsic, in the sense that it only depends on the family of distributions $\{ p \cdot \mu^n : p \in \probas \}$ on $\Y^n$ and not on the dominating measure $\mu^n$.

As mentioned in the introduction, an alternative deterministic framework~\cite{shtarkov1987universal,cesabianchi2006plg} is obtained by taking the sequence $y_{1:n} \in \Y^n$ to be arbitrary, and comparing the codelength of a density $q$ to that of the best density in the model $\probas$ for any sequence.
This leads to the notion of regret~\eqref{eq:regret}, and to
the complexity measure of the problem in the deterministic setting,
the \emph{minimax regret} $\mmr (\probas)$ in~\eqref{eq:minimax-regret}, which writes:
\begin{equation}
  \label{eq:regret-ratio}
  \mmr (\probas)
  = \inf_{q} \sup_{y \in \Y^n} \log \Big( \frac{\sup_{p \in \probas} p (y)}{q (y)} \Big)
  \, .
\end{equation}
As shown by Shtarkov~\cite{shtarkov1987universal}, this quantity equals the log-integral of the upper envelope of $\probas$:
\begin{equation}
  \label{eq:regret-shtarkov}
  \mmr (\probas)
  = \log \bigg( \int_{\Y^n} \sup_{p \in \probas} p (y) \mu^n (\di y) \bigg)
  \, .
\end{equation}
The reason for this
remarkable expression is simple:
due to the constraint that $q$ is a probability density, the infimum is obtained by making the ratio in~\eqref{eq:regret-ratio} %
constant over $y \in \Y^n$.
Formally, if $C > 0$ satisfies $C > \mmr (\probas)$, then there exists a density $q$ such that $\sup_{p \in \probas} p (y) \leq e^C q (y)$ for any $y \in \Y^n$, so that $\int_{\Y^n} \sup_{p \in \probas} p(y) \mu^n (\di y) \leq e^C \int_{\Y^n} q = e^C$, which gives the lower bound in~\eqref{eq:regret-shtarkov}.
The upper bound is obtained by taking the density $q (y) = \sup_{p \in \probas} p (y) / \int_{\Y^n} \sup_{p \in \probas} p (y') \mu^n (\di y')$
(known as \emph{normalized maximum likelihood})
whenever the integral in~\eqref{eq:regret-shtarkov} is finite.

\begin{remark}[Measurability issues]
  \label{rem:measurability-regret}
  When the model $\probas$ is countable, the integral in~\eqref{eq:regret-shtarkov} is well-defined and this expression for $\mmr (\probas)$ holds (with the proof above) since the envelope $y \mapsto \sup_{p \in \probas} p (y)$ is measurable.
  In addition, the quantity $\mmr (\probas)$ does not depend on the choice of measure $\mu$ or densities $p$, but only on the class $\{ p \cdot \mu : p \in \probas \}$ of probability distributions.
  When $\probas$ is uncountable, the expression~\eqref{eq:regret-shtarkov} still holds provided that the integral is an outer integral; however, in this case $\mmr (\probas)$ depends on the %
  choice of densities $p$.
  A better behaved definition for uncountable $\probas$ is to define $\mmr (\probas)$ as the supremum of $\mmr (\mathcal{Q})$ over all countable (or even finite) subsets $\mathcal{Q} \subset \probas$, which does not depend on the choice of densities.
  In the Gaussian case we consider, it equivalently suffices to choose the usual continuous densities, as we do.
\end{remark}

It stands to reason that the deterministic setting, which deals with arbitrary data sequences, should be harder than the statistical one.
The relationship between minimax redundancy and minimax regret is apparent when writing (inverting the order of suprema in~\eqref{eq:regret-ratio} for the regret):
\begin{align*}
  \mred (\probas)
  &= \inf_{q} \sup_{p \in \probas} \bigg[ \int_{\Y^n} \log \bigg( \frac{p (y)}{q (y)} \bigg) p (y) \mu (\di y) \bigg]
    \, ; \\
  \mmr (\probas)
  &= \inf_{q} \sup_{p \in \probas} \bigg[ \sup_{y \in \Y^n} \log \bigg( \frac{p (y)}{q (y)} \bigg) \bigg]
    \, .    
\end{align*}
In words, both settings require that the predictive density $q$ be competitive to all densities $p \in \probas$ simultaneously;
the distinction is in how the difference in codelength between $q$ and $p$ %
is measured.
For the redundancy, the \emph{average} (under $p$) difference is considered, while for the regret the \emph{worst-case} (maximum) difference is considered.\,\footnote{For this reason, the redundancy is sometimes called ``average redundancy'' while the maximum regret is sometimes called ``worst-case redundancy'', though we will not use this terminology to avoid possible confusion.}
The average being
smaller than the supremum,
one has
\begin{equation}
  \label{eq:redundancy-regret}
  \mred (\probas)
  \leq \mmr (\probas)
  \, .
\end{equation}
Most of this work is concerned with the minimax regret, but minimax redundancy will also be considered in Section~\ref{sec:metr-char-minim}.

\subsection{Convex bodies and intrinsic volumes}
\label{sec:convex-bodies}

We now recall basic facts and definitions about convexity.
A \emph{convex body} is a nonempty convex and compact subset of $\R^n$ for some $n \geq 1$ (following~\cite{schneider2013convex}, we do not require that this set has nonempty interior%
).
We will need the following quantities associated to a convex body:

\begin{definition}[Intrinsic volumes]
  \label{def:intrinsic-volumes}
  Let $K \subset \R^n$ be a convex body.
  For $0 \leq j \leq n$, the $j$-th \emph{intrinsic volume} of $K$ is defined as
  \begin{equation}
    \label{eq:def-kubota}
    V_j (K)
    = \binom{n}{j} \frac{\kappa_n}{\kappa_j \kappa_{n-j}} \, \E \big[ \vol_{j} (P_{\mathsf{V}^j} K) \big]
    \, ,
  \end{equation}
  where $\kappa_j = \vol_j (B_2^j) = \pi^{j/2} / \Gamma (1+j/2)$, the expectation is with respect to the random subspace $\mathsf{V}^j$
  whose distribution is uniform (rotation-invariant) over %
  $j$-dimensional linear
  subspaces of $\R^n$, %
  $P_{\mathsf{V}^j} : \R^n \to \mathsf{V}^j$ denotes the orthogonal projection on $\mathsf{V}^j$, and $\vol_j$ refers to the $j$-dimensional volume (Haussdorff measure) on $\mathsf{V}^j$ induced by its Euclidean structure.
  In particular, $V_0 (K) = 1$.
\end{definition}

Intrinsic volumes (as well as the more general \emph{mixed volumes}) play an important role in the theory of convex bodies~\cite{schneider2013convex,burago1988geometric} and in stochastic geometry~\cite{klain1997introduction,schneider2008stochastic} through the so-called \emph{kinematic formula}.
The definition~\eqref{eq:def-kubota} is known as \emph{Kubota's formula}.

Clearly, the $j$-th intrinsic volume $V_j$ is homogeneous of degree $j$, and invariant under affine isometries of $\R^n$.
Some intrinsic volumes amount %
to familiar quantities: for $K \subset \R^n$ a convex body,
$V_n (K)$ is simply the volume $\vol_n (K)$, while $V_{n-1} (K) = \frac{1}{2} \mathrm{surface}_{n-1} (\partial K)$ is half of the surface area (that is, the $(n-1)$-dimensional Haussdorff measure) of the %
boundary $\partial K$ of $K$.
The first intrinsic volume is also particularly important:
as shown by Sudakov~\cite{sudakov1976geometric}, it corresponds up to a $\sqrt{2\pi}$ factor to the \emph{Gaussian width} %
$w (K)$, defined below:
letting $X \sim \gaussdist (0, I_n)$, one has
\begin{equation}
  \label{eq:first-iv-gaussian-width}
  V_1 (K/\sqrt{2\pi})
  = w (K)
  = \E \sup_{\theta \in K} \langle \theta, X\rangle
  \, .
\end{equation}

In fact, the Gaussian representation~\eqref{eq:first-iv-gaussian-width} extends to higher-order intrinsic volumes, as shown by Tsirel'son~\cite{tsirelson1986geometric}:
for $1 \leq j \leq n$, let $G_j$ be a random $j \times n$ matrix with \iid standard Gaussian entries,
seen as a linear map $\R^n \to \R^j$; then,
\begin{equation}
  \label{eq:intrinsic-gaussian}
  V_j \big( K / \sqrt{2\pi} \big)
  = \frac{1}{j!} \cdot \E \bigg[ \frac{\vol_j (G_j K)}{\kappa_j} \bigg]
  \, .
\end{equation}
The representation~\eqref{eq:intrinsic-gaussian} (together with properties of Gaussian vectors)
makes it clear
that intrinsic volumes do not depend on the ambient dimension, in the sense that $K'=K \times \{ 0 \} \subset \R^{n+m}$ (with $m \geq 1$) satisfies $V_j (K') = V_j (K)$ for all $j = 0, \dots, n$.
This property is what motivates (the dependence in $n$ in) the specific choice of
normalization constant in the definition~\eqref{eq:def-kubota}.

While Definition~\ref{def:intrinsic-volumes} has an intuitive geometric appeal, it does not convey what makes intrinsic volumes %
fundamental quantities associated to a convex body.
The canonical nature of these quantities is best appreciated through a striking rigidity theorem of Hadwiger.
A \emph{valuation} on convex bodies is a function $\varphi$ associating a real number to each convex body in $\R^n$ (and $0$ to the empty set), satisfying the inclusion-exclusion identity
$\varphi (K \cup L) = \varphi (K) + \varphi (L) - \varphi (K \cap L)$ for every convex bodies $K, L$ such that $K \cup L$ is also convex.
Consider the linear space of all valuations on convex bodies which are also invariant under affine isometries of $\R^n$ and continuous (with respect to the Haussdorff distance).
A priori,
there could %
be infinitely many degrees of freedom for such valuations.
However, Hadwiger's theorem (\eg, \cite[Theorem~6.4.14 p.~360]{schneider2013convex}) states that this space is in fact finite-dimensional: %
for any such valuation $\varphi$, there exist $c_0, \dots, c_n \in \R$ such that $\varphi (K) = c_0 V_0 (K) + \dots + c_n V_n (K)$ for any convex body $K \subset \R^n$.

The reason why intrinsic volumes appear in various contexts (including the present one) is Steiner's classical formula on the volume of parallel sets.
Recall that for $r \geq 0$, the \emph{parallel set} of $K$ at radius $r$ is $K + r B_2^n$, the set of points at distance at most $r$ of $K$.
Since the parallel set increases with $r$ (in the sense of inclusion), so does its volume $%
\vol_n (K + rB_2^n)$.
Steiner's formula (\eg, Equation~(4.8) p.~213 in~\cite{schneider2013convex}) states that this function of $r$ is in fact a polynomial, whose coefficients are given by the intrinsic volumes:
for every $r \geq 0$,
\begin{equation}
  \label{eq:steiner}
  \vol_n (K + r B_2^n)
  = \sum_{j=0}^n V_{n-j} (K) \kappa_j r^j
  \, , \qquad \kappa_j = \vol_j (B_2^j)
  \, .
\end{equation}

\subsection{Minimax regret, intrinsic volumes and Wills functional}
\label{sec:regret-intrinsic}

With these prerequisites at hand, we turn to our main problem, namely universal coding in the Gaussian setting.
Recall from~\eqref{eq:gaussian-model-subset} that for a nonempty subset $A \subset \R^n$, we let $\probas_A$ denote the corresponding subset of the Gaussian model, and $\mmr (A)$ its associated minimax regret.

Before proceeding further, it is worth briefly recalling the way in which the Gaussian model arises and is sometimes formulated.
Consider a regression-type setting, where each observation $y_i$ with $i=1, \dots, n$ is the \emph{outcome}, or \emph{response}, associated to some \emph{feature} $t_i \in \Ts$.
Here, $\Ts$ is some abstract set of possible features, which typically consist of several predictive variables.
The features (or \emph{design points}) $t_1, \dots, t_n$ are known, but not the associated outcomes $y_1, \dots, y_n \in \R$.
Now, given a class $\F$ of functions $\Ts \to \R$, let us model
\begin{equation*}
  Y_i = f (t_i) + \eps_i \, ,
  \qquad i= 1, \dots, n \, ,
\end{equation*}
where $f \in \F$ is some (unknown) regression function in the class, and the errors $\eps_1, \dots, \eps_n$ are assumed to be independent $\gaussdist (0, \sigma^2)$ variables.
Up to rescaling the observations $Y_i$, one may reduce to the case $\sigma = 1$; the dependence on the noise level $\sigma > 0$ is discussed in more detail in Section~\ref{sec:depend-noise-sample}.
This amounts to saying that $Y_{1:n} \sim \gaussdist (\theta, I_n)$, where $\theta = (f (t_1), \dots, f (t_n))$.
Hence, the corresponding model is simply $\probas_A$, where
\begin{equation*}
  A
  = A (\F, t_{1:n})
  = \Big\{ \big( f (t_1), \dots, f (t_n) \big) : f \in \F \Big\}
  \subset \R^n
\end{equation*}
is the set of prediction vectors, which depends on both the class $\F$ and the design points $t_1, \dots, t_n \in \Ts$.
Note that the set $A$ inherits some structural properties of the class $\F$: for instance, if $\F$ is convex, then so is $A$.
Since the problem only depends on $\F$ and $t_{1:n}$ through the domain $A \subset \R^n$, in what follows we shall only refer to this set.

Our first result expresses the minimax regret in terms of the geometry of the convex domain.

\begin{theorem}[Minimax regret and intrinsic volumes]
  \label{thm:minimax-wills-volumes}
  If $K \subset \R^n$ is a convex body, then
  \begin{equation}
    \label{eq:minimax-intrinsic-volumes}
    \mmr (K)
    = \log \bigg( \sum_{j=0}^n %
    V_j \big( K / \sqrt{2\pi} \big)
    \bigg)
    \, .
  \end{equation}
  More generally, for any nonempty subset $A \subset \R^n$, %
  \begin{equation}
    \label{eq:minimax-wills-integral}
    \mmr (A)
    = \log \bigg( \frac{1}{(2\pi)^{n/2}} \int_{\R^n} e^{-\dist^2 (x, A)/2} \di x \bigg)
    \, .
  \end{equation}
\end{theorem}

In the case of a convex body $K$, Theorem~\ref{thm:minimax-wills-volumes} %
relates a statistical and information-theoretic quantity, the \emph{minimax regret} measuring the hardness of the prediction problem associated to $K$, %
to basic geometric quantities associated to $K$, namely its \emph{intrinsic volumes}.

\begin{proof}%
  By Shtarkov's integral representation~\eqref{eq:regret-shtarkov} of the minimax regret, %
  \begin{align*}
    \mmr (A)
    &= \log \bigg( \int_{\R^n} \sup_{\theta \in A} \Big[ \frac{1}{(2\pi)^{n/2}} \, e^{-\norm{\theta - x}^2/2} \Big] \di x \bigg) %
    = \log \bigg( \frac{1}{(2\pi)^{n/2}} \int_{\R^n} e^{-\dist(x, A)^2/2} \di x \bigg)
  \end{align*}
  which corresponds to the expression~\eqref{eq:minimax-wills-integral}.
  In the case where $A = K$ is a convex body, 
  the equality between the right-hand sides of equations~\eqref{eq:minimax-wills-integral} and~\eqref{eq:minimax-intrinsic-volumes} is due to Hadwiger~\cite{hadwiger1975wills};
  since the proof is short and follows directly from Steiner's formula, we include it for the reader's convenience.
  Recall that if $(\Zs, \mu)$ is a measured space and $f : \Zs \to \R^+$ a measurable function, by Fubini's theorem
  \begin{equation*}
    \int_\Zs f \di \mu
    = \int_{\Zs} \bigg( \int_{0}^\infty \indic{f \geq t} \di t \bigg) \di \mu
    = \int_0^\infty \mu (\{ f \geq t \}) \di t
    \, .
  \end{equation*}
  Applying this to $(\Zs, \mu) = (\R^n, \vol_n)$ and $f = e^{-\dist^2 (\cdot, A)/2}$, for any nonempty closed $A \subset \R^n$,
  \begin{align}
    \int_{\R^n} e^{-\dist(x, A)^2/2} \di x
    &= \int_0^{\infty} \vol_n \Big( \Big\{ x \in \R^n : e^{-\dist^2 (x, A)/2} \geq t \Big\} \Big) \di t \nonumber \\
    &= \int_{0}^\infty \vol_n \big( \big\{ x \in \R^n : \dist (x, A) \leq r \big\} \big) r e^{-r^2/2} \di r \nonumber \\
    &= \int_{0}^\infty \vol_n (A + r B_2^n) \, r e^{-r^2/2} \di r 
      \label{eq:other-integral-steiner}
      \, ,
  \end{align}
  where the last inequality uses that $A$ is closed (so that $A \cap B (x, r')$ is compact for $x \in \R^n$ and $r' > \dist (x, A)$, allowing to extract a convergent in $A$ sub-sequence of $(\theta_n)_{n \geq 1} \in A^{\N^*}$ such that $\norm{\theta_n - x} \to \dist (x, A)$).
  Now, assume that $A = K$ is a convex body.
  Applying Steiner's formula~\eqref{eq:steiner} in~\eqref{eq:other-integral-steiner} gives:
  \begin{align}
    \frac{1}{(2\pi)^{n/2}} \int_{\R^n} e^{-\dist(x, K)^2/2} \di x
    &= \frac{1}{(2\pi)^{n/2}} \int_{0}^\infty \bigg( \sum_{j=0}^n V_{j} (K) \kappa_{n-j} r^{n-j} \bigg) \, r e^{-r^2/2} \di r \nonumber \\
    &= \sum_{j=0}^n V_j (K) (2 \pi)^{-n/2} \int_0^\infty \vol_{n-j} (r B_2^{n-j}) \, r e^{-r^2/2} \di r  \nonumber \\
    &= \sum_{j=0}^n V_j (K) (2 \pi)^{-n/2} \, (2 \pi)^{(n-j)/2}
      \label{eq:proof-wills-trick} 
    = \sum_{j=0}^n V_j (K/\sqrt{2\pi}) 
      \, ,
  \end{align}
  where the penultimate equality~\eqref{eq:proof-wills-trick} is obtained by applying~\eqref{eq:other-integral-steiner} to $A = \{ 0 \} \subset \R^{n-j}$.
\end{proof}

For a convex body $K \subset \R^n$, the sum of its intrinsic volumes
\begin{equation}
  \label{eq:def-wills-convex}
  W (K)
  = \sum_{j=0}^n V_j (K)
\end{equation}
is called the \emph{Wills functional}~\cite{wills1973gitterpunktanzahl,hadwiger1975wills}.
This quantity was introduced by Wills~\cite{wills1973gitterpunktanzahl} in the context of lattice point enumeration (that is, the study of the number of points in $K$ with integer coordinates), see the survey~\cite{gritzmann1993lattice} for more information on this topic.
Several relevant results on the Wills functional are obtained in~\cite{wills1973gitterpunktanzahl,hadwiger1975wills,mcmullen1991inequalities} and the recent work~\cite{alonso2021further}, some of which will be recalled and discussed below.
The Wills functional has found %
applications to Gaussian processes~\cite{vitale1996wills} and the study of intrinsic volumes~\cite{lotz2020concentration}.
As shown by Theorem~\ref{thm:minimax-wills-volumes}, %
it also has a natural interpretation in terms of %
universal coding.

Many results in this article beyond Theorem~\ref{thm:minimax-wills-volumes} will apply to general, not necessarily convex sets.
For this reason, it is convenient to extend the definition of the Wills functional to general subsets $A \subset \R^n$ by
\begin{equation}
  \label{eq:def-wills-general}
  W (A)
  = \int_{\R^n} e^{-\pi\, \dist^2 (x, A)} \di x
  \, .
\end{equation}
Several results in the following sections support the idea that this definition is the most appropriate extension of the Wills functional to nonconvex sets,
a first indication of this being the fact that the %
characterization~\eqref{eq:minimax-wills-integral} of the minimax regret writes:
\begin{equation}
  \label{eq:regret-wills-general}
  \mmr (A)
  = \log W (A / \sqrt{2\pi})  
\end{equation}
for any nonempty domain $A \subset \R^n$.
While individual intrinsic volumes (and Steiner's formula) are specific to convex bodies, their sum---the Wills functional---may %
be thought of as defined for general subsets $A \subset \R^n$ via~\eqref{eq:def-wills-general}, since many of its fundamental properties extend to the nonconvex case.

Before this, in the remainder of this section we comment on the case of convex bodies.
First, a fundamental property of the sequence of intrinsic volumes of a convex body $K \subset \R^n$ is its \emph{Poisson-log-concavity},
which states that for $1 \leq j \leq n - 1$,
\begin{equation}
  \label{eq:alexandrov-fenchel}
  (j+1) \, V_{j+1} (K) V_{j-1} (K)
  \leq j \, V_j (K)^2
  \, .
\end{equation}
This inequality, a consequence of the general Alexandrov-Fenchel inequality between mixed volumes~\cite[Section~7.3]{schneider2013convex};~\cite[\S 20]{burago1988geometric}, was obtained independently by Chevet~\cite{chevet1976processus} and McMullen~\cite{mcmullen1991inequalities}.

A consequence of inequality~\eqref{eq:alexandrov-fenchel} is that the sequence of intrinsic volumes exhibits strong concentration (that is, decay before and after it reaches its maximal value), as shown in~\cite{lotz2020concentration,aravinda2022concentration}.
This phenomenon suggests that in the expression~\eqref{eq:minimax-intrinsic-volumes} of the minimax regret,
one may replace the sum of intrinsic volumes by the largest one without affecting the order of magnitude of this quantity.
The following result states that this is indeed the case, and that one may additionally restrict to indices on a geometric grid.

\begin{proposition}
  \label{prop:regret-max-intrinsic}
  If $K \subset \R^n$ is a convex body, and if $w (K) = V_1 (K/\sqrt{2\pi}) \geq 2$, then
  \begin{equation}
    \label{eq:regret-max-intrinsic}
    \log \Big( \max_{1 \leq j \leq n} V_j (K/\sqrt{2\pi}) \Big) \leq
    \mmr (K)
    \leq 8 \log \Big( \max_{k \geq 0} V_{2^k} (K/\sqrt{2\pi}) \Big)
    \, .
  \end{equation}
\end{proposition}

A proof of Proposition~\ref{prop:regret-max-intrinsic} is provided in %
Appendix~\ref{sec:proof-max-intrinsic}.
This proof does not rely on the concentration results in~\cite{lotz2020concentration,aravinda2022concentration}, but instead
proceeds directly from inequality~\eqref{eq:alexandrov-fenchel}.

The restriction to $w (K) \gtrsim 1$ in Proposition~\ref{prop:regret-max-intrinsic} is necessary, since for $w (K) < 1$ the right-hand side is negative (due to inequality~\eqref{eq:alexandrov-fenchel}).
The regime $w (K) \gg 1$ is arguably the most relevant one in the context of sequential probability assignment, especially in high dimension;
in any case, for $w (K) \lesssim 1$ one simply has $\mmr (K) \asymp w (K)$ by Proposition~\ref{prop:finiteness} below.

\section{Infinite-dimensional version, finiteness and asymptotics}
\label{sec:infin-dimens-vers}

In this section, we consider the infinite-dimensional version of the Wills functional and minimax regret, and describe some of its qualitative properties such as finiteness and asymptotics.
Some first quantitative bounds are discussed along the way, but we defer
the study of sharp quantitative estimates to Section~\ref{sec:comp-metric-estimates}.

\subsection{Basic properties and infinite-dimensional extension}

We start with some basic properties of the Wills functional and minimax regret, which follow from the definition~\eqref{eq:minimax-wills-integral} and will be used repeatedly.

\begin{proposition}%
  \label{prop:properties-minimax}
  The minimax regret satisfies the following properties:
  \begin{enumerate}
  \item %
    For any subsets $A \subset B \subset \R^n$, one has $\mmr (A) \leq \mmr (B)$.
  \item %
    For any $A \subset \R^n$, one has $\mmr (\ol A) = \mmr (A)$, where $\ol A$ is the closure of $A$.
  \item %
    If $(A_k)_{k \geq 1}$ is a sequence of subsets of $\R^n$ such that $A_k \subset A_{k+1}$ for any $k \geq 1$ and if $A = \bigcup_{k \geq 1} A_k$, then $\mmr (A) = \lim_{k \to \infty} \mmr (A_k)$.
  \item If $A \subset \R^m$ and $B \subset \R^n$ \textup(so that $A \times B \subset \R^{m+n}$\textup), then $\mmr (A \times B) = \mmr (A) + \mmr (B)$.
  \item %
    If $\iota : \R^n \to \R^m$ \textup(with $m \geq n$\textup) is an affine isometry, %
    in the sense that $\iota$ is affine and $\norm{\iota (x) - \iota (y)} = \norm{x - y}$ for any $x, y \in \R^n$, then $\mmr (\iota (A)) = \mmr (A)$ for any $A\subset \R^n$.
  \end{enumerate}
\end{proposition}

\begin{proof}
  The first property comes from the fact that $\dist (\cdot, A) \geq \dist (\cdot, B)$; the second from the identity $\dist (\cdot, \ol A) = \dist (\cdot, A)$; the third from the fact that the increasing sequence of positive functions $e^{-\dist^2 (\cdot, A_k)/2}$ (for $k \geq 1$) converges towards $e^{-\dist^2 (\cdot, A)/2}$, which concludes by monotone convergence; and the fourth from the identity $\dist^2 ( (x,y), A \times B ) = \dist^2 (x, A) + \dist^2 (y, B)$ for $(x,y) \in \R^m \times \R^n$ (followed by Fubini's theorem).

  To prove the fifth property, we proceed in two steps.
  First, assume that $m = n$; %
  in this case, $\iota (x) = U x + v$ for every $x \in \R^n$, for some orthogonal matrix $U \in \O (n)$ and $v \in \R^n$.
  Then $\mmr (\iota (A)) = \mmr (A)$ comes from the fact that $\iota$ is a bijection preserving the Lebesgue measure and that $\dist (\iota (x), \iota (A)) = \dist (x, A)$ for any $x \in \R^n$.
  Consider now the general case $m \geq n$.
  Let $\wt \iota$ be an affine isometry of $\R^{m}$ such that $\wt \iota (\iota (x)) = (x, 0) \in \R^n \times \R^{m - n} \simeq \R^{m}$ for any $x \in \R^n$.
  By the case $m = n$ and the fourth property, one has
  $\mmr (\iota (A)) = \mmr (\wt\iota (\iota (A))) = \mmr (A \times \{ 0 \}) = \mmr (A) + \mmr (\{ 0 \}) = \mmr (A)$.
\end{proof}

The properties in Proposition~\ref{prop:properties-minimax} allow one to extend the definition of the minimax regret functional $\mmr$ (or equivalently of the Wills functional) to subsets of the infinite-dimensional sequence space $\ell^2$ in the usual way:
for any nonempty subset $A \subset \ell^2$, let
\begin{equation}
  \label{eq:def-infinite-dim}
  \mmr (A)
  = \sup \big\{ \mmr (B) : B \subset A \text{ finite dimensional} \big\}
  \in \R^+ \cup \{ + \infty \}
\end{equation}
where a finite-dimensional set $B$ is a set whose affine span is finite-dimensional, and for such a set the Wills functional/minimax regret is defined through an affine isometry between this span and $\R^n$ for some $n$ (by Proposition~\ref{prop:properties-minimax}, this value does not depend on the choice of the affine isometry).
It follows from Proposition~\ref{prop:properties-minimax} (specifically, the monotonicity property) that this definition coincides with the previous one for a finite-dimensional set.
Note that in the definition~\eqref{eq:def-infinite-dim}, ``finite-dimensional'' could be replaced by ``finite'' %
due to monotone convergence and stability under closure properties of $\mmr$.

The infinite-dimensional extension of the Wills functional is similar to that of another %
dimension-free quantity, namely the Gaussian width, defined for $A \subset \ell^2$ as
\begin{equation*}
  w (A)
  = \sup \big\{ w (B) : B \subset A \text{ finite} \big\}
  \, .
\end{equation*}
Likewise, intrinsic volumes of a closed convex set $K \subset \ell^2$ are defined by taking the supremum over finite-dimensional convex bodies $L \subset K$, and the identity $V_1 (K/\sqrt{2\pi}) = w(K)$ still holds. %

\subsection{Finiteness and connection with Gaussian width}
\label{sec:finiteness-link-with}

Perhaps the most basic question one may ask about the minimax regret is: for which sets is this quantity finite?
The purpose of this section is to address this question.
As we will see, finiteness of the minimax regret is a more %
interesting property in infinite dimension than in finite dimension,
since it depends on both the ``scale'' and ``dimensionality'' of the domain.

A starting point is the following inequality, due to McMullen~\cite[Theorem~2]{mcmullen1991inequalities}: for every convex body $K \subset \R^n$,
\begin{equation}
  \label{eq:wills-first-iv}
  W (K)
  \leq e^{V_1 (K)}
  \, .
\end{equation}
This inequality is a consequence of the Poisson-log-concavity property~\eqref{eq:alexandrov-fenchel}, which implies (by induction) that $V_j (K) \leq V_1 (K)^j/j!$ for every $j \geq 0$; summing over $j \geq 0$ yields~\eqref{eq:wills-first-iv}.
(An alternative proof is provided in~\cite{alonso2021further}, where it is deduced from a functional form of the Blaschke-Santaló inequality from~\cite{artstein2004santalo,klartag2005geometry}; in Section~\ref{sec:short-proof-mcmull}, we provide another short proof of this inequality, which essentially only uses the Prékopa-Leindler inequality.)
Inequality \eqref{eq:wills-first-iv} is of isoperimetric nature: it stems from the fact that, for any $n \geq 1$, among all $n$-dimensional convex bodies with fixed first intrinsic volume, the higher-order intrinsic volumes (and thus also the Wills functional~\eqref{eq:def-wills-convex}) are maximized by a multiple of the unit ball $B_2^n$---a consequence of the Alexandrov-Fenchel inequality.
Letting $n \to \infty$ and evaluating the limiting Wills functional for these rescaled $\ell^2$ balls of increasing dimension leads to the dimension-free inequality~\eqref{eq:wills-first-iv}.

While the isoperimetric inequality~\eqref{eq:wills-first-iv} is not sharp in many situations of interest, it plays an important role in the quantitative study of the minimax regret and Wills functional carried out in Section~\ref{sec:comp-metric-estimates} (as well as the qualitative study in the present section).

The Gaussian width is defined in terms of the expectation of a function of a Gaussian vector~\eqref{eq:first-iv-gaussian-width},
and intrinsic volumes admit a similar representation~\eqref{eq:intrinsic-gaussian}.
The Wills functional also admits a convenient Gaussian representation~\cite{vitale1996wills}:
if $A \subset \R^n$, letting $X \sim \gaussdist (0, I_n)$,
\begin{align}
  W (A/\sqrt{2\pi})
  &= (2\pi)^{-n/2} \int_{\R^n} e^{- \dist^2 (x, A)/2} \di x \nonumber \\
  &= (2 \pi)^{-n/2} \int_{\R^n} \exp \bigg( \frac{\norm{x}^2}{2} - \inf_{\theta \in A} \frac{\norm{\theta - x}^2}{2} \bigg) e^{-\norm{x}^2/2} \di x \nonumber \\
  &= \E \exp \bigg\{ \sup_{\theta \in A} \bigg[ \langle \theta, X\rangle - \frac{\norm{\theta}^2}{2} \bigg] \bigg\}
    \label{eq:gaussian-representation-wills}
    \, .
\end{align}

A remarkable consequence of this representation is the following inequality, due to Alonso-Gutiérrez, Hernández Cifre, and Yepes Nicolás~\cite[Theorem~1.2]{alonso2021further}: if $A \subset B (\theta, r)$ for some $\theta \in \R^n$ and $r > 0$, then
\begin{equation}
  \label{eq:reverse-isometry-wills}
  \log W (A/\sqrt{2\pi})
  \geq w (A) - \frac{r^2}{2}
  \, .
\end{equation}
(In~\cite{alonso2021further} this inequality is stated for convex bodies, but the proof only uses the integral representation~\eqref{eq:def-wills-general} and does not rely on convexity.)
Indeed, up to translating $A$ (see Proposition~\ref{prop:properties-minimax}), we may assume that $\theta = 0$ so $A \subset r B_2^n$.
Plugging this in~\eqref{eq:gaussian-representation-wills} and applying Jensen's inequality gives
\begin{equation*}
  W (A/\sqrt{2\pi})
  \geq \E \exp \bigg\{ \sup_{\theta \in A} \bigg[ \langle \theta, X\rangle - \frac{r^2}{2} \bigg] \bigg\}
  \geq \exp \bigg\{ \E \sup_{\theta \in A} \langle \theta, X\rangle - \frac{r^2}{2} \bigg\}
  = \exp \Big\{ w (A) - \frac{r^2}{2} \Big\}
  \, .
\end{equation*}

The lower bound~\eqref{eq:reverse-isometry-wills} can be seen as a partial reversal of the upper bound~\eqref{eq:wills-first-iv}.
Since one can always take $r = \diam (A)$ in~\eqref{eq:reverse-isometry-wills}, this lower bound implies that $\mmr (A)$ is of order $w (A)$ when $\diam^2 (A) \ll w (A)$.

Similarly to the upper bound~\eqref{eq:wills-first-iv}, the lower bound~\eqref{eq:reverse-isometry-wills} is typically not sharp.
It is fact negative (and thus trivial) in many regimes of interest, in particular in the ``large-scale'' (or ``low noise'', ``large sample size''%
---the meaning of these terms
being clarified in Section~\ref{sec:depend-noise-sample}%
) regime.
Nevertheless, once suitably strengthened, this lower bound will play an important role in Section~\ref{sec:comp-metric-estimates}.

We now state the main result of this section, namely a characterization of sets $A \subset \ell^2$
for which the Wills functional and minimax regret are finite.

\begin{proposition}
  \label{prop:finiteness}
  If $K \subset \ell^2$ is closed and convex, then
  \begin{equation}
    \label{eq:wills-width}
    \log (1 + w (K))
    \leq \mmr (K)
    \leq w (K)
    \, .
  \end{equation}
  In particular, $\mmr(K) < \infty$ if and only if $w (K)< \infty$.
  
  More generally, for any $A \subset \ell^2$, $\mmr (A) < \infty$ if and only if $w (A)< \infty$.
  In addition, $\mmr (A) \leq w (A)$.
  However, $\mmr (A)$ cannot be lower-bounded in terms of $w(A)$ in the nonconvex case: for any $t > 0$, there exists $A \subset \ell^2$ such that $w (A) \geq t$ but $\mmr (A) \leq 1$.
\end{proposition}

The proof of Proposition~\ref{prop:finiteness} is provided below.
The statements in the convex case follow from Theorem~\ref{thm:minimax-wills-volumes} and the results above, the main point left to prove being that for a nonconvex set $A \subset \ell^2$, if $\mmr (A) < \infty$ then $w (A) < \infty$.
The proof actually provides an explicit lower bound, but this quantitative lower bound is highly suboptimal.
Much more precise lower bounds will be obtained in Section~\ref{sec:comp-metric-estimates}.
We have nonetheless included this lower bound since it is simpler than those of Section~\ref{sec:comp-metric-estimates} and suffices to characterize finiteness.

Proposition~\ref{prop:finiteness} characterizes finiteness of the minimax regret for general sets $A \subset \ell^2$, a necessary and sufficient condition being finiteness of the Gaussian width.
(Sets $A \subset \ell^2$ such that $w(A) < \infty$ are sometimes called ``Gaussian bounded sets'' in the literature.)
However, the quantitative aspect of this qualitative characterization is rather weak: in the nonconvex case, $w (A)$ provides an upper bound on $\mmr (A)$, but no lower bound in terms of $w (A)$ holds.
In the convex case, the minimax regret can be lower-bounded in terms of the Gaussian width, but there is an exponential gap between the upper and lower bounds in~\eqref{eq:wills-width}.

On this last point, we note that both the upper and lower bounds in~\eqref{eq:wills-width} are best possible in terms of the Gaussian width for convex sets.
Indeed, for any $\theta \in \R^+$, the segment $K = [0, \sqrt{2\pi} \theta]$ satisfies $w (K) = V_1 ([0, \theta]) = \theta$ and $\mmr (K) = \log (1+ \theta)$ (since $V_j (K) = 0$ for $j \geq 2$).
On the other end, the upper bound is also optimal:
for any $\theta \geq 0$, it follows from~\eqref{eq:width-ball} that $w (\frac{\theta}{\sqrt{n}} B_2^n) \to \theta$ as $n \to \infty$, while $\diam (\frac{\theta}{\sqrt{n}} B_2^n) = \frac{\theta}{\sqrt{n}} \to 0$, so combining~\eqref{eq:wills-width} and~\eqref{eq:reverse-isometry-wills} gives
\begin{equation}
  \label{eq:wills-balls}
  \lim_{n \to \infty} \mmr \Big( \frac{\theta}{\sqrt{n}} B_2^n \Big)
  = \theta
  \, .
\end{equation}
It follows from this exponential gap that the Gaussian width does not suffice to characterize the order of magnitude of $\mmr$.
In the case of convex sets, one must consider higher-order intrinsic volumes to obtain the correct scaling of the minimax regret.
In Section~\ref{sec:comp-metric-estimates}, we will see an alternative characterization in terms of different complexity parameters for general sets.

\begin{proof}[Proof of Proposition~\ref{prop:finiteness}]
  The upper bound in~\eqref{eq:wills-width} corresponds to McMullen's inequality~\eqref{eq:wills-first-iv} (since $\mmr (K) = \log W (K/\sqrt{2\pi})$ and $w (K) = V_1 (K/\sqrt{2\pi})$), while the lower bound follows from Theorem~\ref{thm:minimax-wills-volumes} since $V_j (K/\sqrt{2\pi}) \geq 0$ for $j \geq 2$.

  If $A \subset \ell^2$, let us show that $\mmr (A) \leq w (A)$.
  By monotone convergence (Proposition~\ref{prop:properties-minimax}) and an approximation argument, one may assume that $A$ is a compact subset of $\R^n$, so that $K = \conv (A)$ is a convex body.
  By monotonicity and~\eqref{eq:wills-width}, $\mmr (A) \leq \mmr (K) \leq w (K) = w (A)$.
  In particular, if $w (A) < \infty$ then $\mmr (A) < \infty$.

  Conversely, let us show that if $A \subset \ell^2$ is such that $\mmr (A) < \infty$, then $w (A) < \infty$.
  For this, it suffices to show that $A$ is bounded, since in this case applying~\eqref{eq:reverse-isometry-wills} with $r = \diam (A)$ yields $w (A) \leq \mmr (A) + \diam^2 (A)/2 < \infty$.
  We proceed by contraposition, proving that if $A \subset \ell^2$ is unbounded then $\mmr (A) = + \infty$.

  For any $N \geq 2$, one can find elements $\theta_1, \dots, \theta_N \in A$ such that $\norm{\theta_i - \theta_j} \geq 5 \sqrt{N}$ for $i \neq j$ (start with any $\theta_1 \in A$, and given $\theta_1, \dots, \theta_{i-1}$ pick any $\theta_i \in A \setminus \bigcup_{1\leq j < i} B (\theta_j, 5 \sqrt{N})$, which exists since $A$ is unbounded).
  Let $W \subset \ell^2$ be the affine span of $A_N = \{ \theta_1, \dots, \theta_N \}$, and $n \leq N$ its dimension.
  Let $\iota : W \to \R^n$ be an affine isometry and $A'_N = \{ \vartheta_1, \dots, \vartheta_N \}$ with $\vartheta_i = \iota (\theta_i)$; by Proposition~\ref{prop:properties-minimax}, $\mmr (A) \geq \mmr (A_N) = \mmr (A'_N)$.
  Now, for $1 \leq i \leq N$, let $B_i = B (\vartheta_i, 2 \sqrt{N}) \subset \R^n$; since $\norm{\vartheta_i - \vartheta_j} = \norm{\theta_i - \theta_j} \geq 5 \sqrt{N}$ for $i \neq j$, the balls $(B_i)_{1 \leq i \leq N}$ are pairwise disjoint.
  In addition, letting $p_{\vartheta}$ be the density of the measure $\gaussdist (\vartheta, I_n)$ and $X \sim \gaussdist (0, I_n)$, one has for $i=1, \dots, n$,
  \begin{equation*}
    1 - \int_{B_i} p_{\vartheta_i}
    = \P \big( (\vartheta_i + X) \not\in B (\vartheta_i, 2 \sqrt{N}) \big)
    = \P (\norm{X} > 2 \sqrt{N})
    \leq \frac{\E \norm{X}^2}{4 N}
    = \frac{n}{4 N}
    \leq \frac{1}{4}
    \, .
  \end{equation*}
  One can therefore write:
  \begin{equation*}
    \mmr (A'_N)
    = \log \bigg( \int_{\R^n} \max_{1 \leq i \leq N} p_{\vartheta_i} \bigg)
    \geq \log \bigg( \sum_{i=1}^N \int_{B_i} p_{\vartheta_i} \bigg)
    \geq \log (3 N/4)
    \, .
  \end{equation*}
  Hence $\mmr (A) \geq \log (3 N / 4)$ for any $N \geq 2$ and thus $\mmr (A) = + \infty$.

  To see that $\mmr (A)$ cannot be lower-bounded in terms of $w (A)$, consider for any $t > 0$ the set $A = \{ 0, t \} \subset \R$.
  One has $w (A) = t/\sqrt{2\pi}$, but $\mmr (A) \leq \log 2 \leq 1$ since $A$ has at most two elements
  (see~\eqref{eq:mixture-regret}).
\end{proof}

\subsection{Small and large scale asymptotics}
\label{sec:large-scale-small}

Having characterized finiteness of the minimax regret and Wills functional, in this section we conclude the qualitative study of these functionals by discussing small and large scale asymptotics; namely, the asymptotic behavior of $\mmr (t A)$ as $t \to 0$ or $t \to \infty$.

For a convex body $K \subset \R^n$, the asymptotic behavior of $\mmr (t K)$ is readily understood from the expression~\eqref{eq:minimax-intrinsic-volumes}, which %
implies %
that
\begin{equation*}
  \mmr (t K) \sim t \, w (K) \quad \text{as } t \to 0 ;
  \qquad
  \mmr (t K)
  = n \log t + \log \vol_n (K/\sqrt{2\pi}) + o (1) \quad \text{as } t \to \infty
  \, ,
\end{equation*}
the latter being true assuming (up to restricting to a subspace) %
that the affine span of $K$ is $\R^n$, or equivalently $\vol_n (K) > 0$.
This asymptotic behavior subsists in the nonconvex case, at both small and large scales.

\begin{proposition}
  \label{prop:small-scale}
  For any $A \subset \ell^2$ such that $\mmr (A) < + \infty$, one has $\mmr (t A) \sim t \cdot w (A)$ as $t \to 0$.
  In addition, if $\diam^2 (A) \leq w (A)$ \textup(which holds if $\diam (A) \leq 1/\sqrt{2\pi}$\textup), then $\mmr (A) \geq w (A)/2$.
\end{proposition}

\begin{proof}
  Recall that $\mmr (A) < \infty$ if and only if $w (A) < \infty$, and that $\mmr (t A) \leq t w (A)$ for any $t > 0$.
  In this case, \eqref{eq:reverse-isometry-wills} gives $\mmr (t A) \geq t w (A) - t^2 \diam^2 (A)/2 \sim t w (A)$ as $t \to 0$.
  The second inequality also follows from~\eqref{eq:reverse-isometry-wills}, and the fact that $\diam (A) \leq 1/\sqrt{2\pi}$ implies $\diam^2 (A) \leq w(A)$ follows from the inequality $w (A) \geq \diam (A)/\sqrt{2\pi}$ (verified by reducing to two-point sets).
\end{proof}

We now turn to large-scale asymptotics, in finite dimension as in the convex case.
(The large-scale asymptotic behavior can be more general in the infinite-dimensional case.
We do not expand on this since Section~\ref{sec:comp-metric-estimates} deals with the non-asymptotic order of magnitude, 
which gives more %
information %
than the asymptotic behavior.)

\begin{proposition}
  \label{prop:large-scale}
  For any compact set $A \subset \R^n$, one has $\lim_{t \to \infty} W (t A) / t^n = \vol_n (A)$.
  In particular, if $\vol_n (A) > 0$, then
  $\mmr (t A) = n \log t + \log \vol_n (A/\sqrt{2\pi}) + o (1)$ as $t \to \infty$.
\end{proposition}

The proof of Proposition~\ref{prop:large-scale} is provided in %
Appendix~\ref{sec:proof-large-scale}.
The assumption that $A \subset \R^n$ is compact is not restrictive, since by Proposition~\ref{prop:finiteness} (in the finite-dimensional case) $W (A)$ is finite if and only if $A$ is bounded, and by Proposition~\ref{prop:properties-minimax} one has $W (A) = W (\ol A)$, where $\ol A$ is closed bounded in $\R^n$ and therefore compact.

\section{Comparison theorem and %
  metric characterization}
\label{sec:comp-metric-estimates}

We now come back to the quantitative study %
of the minimax regret and Wills functional.
Theorem~\ref{thm:minimax-wills-volumes} provides an exact formula in terms of intrinsic volumes in the convex case.
In this section, we obtain a characterization in terms of different geometric complexity measures, which holds in the general nonconvex case and also provides new information on intrinsic volumes.

\subsection{Comparison theorem for the Wills functional}
\label{sec:comp-theor-wills}

An important feature of the Wills functional and minimax regret, which is not immediately apparent from the definition~\eqref{eq:def-wills-general}, is that they are in fact ``metric'' quantities.
The metric character of the Wills functional means that, if a set $A \subset \ell^2$ is ``larger'' than another set $B$ in terms of distances, then $W (A) \geq W (B)$.
This is expressed %
by the following comparison principle:

\begin{theorem}[Comparison inequality]
  \label{thm:comparison-wills}
  Let $A, B$ be two nonempty subsets of $\ell^2$ such that there exists a map $\varphi : A \to B$ such that $\varphi (A) = B$ and $\norm{\varphi (\theta) - \varphi (\theta')} \leq \norm{\theta-\theta'}$ for all $\theta, \theta' \in A$.
  Then, one has $W (A) \geq W (B)$ and $\mmr (A) \geq \mmr (B)$.
\end{theorem}

Theorem~\ref{thm:comparison-wills} is an analogue of a fundamental result in the theory of Gaussian processes, namely the Sudakov-Fernique comparison inequality~\cite{sudakov1976geometric,fernique1975regularite} (which is closely related to Slepian's lemma~\cite{slepian1962one}), which states that under the assumptions of Theorem~\ref{thm:comparison-wills} one has $w (A) \geq w (B)$.
In fact, Theorem~\ref{thm:comparison-wills} generalizes the Slepian-Sudakov-Fernique comparison inequality, which is recovered as follows:
for any $t > 0$, the map $\varphi_t (x) = t \varphi (x/t)$ contracts $t A$ into $t B$, so by Theorem~\ref{thm:comparison-wills} one has $\mmr (t A) \geq \mmr (t B)$.
Now, by Proposition~\ref{prop:small-scale} as $t \to 0$ one has $\mmr (t A) \sim t \cdot w (A)$ and $\mmr (t B) \sim t \cdot w (B)$, so that $w (A) \geq w (B)$.

The proof of Theorem~\ref{thm:comparison-wills} %
relies
on classical arguments in the proof of comparison inequalities,
namely Gaussian representation, Gaussian interpolation and integration by parts, and smoothing (see~\cite{kahane1986inegalite}, \cite[\S1.3]{talagrand2011mean}, \cite[\S7.2]{vershynin2018high} and
especially~\cite{chatterjee2005error}).
Interestingly, while the Gaussian representation~\eqref{eq:gaussian-representation-wills} of the Wills functional features an additional variance correction (a second source of dependence on the covariance of the underlying Gaussian process),
the proof involves simplifications reminiscent of %
those that appear in the proof of Slepian's inequality.

\begin{proof}
  The inequalities for the Wills functional and minimax regret are equivalent by Theorem~\ref{thm:minimax-wills-volumes}, since a contraction $\varphi : A \to B$ can be rescaled into
  a contraction $A/\sqrt{2\pi} \to B / \sqrt{2\pi}$.
  
  In addition, by an approximation argument, it suffices to prove Theorem~\ref{thm:comparison-wills}
  in the case where $A, B$ are finite.
  Specifically, since $B \subset \ell^2$ which is separable, one may find a sequence $(y_n)_{n \geq 1} \in B^{\N^*}$ which is dense in $B$.
  Hence, letting $B_n = \{ y_i : 1 \leq i \leq n\}$ for $n \geq 1$, the sequence $(B_n)_{n \geq 1}$ is increasing and such that $\bigcup_{n \geq 1} B_n$ is dense in $B$, so that (by Proposition~\ref{prop:properties-minimax}) $W (B_n) \to W (B)$.
  Now since $B = \varphi (A)$, one can write $y_n = \varphi (x_n)$ for some $x_n \in A$ for all $n$, and letting $A_n = \{ x_i : 1 \leq i \leq n\}$ one has $W (A) \geq W (A_n)$ since $A_n \subset A$.
  It therefore suffices to prove that $W (A_n) \geq W (B_n)$ for all $n$, noting that
  $B_n = \varphi (A_n)$.
  Finally, since $A_n, B_n \subset \ell^2$ have at most $n$ elements,
  by linearly embedding %
  their respective linear spans in $\R^n$ and using Proposition~\ref{prop:properties-minimax}, one may also assume that $A, B$ are finite subsets of $\R^n$.
  
  We now turn to the main proof.
  Let $A, B \subset \R^n$ be finite sets as in Theorem~\ref{thm:comparison-wills}, and write $A = \{ \theta_i : 1 \leq i \leq N \}$ for some $N$.
  By the Gaussian representation~\eqref{eq:gaussian-representation-wills}, one has 
  \begin{equation*}
    W (A/\sqrt{2\pi})
    = \E \exp \bigg\{ \max_{1 \leq i \leq N} \bigg[ \langle \theta_i, Z\rangle - \frac{\norm{\theta_i}^2}{2} \bigg] \bigg\}
  \end{equation*}
  where $Z \sim \gaussdist (0, I_n)$.
  Now, $(\langle \theta_i, Z\rangle)_{1 \leq i \leq N}$ is a centered Gaussian vector, with covariance $\Sigma^A \in \R^{N \times N}$ whose $(i,j)$-entry is given by $\Sigma^A_{ij} = \E \langle \theta_i, Z\rangle \langle \theta_j, Z\rangle = \langle \theta_i, \theta_j\rangle$---the Gram matrix of the family $(\theta_i)_{1 \leq i \leq N}$.
  In addition, $\norm{\theta_i}^2 = \Sigma_{ii}^A$ for any $i \leq N$.
  It follows that
  \begin{equation*}
    W (A/\sqrt{2\pi})
    = \E \exp \Big\{ \max_{1 \leq i \leq N} \big[ X_i^A - \Sigma^A_{ii}/2 \big] \Big\}
  \end{equation*}
  where $X^A = (X^A_i)_{1 \leq i \leq N} \sim \gaussdist (0, \Sigma^A)$.
  Likewise and since $B = \varphi (A)$, letting $\Sigma^B = (\Sigma_{ij}^B)_{1 \leq i, j \leq N}$ with $\Sigma_{ij}^B = \langle \varphi (\theta_i), \varphi (\theta_j)\rangle$ and $X^B \sim \gaussdist (0, \Sigma^B)$ (independent of $X^A$), one has
  \begin{equation*}
    W (B/\sqrt{2\pi})
    = \E \exp \Big\{ \max_{1 \leq i \leq N} \big[ X_i^B - \Sigma^B_{ii}/2 \big] \Big\}
    \, .
  \end{equation*}

  For every $\beta > 0$, define a function $f_{\beta} : \R^n \times \R^{n \times n} \to \R$ by
  \begin{equation*}
    f_{\beta} (x, \Sigma)
    = \bigg\{ \sum_{i=1}^N e^{\beta [ x_i - \Sigma_{ii}/2 ]} \bigg\}^{1/\beta}
    \, .
  \end{equation*}
  Also, let $W_\beta (D/\sqrt{2\pi}) = \E f_{\beta} (X^D, D)$ for $D = A, B$.
  Since for every $\beta>0$, $x$ and $\Sigma$,
  \begin{equation*}
    \exp \Big\{ \max_{1 \leq i \leq N} \big[ x_i - \Sigma_{ii}/2 \big] \Big\}
    \leq f_{\beta} (x, \Sigma)
    \leq N^{1/\beta} \exp \Big\{ \max_{1 \leq i \leq N} \big[ x_i - \Sigma_{ii}/2 \big] \Big\}
    \, ,
  \end{equation*}
  one has $W_\beta (D/\sqrt{2\pi}) \to W (D/\sqrt{2\pi})$ as $\beta \to \infty$.
  Now, letting $\Sigma (t) = (1-t) \Sigma^A + t \Sigma^B$ and $X (t) = \sqrt{1 - t} \, X^A + \sqrt{t} \, X^B \sim \gaussdist (0, \Sigma (t))$ for $t \in [0,1]$, define for $u, v \in [0,1]$:
  \begin{equation*}
    F_{\beta} (u, v)
    = \E \big[ f_{\beta} (X (u), \Sigma (v)) \big]
    \, ,
  \end{equation*}
  so that $F_\beta (0, 0) = W_\beta (A/\sqrt{2\pi})$ and $F_\beta (1, 1) = W_\beta (B/\sqrt{2\pi})$.
  It therefore suffices to prove that
  $t \mapsto F_\beta (t, t)$ is decreasing on $[0,1]$.

  Using Lemma~\ref{lem:gaussian-interpolation} below applied to $f = f_\beta (\cdot, \Sigma (t))$, 
  \begin{equation*}
    \frac{\partial F_\beta}{\partial u} (t, t)
    = \frac{1}{2} \sum_{1\leq i, j \leq N} \big( \Sigma^B_{ij} - \Sigma^A_{ij} \big) \, \E \bigg[ \frac{\partial^2 f_{\beta}}{\partial x_i \partial x_j} (X (t), \Sigma (t)) \bigg]
    \, .
  \end{equation*}
  An elementary computation shows that, for $1 \leq i, j \leq N$, %
  \begin{align*}
    \frac{\partial^2 f_{\beta}}{\partial x_i \partial x_j} (X (t), \Sigma (t)) 
    &= \bigg\{ \sum_{i=1}^N e^{\beta [ X_i (t) - \Sigma_{ii} (t)/2 ]} \bigg\}^{\frac{1}{\beta} - 1} \beta e^{\beta [ X_i (t) - \Sigma_{ii} (t)/2 ]} \indic{i = j} + \\ & \quad + \Big( \frac{1}{\beta} - 1 \Big) \bigg\{ \sum_{i=1}^N e^{\beta [ X_i (t) - \Sigma_{ii} (t)/2 ]} \bigg\}^{\frac{1}{\beta} - 2} \beta e^{\beta [ X_i (t) - \Sigma_{ii} (t)/2 ]} e^{\beta [ X_j (t) - \Sigma_{jj} (t)/2 ]} \\
    &= \fs (t) \big\{ \beta \pr_i (t) \indic{i=j} - (\beta - 1) \pr_i (t) \pr_j (t) \big\}
    \, ,
  \end{align*}
  where we let
  \begin{equation*}
    \pr_i (t)
    = \pr_i (t; \beta)
    = \frac{e^{\beta [ X_i (t) - \Sigma_{ii} (t)/2 ]}}{\sum_{k=1}^N e^{\beta [ X_k (t) - \Sigma_{kk} (t)/2 ]}} \, ;
    \qquad
    \fs (t)
    = \fs (t; \beta)
    = f_\beta (X (t), \Sigma (t))
    \, .
  \end{equation*}
  This leads to
  \begin{equation}
    \label{eq:proof-comparison-der-u}
    \frac{\partial F_\beta}{\partial u} (t, t)
    = \frac{\beta}{2} \sum_{i=1}^N \E \big[ \fs (t) \pr_i (t) \big] (\Sigma_{ii}^B - \Sigma_{ii}^A) - \frac{\beta - 1}{2} \sum_{1 \leq i, j \leq N} \E \big[ \fs (t) \pr_i (t) \pr_j (t) \big] (\Sigma_{ij}^B - \Sigma_{ij}^A)
    \, .
  \end{equation}
  On the other hand, by differentiating the expectation defining $F_\beta$ with respect to its second argument (which is valid since $f_\beta$ and its differential grow at most exponentially at infinity),
  \begin{align}    
    \frac{\partial F_\beta}{\partial v} (t, t)
    &= \E \bigg[ \sum_{i=1}^N \frac{1}{\beta} \bigg\{ \sum_{k=1}^N e^{\beta [X_k (t) - \Sigma_{kk} (t) / 2]} \bigg\}^{\frac{1}{\beta} - 1} e^{\beta [X_i (t) - \Sigma_{ii} (t) / 2]} (- \beta) \cdot \frac{1}{2} (\Sigma_{ii}^B - \Sigma_{ii}^A) \bigg] \nonumber \\
    &= - \frac{1}{2} \sum_{i=1}^N  \E \big[ \fs (t) \pr_i (t) \big] (\Sigma_{ii}^B - \Sigma_{ii}^A)
    \label{eq:proof-comparison-der-v}
    \, .
  \end{align}
  Combine~\eqref{eq:proof-comparison-der-u} and~\eqref{eq:proof-comparison-der-v} to obtain, denoting $\Delta = \Sigma^B - \Sigma^A$,
  \begin{align}
    \frac{\di}{\di t} F_\beta (t, t)
    &= \frac{\partial F_\beta}{\partial u} (t, t) + \frac{\partial F_\beta}{\partial v} (t, t) \nonumber \\
    &= \frac{\beta - 1}{2} \bigg\{ \sum_{i=1}^N \E \big[ \fs (t) \pr_i (t) \big] \Delta_{ii} - \sum_{1 \leq i, j \leq N} \E \big[ \fs (t) \pr_i (t) \pr_j (t) \big] \Delta_{ij} \bigg\} \nonumber \\
    &= \frac{\beta - 1}{4} \sum_{1 \leq i, j \leq N} \E \big[ \fs (t) \pr_i (t) \pr_j (t) \big] \big\{ \Delta_{ii} + \Delta_{jj} - 2 \Delta_{ij} \big\}
      \label{eq:proof-derivative}
  \end{align}
  where the last step uses that $\sum_{j=1}^N \pr_j (t) = 1$.
  Now, noting that
  \begin{equation*}
    \Delta_{ii} + \Delta_{jj} - 2 \Delta_{ij}
    = \norm{\varphi (\theta_i) - \varphi (\theta_j)}^2 - \norm{\theta_i - \theta_j}^2
    \leq 0
  \end{equation*}
  by assumption on $\varphi$, 
  we deduce that $\frac{\di}{\di t} F_\beta (t, t) \leq 0$, which concludes the proof.
\end{proof}

The proof above used the following standard lemma on Gaussian interpolation (\eg, \cite[Lemma~1.3.1 p.~13]{talagrand2011mean} or \cite[Lemma~7.2.7 p.~162]{vershynin2018high}),
proved by
Gaussian integration by parts.

\begin{lemma}
  \label{lem:gaussian-interpolation}
  Let $f : \R^n \to \R$ be a twice differentiable function whose Hessian $\nabla^2 f$ has at most exponential growth at infinity.
  Let $\Sigma^A, \Sigma^B$ be positive semi-definite $n \times n$ matrices, and $X^A \sim \gaussdist (0, \Sigma^A)$, $X^B \sim \gaussdist (0, \Sigma^B)$ be independent.
  Define $X (t) = \sqrt{1-t} \, \Sigma^A + \sqrt{t} \,\Sigma^B$ for $t \in [0, 1]$.
  Then,
  \begin{equation}
    \label{eq:gaussian-interpolation}
    \frac{\di}{\di t} \E \big[ f (X (t)) \big]
    = \frac{1}{2} \sum_{1\leq i, j \leq n} \big( \Sigma^B_{ij} - \Sigma^A_{ij} \big) \, \E \bigg[ \frac{\partial^2 f}{\partial x_i \partial x_j} (X (t)) \bigg]
    \, .
  \end{equation}  
\end{lemma}

\subsection{Sharp metric characterization of the minimax regret and Wills functional}

The comparison inequality for the Wills functional and minimax regret (Theorem~\ref{thm:comparison-wills}) underlines the ``metric'' nature of these quantities.
This qualitative property suggests that these functionals can be understood quantitatively in terms of metric quantities.
A fundamental %
example of this principle comes from 
the theory of Gaussian processes (see~\cite{talagrand2021upper} for more information on this topic), 
where Slepian's lemma motivated the quantitative study of the Gaussian width of sets in terms of their metric structure~\cite{dudley1967sizes,sudakov1976geometric,fernique1975regularite,talagrand1987regularity,talagrand2021upper}.
This study culminated in the Majorizing Measure (or \emph{generic chaining}) theorem due to Talagrand and Fernique~\cite{talagrand1987regularity} (see \cite[Theorem~2.10.1 p.~59]{talagrand2021upper} and~\eqref{eq:majorizing-measure}), which provides an explicit metric characterization of the Gaussian width (up to universal constants) through a suitable multi-scale decomposition.

As discussed in Section~\ref{sec:finiteness-link-with}, the quantity $\mmr (A)$ already has some connection with the Gaussian width, in particular through the inequality $\mmr (A) \leq w (A)$.
However, this upper bound is far from being sharp in general.
In order to obtain more accurate estimates, one must therefore consider additional complexity parameters.

The first such complexity measure is the \emph{local Gaussian width}, defined for any nonempty set $A \subset \R^n$ and $r \geq 0$ by
\begin{equation}
  \label{eq:def-local-width}
  w_A (r)
  = \sup_{\theta \in A} w \big( A \cap B (\theta, r) \big)
  \, .
\end{equation}
The quantity $w_A (r)$ controls the ``local'' structure of the set $A$, at scales smaller than $r$.
Local widths have a natural interpretation in terms of Gaussian processes, since they correspond to the modulus of continuity\footnote{Specifically, the quantity $w_A (r)$ is connected to a ``local'' modulus of continuity of the process (with one deterministic endpoint), while the ``global'' modulus of continuity corresponds to
  the \emph{a priori} larger quantity $w ( (A-A) \cap r B_2 )$.
  However, a remarkable property of Gaussian processes discovered by Fernique is that %
  these two moduli of continuity are (in a suitable sense) equivalent up to constants;
  see~\cite[p.~82]{talagrand2021upper}.%
} of the Gaussian process $(\langle \theta, X\rangle)_{\theta \in A}$ associated to $A$ (where $X \sim \gaussdist (0, I_n)$).
While this is not immediate from its definition, it follows from the %
Talagrand-Fernique characterization of the (global) Gaussian width that local widths $w_A (r)$ can be related to the metric structure of $A$; see~\cite[\S2.15]{talagrand2021upper}.
As we will see, the richness of the local structure of $A$, as measured by local widths, provides a first obstruction to achieving small regret.

The second complexity measure is a ``global'' one.
For any $r \geq 0$, the \emph{covering number} $N (A, r)$ of $A$ at scale $r$ is the number of balls of radius $r$ needed to cover $A$:
\begin{equation}
  \label{eq:def-covering-numbers}
  N (A, r)
  = \inf \bigg\{ N \geq 1 : \exists \theta_1, \dots, \theta_N \in A, \ A \subset \bigcup_{i=1}^N B (\theta_i, r) \bigg\}
  \, .
\end{equation}
The parameter $\log_2 N (A, r)$ (or more accurately, its upper integral part) has a natural interpretation: it is the number of bits required to describe a
finite approximation of $A$
with error at most $r$; for this reason, it is sometimes referred to as the metric entropy of $A$.
The quantity $\log N (A, r)$ %
accounts for the global complexity of $A$, at scales larger than $r$.

One can associate suitable ``fixed points'' to both local Gaussian widths and global covering numbers.
Specifically, define for any nonempty set $A \subset \R^n$:
\begin{align}
  \label{eq:fixed-point-local}
  r_* (A)
  &= \sup \big\{ r \geq 0 : w_A (r) \geq r^2 \big\} \, ; \\
  \label{eq:fixed-point-covering}
  \wt r (A)
  &= \sup \big\{ r \geq 0 : \log N (A, r) \geq r^2 \big\}
    \, .
\end{align}
(Since $w_A (r) \leq \min (w (A), w (r B_2^n)) \leq \min (w (A), r \sqrt{n})$, one has $r_* (A) \leq \min (w (A), \sqrt{n})$; on the other hand, $\wt r (A)$ is finite if and only if $A$ is bounded, and $\wt r (A) \leq \diam (A)$ since $N (A, r) = 1$ for $r \geq \diam (A)$.)

Both quantities $r_* (A)$ and $\wt r (A)$ are classical complexity measures in high-dimensional statistics.
Specifically, fixed points associated to local complexities such as Gaussian widths are known to control the error of the least squares estimator, see~\cite{donoho1990gelfand,vandegeer1999empirical,tsirelson1982geometrical,massart2000some,koltchinskii2011oracle,bartlett2005local,lecue2013subgaussian,chatterjee2014convex} (and references therein) for more information and various extensions.
In addition, it is well-known that the fixed point~\eqref{eq:fixed-point-covering} associated to covering numbers provides an upper bound on the minimax risk for statistical estimation in various settings~\cite{yang1999information}.

The term ``fixed points'' stems from the fact that, if $g : \R^+ \to \R^+ \cup \{ + \infty \}$ is (say) non-increasing, continuous and neither identically $0$ or $+ \infty$, then $\sup \{ r \geq 0 : g (r) \geq r^2\}$ is the unique fixed point of the function $\sqrt{g}$.
This being said, some care is required in our context as none of the functions $w_A$ and $\log N (A, \cdot)$ (that define the ``fixed points'' $r_* (A)$ and $\wt r (A)$) satisfies all these properties.
The function $\log N (A, \cdot)$ is indeed non-increasing, but not continuous: first, it only takes ``quantized'' values, and more importantly its value can drop sharply between two values of $r$ of the same order.
For instance, one has $\log N (B_2^n, 1) = 0$, but $\log N (B_2^n, 1/2) \asymp n$.

As for the function $w_A$, it is not decreasing but in fact non-decreasing.
However, in the case where $A = K$ is a convex body, one can essentially reduce to the non-increasing case by considering the function $r \mapsto w_K (r) / r$ which is actually non-increasing (see Lemma~\ref{lem:local-gaussian-width}); this ensures in particular that $r_* (K)$ is the unique fixed point of $w_K^{1/2}$.
The case where $A$ is a general nonconvex set is more subtle: in this case, $w_A$ needs not be continuous (it can sharply increase at some points), and the function $r \mapsto w_A (r)/r$ may not be decreasing.
In this case, there may be several ``local solutions'' of the fixed-point ``equation'' (that is, values of $r$ such that for $t$ in a neighborhood of $r$, one has $w_A (t) \geq t^2$ if $t<r$ and $w_A (t)<t^2$ if $t > r$), of possibly different orders of magnitude.
Under such a configuration, by definition $r_* (A)$ corresponds to the supremum of these ``local fixed points''.

Some convenient properties of the fixed points~\eqref{eq:fixed-point-local} and~\eqref{eq:fixed-point-covering} are gathered in Section~\ref{sec:prop-fixed-points}.

With these definitions at hand, we are now in position to state our main quantitative result: a characterization of the minimax regret $\mmr (A) = \log W (A/\sqrt{2\pi})$ for general sets $A$, up to universal constant factors.

\begin{theorem}[Characterization of the minimax regret]
  \label{thm:minimax-regret-metric}
  For every subset $A \subset \R^n$,
  \begin{equation}
    \label{eq:minimax-regret-fixed}
    \max \bigg( \frac{r_* (A)^2}{2}, \frac{\wt r (A)^2}{300} \bigg)
    \leq \mmr (A)
    \leq 2 \max \big( r_* (A)^2, \wt r (A)^2 \big)
    \, .
  \end{equation}
  Alternatively,
  \begin{equation}
    \label{eq:minimax-regret-sum}
    \frac{1}{600} \inf_{r > 0} \Big\{ w_A (r) + \log N (A, r) \Big\}
    \leq \mmr (A)
    \leq \inf_{r > 0} \Big\{ w_A (r) + \log N (A, r) \Big\}
    \, .
  \end{equation}
\end{theorem}

It follows from Theorem~\ref{thm:minimax-regret-metric} that the minimax regret $\mmr (A)$ is sharply characterized by a combination of local Gaussian widths and global covering numbers.
As noted above, for a given $r$, the local width $w_A (r)$ controls the complexity of $A$ at scales smaller than $r$, while the global covering number $\log N (A, r)$ measures the complexity of $A$ at scales larger than $r$.
In addition, the quantities that appear in Theorem~\ref{thm:minimax-regret-metric} can be further expressed explicitly in terms of the metric structure of the set $A$ (see Corollary~\ref{cor:metric-intrinsic-volumes} below), therefore
making explicit the metric nature of $\mmr (A)$, which is implicit in the comparison inequality of Theorem~\ref{thm:comparison-wills}.

As a side remark, the equivalence (up to constants) between the quantities in the right-hand sides of~\eqref{eq:minimax-regret-fixed} and~\eqref{eq:minimax-regret-sum} is not immediate, since the definition of  $\max (r_* (A)^2, \wt r (A)^2)$ involves an additional term in $r^2$.
The ``direct'' direction is that the infimum in~\eqref{eq:minimax-regret-fixed} is controlled by $\max (r_* (A)^2, \wt r (A)^2)$ (see~Section~\ref{sec:proof-relating-upper-lower}%
),
while the reverse inequality %
follows from %
Theorem~\ref{thm:minimax-regret-metric}.

Theorem~\ref{thm:minimax-regret-metric} can be interpreted as follows.
There are at least two basic ways to upper bound $\mmr (A)$.
The first one is McMullen's inequality $\mmr (A) \leq w (A)$, of isoperimetric nature.
This inequality is established using tools from Convex Geometry, such as the Alexandrov-Fenchel inequalities between intrinsic volumes as in~\cite{mcmullen1991inequalities} (or the functional Blaschke-Santaló inequality in~\cite{alonso2021further}, or the Prékopa-Leindler inequality in Section~\ref{sec:short-proof-mcmull}).
It therefore originates %
from the convex-geometric view of $\mmr$, in connection with the Wills functional and intrinsic volumes.
The second one is the following inequality (see~\eqref{eq:mixture-regret}): if $A_1, \dots, A_N \subset \R^n$ for some $N \geq 1$, then
\begin{equation*}
  \mmr \bigg( \bigcup_{i=1}^N A_i \bigg)
  \leq \max_{1 \leq i \leq N} \mmr (A_i) + \log N
  \, .
\end{equation*}
This simple but useful inequality is connected to the statistical interpretation of $\mmr (A)$ as a minimax regret; the ``mixture''\footnote{This terminology comes from the fact that, from the perspective of sequential probability assignment, the predictive distribution that achieves this upper bound is a mixture of the minimax distributions on each set $A_i$.} argument which establishes it is the basis of the
\emph{aggregation} approach to statistical estimation~\cite{barron1987bayes,yang1999information,catoni2004statistical,yang2000mixing} and sequential prediction~\cite{desantis1988learning,vovk1998mixability,littlestone1994weighted}.
Combining these two basic inequalities in a natural way (see Section~\ref{sec:proof-upper-regret}) leads to the upper bound in~\eqref{eq:minimax-regret-sum}.
The matching lower bound in~\eqref{eq:minimax-regret-sum} from Theorem~\ref{thm:minimax-regret-metric}
implies that
this combination of the ``isoperimetric'' and ``mixture'' bounds is optimal in full generality.

Theorem~\ref{thm:minimax-regret-metric} is proved in Section~\ref{sec:proof-theorems}.
In fact, as part of the proof of the lower bound, we also obtain results in the statistical setting (defined in Section~\ref{sec:coding-minimax-regret}), which we describe in the next section.

Before this, we conclude this section by making the ``metric'' characterization of $\mmr (A)$ explicit.
To this end, we need the following definition (see~\cite[Definitions~2.7.1 and~2.7.3]{talagrand2021upper}):

\begin{definition}[Admissible sequence, $\gamma_2$ functional]
  \label{def:admissible-sequence}
  Let $A \subset \ell^2$ be a set.
  An \emph{admissible sequence} of partitions of $A$ is a sequence $(\A_j)_{j \geq 0}$ of partitions of $A$ such that, for each $j$:
  \begin{itemize}
  \item $\A_{j+1}$ is a refinement of $\A_j$ (\ie, each element of $\A_{j+1}$ is contained in an element of $\A_j$);
  \item $\A_j$ contains at most $N_j$ elements, where $N_0 = 1$ and $N_j = 2^{2^j}$ for $j \geq 1$.
  \end{itemize}
  Given an admissible sequence, for $\theta \in A$ we denote by $A_j (\theta)$ the unique element of $\A_j$ that contains $\theta$.
  In addition, for any $p \geq 0$, we define the $p$-th \emph{truncated $\gamma_2$ functional} by
  \begin{equation}
    \label{eq:truncated-gamma-2}
    \gamma_2^{(p)} (A)
    = \inf_{(\A_j)_{j \geq 0}} \sup_{\theta \in A} \sum_{j \geq p} 2^{j/2} \diam A_j (\theta)
    \, ,
  \end{equation}
  where the infimum is taken over all admissible sequences of partitions of $A$.
  (The quantity inside the infimum only depends on $(\A_j)_{j \geq p}$.)
  The (non-truncated) \emph{$\gamma_2$ functional} is $\gamma_2 = \gamma_2^{(0)}$.
\end{definition}

The (truncated) $\gamma_2$ functionals are explicitly metric, since they are defined in terms of the diameter of suitable multi-scale partitions of $A$.
In particular, it is straightforward to verify that if $\varphi : A \to \ell^2$ is a contraction, then $\gamma_2^{(p)} ( \varphi (A) ) \leq \gamma_2^{(p)} (A)$ for every $p \geq 0$.

Talagrand-Fernique's Majorizing Measure theorem~\cite[Theorem~2.10.1]{talagrand2021upper}, a fundamental result in the theory of Gaussian processes, characterizes the magnitude of the Gaussian width in terms of distances through the $\gamma_2$ functional.
It asserts that there is a universal constant $c > 1$ such that, for any set $A \subset \ell^2$,
\begin{equation}
  \label{eq:majorizing-measure}
  \frac{1}{c}\, \gamma_2 (A)
  \leq w (A)
  \leq c \, \gamma_2 (A)
  \, .
\end{equation}
The upper bound, due to Fernique, is obtained through an optimal application of the so-called chaining method; the more delicate part of this theorem is the lower bound, due to Talagrand.

Combining Theorem~\ref{thm:minimax-regret-metric} with the Majorizing Measure theorem gives an analogous characterization for the regret functional:

\begin{corollary}
  \label{cor:majorizing-measure-regret}
  There
  exists an absolute constant $C > 1$
  such that the following holds.
  For every $A \subset \ell^2$,
  \begin{equation}
    \label{eq:mm-regret-sum}
    \frac{1}{C} \inf_{p \geq 0} \big\{ (2^p - 1) + \gamma_2^{(p)} (A) \big\}
    \leq \mmr (A)
    \leq C \inf_{p \geq 0} \big\{ (2^p - 1) + \gamma_2^{(p)} (A) \big\}
    \, .
  \end{equation}
  Alternatively, let \textup(with the convention that $\sup \varnothing = - \infty$\textup)
  \begin{equation*}
    p^* (A)
    = \sup \big\{ p \geq 0 : \gamma_2^{(p)} (A) \geq 2^p \big\}
    \, .
  \end{equation*}
  Define $\gamma^* (A) = \gamma_2 (A)$ if $p^* (A) = - \infty$, and $\gamma^* (A) = 2^{p^* (A)}$ otherwise.
  Then,
  \begin{equation}
    \label{eq:mm-regret-fixed}
    \frac{1}{C} \gamma^* (A)
    \leq \mmr (A)
    \leq C \gamma^* (A)
    \, .
  \end{equation}
\end{corollary}

The proof of Corollary~\ref{cor:majorizing-measure-regret} is provided in Section~\ref{sec:proof-chain-char}.
It is worth noting that, in order to obtain a ``monotone'' metric characterization like~\eqref{eq:mm-regret-sum}, it is important that the maximum $\max (\wt r (A), r_* (A) )$ is considered, rather than merely $r_* (A)$.
Indeed, the quantity $r_* (A)$ alone does \emph{not} necessarily decrease under contractions: if $A = \{ 0, \sqrt{2\pi} \, t \}$ with $t \geq 0$, then $r_* (A)^2 = t$ if $t \leq 1/(2\pi)$ but $r_* (A)^2 = 0$ for $t > 1/(2\pi)$.

Starting with Opper and Haussler~\cite{opper1999worst}, several works obtained upper bounds on the minimax regret for sequential probability assignment through chaining arguments~\cite{opper1999worst,cesabianchi2001logarithmic,rakhlin2015sequential,bilodeau2020tight}.
While related in spirit to the upper bound
in~\eqref{eq:mm-regret-sum},
these results differ from
Corollary~\ref{cor:majorizing-measure-regret}
in that they do not come with matching lower bounds.

While abstract characterizations such as Corollary~\ref{cor:majorizing-measure-regret} are %
of significant conceptual interest, in practice the fixed points $r_* (A)$ and $\wt r (A)$ (and thus $\mmr (A)$) are often evaluated by more direct means.
We refer to Section~\ref{sec:examples} for some examples.

\subsection{Statistical setting: metric characterization of the minimax redundancy}
\label{sec:metr-char-minim}

In this section, we consider the ``statistical'' version of the problem, as defined in Section~\ref{sec:coding-minimax-regret}.
Specifically, for a set $A \subset \R^n$,
we defined the \emph{minimax redundancy} $\mred (A) = \mred (\probas_A)$ as that of corresponding subset $\probas_A$ of the Gaussian model (see~\eqref{eq:gaussian-model-subset} for the definition of $\probas_A$ and~\eqref{eq:def-redundancy} for that of the minimax redundancy).

Our main result on the minimax redundancy is a characterization of this quantity in terms of covering numbers, which is sharp (up to universal constants) for general subsets of $\R^n$.
Recall the definition~\eqref{eq:fixed-point-covering} of the fixed point $\wt r (A)$ associated to global covering numbers.

\begin{theorem}[Characterization of the minimax redundancy]
  \label{thm:minimax-redundancy}
  For any subset $A \subset \R^n$,
  \begin{equation}
    \label{eq:minimax-redundancy-fixed}
    \frac{\wt r (A)^2}{300} %
    \leq \mred (A)
    \leq 2\, \wt r (A)^2
    \, .
  \end{equation}  
  Alternatively,
  \begin{equation}
    \label{eq:minimax-redundancy-sum}
    \frac{1}{600} \, \inf_{r > 0} \big\{ \log N (A, r) + r^2 \big\}
    \leq \mred (A)
    \leq \inf_{r > 0} \big\{ \log N (A, r) + r^2 \big\}
    \, .
  \end{equation}  
\end{theorem}

The proof of Theorem~\ref{thm:minimax-redundancy} (together with that of Theorem~\ref{thm:minimax-regret-metric}) is provided in Section~\ref{sec:proof-theorems}.
We note that the upper bound in~\eqref{eq:minimax-redundancy-sum} follows from a standard argument, namely aggregation over a net; this is the approach of the classic work of Yang and Barron~\cite{yang1999information}.
The contribution in Theorem~\ref{thm:minimax-redundancy} is thus in establishing a matching lower bound.

It is instructive to compare the characterization of the minimax redundancy from Theorem~\ref{thm:minimax-redundancy} to the minimax mean-squared error of Gaussian statistical estimation.
Theorem~\ref{thm:minimax-redundancy} states that the redundancy is characterized, for a general set $A$, by \emph{global} covering numbers.
In contrast, it is shown in~\cite{neykov2022minimax} (see also~\cite{mendelson2017local} for related results in the random-design setting, and~\cite{lecam1973convergence,ibragimov1981estimation,birge1983approximation} for early work in statistics involving local entropy)
that the minimax rate of estimation over a convex body $K$ is characterized by \emph{local} covering numbers, namely $N_{\mathsf{loc}} (K, r) = \sup_{\theta \in K} N \big( K \cap B (\theta, 2 r), r \big)$.
The essential difference between the two problems is that the first one is \emph{sequential}, in that it involves making successive predictions as observations arrive one by one, while the second problem is non-sequential, in that an estimator of the true mean vector is produced with the knowledge of all observations.

Theorem~\ref{thm:minimax-redundancy} together with prior results on estimation suggest the following
principle: optimal rates for sequential estimation problems are governed by global covering numbers, while optimal rates for non-sequential estimation are governed by local covering numbers.
A very rough %
intuitive explanation for this phenomenon is that sequential prediction involves making predictions for different sample sizes (numbers of prior observations), which corresponds to summing individual errors (governed by local covering numbers) at different scales of localization; while on the other hand, global covering numbers behave roughly like a sum of local covering numbers across scales.
In fact, neither of these approximations is fully accurate in general, essentially because the maximum of a sum may be smaller than %
the sum of the maxima; thankfully, this inaccuracy equally affects both terms of the equivalence.

It is also worth comparing the minimax redundancy $\mred (A)$ (corresponding to the average case or the ``statistical'' setting), scaling as $\wt r (A)^2$, to the minimax regret (corresponding to the worst-case or ``deterministic'' setting) $\mmr (A)$ scaling as $\max (r_* (A)^2, \wt r (A)^2)$.
In cases where $r_* (A) \lesssim \wt r (A)$, the optimal errors are of the same order in both settings; while if $r_* (A) \gg \wt r (A)$ the deterministic variant of the prediction/coding problem is more difficult than the statistical one.
The example of ellipsoids (see Section~\ref{sec:case-ellipsoids}) shows that both configurations $r_* (A) \ll \wt r (A)$ and $\wt r (A) \ll r_* (A)$ are possible for different sets $A$.

\paragraph{Redundancy and noise correlation.}
While Theorem~\ref{thm:minimax-redundancy} provides a statistical interpretation (in terms of redundancy) of the %
  term $\wt r (A)^2$ appearing in the minimax regret $\mmr (A)$, the remaining term $r_* (A)^2$ admits a different interpretation, in terms of ``correlation with the noise''.

Specifically, consider the general setting of sequential probability assignment described in Sections~\ref{sec:introduction} and~\ref{sec:coding-minimax-regret}.
  Given a set $\probas$ of probability densities on $\Y^n$ with respect to $\mu^n$ and a density $q \in \probas$, we define the \emph{noise correlation of $\probas$ under $q$} as
  \begin{equation*}
    \corr (q, \probas)
    = \E_{Y \sim q} \Big[ \ell (q, Y) - \inf_{p \in \probas} \ell (p, Y) \Big]
    = \int_{\Y^n} \log \bigg( \frac{\sup_{p \in \probas} p (y)}{q (y)} \bigg) q (y) \mu^n (\di y)
    \, .
  \end{equation*}
  The interpretation of this quantity is that, even when the data sequence $Y$ is drawn from a density $q$ belonging to the model $\probas$,
  the density $p \in \probas$ that best fits the data (\ie, achieves the smallest error $\ell (p, Y)$) is generally not the true density $q$, due to the randomness of $Y$---that is, the presence of ``noise''.
  The noise correlation term $\corr (q, \probas)$ then measures the (expected) gap between the loss of $q$ and that of the best-fitting density within $\probas$.

Now, the \emph{noise correlation of $\probas$} is defined as the maximum noise correlation over the densities
of $\probas$,
namely
\begin{equation*}
  \corr (\probas)
  = \sup_{q \in \probas} \corr (q, \probas)
  \, .
\end{equation*}
An important feature of this definition is that the supremum is taken only over densities $q \in \probas$, rather than over all probability densities with respect to $\mu^n$.
Indeed, allowing $q$ to be unrestricted would lead to a maximum correlation equal to the minimax regret $\mmr (\probas)$, by taking $q \propto \sup_{p \in \probas} p$ to be the Shtarkov density (discussed following~\eqref{eq:regret-shtarkov}, whenever %
$\mmr (\probas) < \infty$).

The relevance of noise correlation
stems from the fact that it provides a natural lower bound on the minimax regret: %
\begin{equation}
  \label{eq:correlation-regret}
  \corr (\probas)
  \leq \mmr (\probas)
  \, .
\end{equation}
To see this, let $q \in \probas$ be arbitrary.
Note that for any density $q'$ with respect to $\mu^n$, equation~\eqref{eq:excess-risk-kl} shows that $\E_{Y \sim q} [ \ell(q, Y) ] \leq \E_{Y \sim q} [ \ell (q', Y) ]$, hence
\begin{equation*}
  \corr (q, \probas)
  \leq \E_{Y \sim q} \Big[ \ell (q', Y) - \inf_{p \in \probas} \ell (p, Y) \Big]
  \leq \sup_{y \in \Y^n} \Big[ \ell (q', y) - \inf_{p \in \probas} \ell (p, y) \Big]
  \, ;
\end{equation*}
taking the
supremum over $q \in \probas$ and the infimum over $q'$
in this inequality yields~\eqref{eq:correlation-regret}.

With these general prerequisites in place, we return to the Gaussian setting.
For a general subset $A \subset \R^n$, we let $\corr (A) = \corr (\probas_A)$.
The next proposition %
relates the noise correlation $\corr (A)$ to the local mean width term $r_* (A)^2$.

\begin{proposition}
  \label{prop:correlation-local}
  For any nonempty subset $A \subset \R^n$,
  \begin{equation}
    \label{eq:corr-upper-lower}
    \frac{r_*(A)^2}{2}
    \leq \corr (A)
    \leq 36 \, r_* (A/12)^2 + 65
    \, .
  \end{equation}
  In addition, if $\diam (A) \leq 1/\sqrt{2\pi}$, then $\corr (A) \leq r_* (A)^2$.
  Hence, for any convex set $K \subset \R^n$,
  \begin{equation}
    \label{eq:corr-upper-lower-convex}
    \frac{r_* (K)^2}{2}
    \leq \corr (K)
    \leq 140\, r_* (K)^2
    \, .
  \end{equation}
\end{proposition}
The proof of Proposition~\ref{prop:correlation-local} (which mainly relies on Gaussian concentration) is provided in Section~\ref{sec:proof-correlation-local}.
In fact, for the application below, we only need the lower bound, whose proof is very short; we nonetheless include the upper bound to confirm the close link between $\corr (A)$ and $r_* (A)^2$.
Putting together the upper bound in Theorem~\ref{thm:minimax-regret-metric} with the lower bounds in Theorem~\ref{thm:minimax-redundancy} and Proposition~\ref{prop:correlation-local} leads to the following
result, 
which can be seen as converse to the
two regret lower bounds in terms of redundancy~\eqref{eq:redundancy-regret} and noise correlation~\eqref{eq:correlation-regret}, respectively.

\begin{corollary}
  \label{cor:redundancy-correlation}
  For any nonempty subset $A \subset \R^n$, one has
  \begin{equation}
    \label{eq:redundancy-correlation}
    \max \big\{ \mred (A), \corr (A) \big\}
    \leq \mmr (A)
    \leq 600 \max \big\{ \mred (A), \corr (A) \big\}
    \, .
  \end{equation}
\end{corollary}

Corollary~\ref{cor:redundancy-correlation} shows that, for a general subset $\probas_A$ of the Gaussian model, the only two obstructions to achieving a small worst-case regret are the redundancy from the statistical setting, and the presence of correlations with the noise.

\begin{proof}
  By the upper bound~\eqref{eq:minimax-regret-fixed} from Theorem~\ref{thm:minimax-regret-metric}, followed by the lower bounds~\eqref{eq:minimax-redundancy-fixed} from Theorem~\ref{thm:minimax-redundancy} and~\eqref{eq:corr-upper-lower} from Proposition~\ref{prop:correlation-local}, we have
  \begin{equation*}
    \mmr (A)
    \leq 2 \max \big\{ \wt r (A)^2, r_* (A)^2 \big\}
    \leq 2 \max \big\{ 300 \, \mred (A), 2\, \corr (A) \big\}
    \, ,
  \end{equation*}
  which implies the upper bound in~\eqref{eq:redundancy-correlation}.
  The lower bound follows from~\eqref{eq:redundancy-regret} and~\eqref{eq:correlation-regret}.
\end{proof}

\subsection{Metric estimates for the intrinsic volume sequence
}
\label{sec:isom-char-wills}

At this point, we have obtained two different characterizations of the minimax regret $\mmr (K)$, when $K \subset \R^n$ is a convex body.
The first one (Theorem~\ref{thm:minimax-wills-volumes}) expresses it in terms of intrinsic volumes, while the second one (Theorem~\ref{thm:minimax-regret-metric}) characterizes it in terms of local Gaussian widths and global covering numbers.

Relating these two characterizations leads to the following result, of purely geometric nature.

\begin{corollary}
  \label{cor:metric-intrinsic-volumes}
  For any convex body $K \subset \R^n$, one has
  \begin{equation}
    \label{eq:isomorphic-intrinsic}
    \frac{1}{600} \inf_{r > 0} \Big\{ \log N (K, r) + w_K (r) \Big\}
    \leq \log \bigg( \sum_{j=0}^n V_j (K) \bigg)
    \leq \sqrt{2\pi} \inf_{r > 0} \Big\{ \log N (K, r) + w_K (r) \Big\}
    \, .
  \end{equation}
\end{corollary}

\begin{proof}
  We plug the expression~\eqref{eq:minimax-intrinsic-volumes} of $\mmr (K)$ (for the lower bound) and $\mmr (\sqrt{2\pi} K)$ (for the upper bound) in terms of intrinsic volumes into the inequalities~\eqref{eq:minimax-regret-sum}.
  The lower bound follows from the fact that $V_j (K/\sqrt{2\pi}) \leq V_j (K)$, while the upper bound uses that if $r' = \lambda r$ for $\lambda = \sqrt{2\pi} \geq 1$, since $N (\lambda K, r') = N (K, r) \leq \lambda N (K, r)$ and $w_{\lambda K} (r') = \lambda w_K (r)$,
  \begin{equation*}
    \inf_{r' > 0} \big\{ \log N (\lambda K, r') + w_{\lambda K} (r') \big\}
    \leq \lambda \inf_{r > 0} \big\{ \log N (K, r) + w_{K} (r) \big\}
    \, . \qedhere
  \end{equation*}
\end{proof}

Corollary~\ref{cor:metric-intrinsic-volumes} relates two types of quantities, both belonging to convex geometry in the broad sense but generally associated to different branches of the topic.
Intrinsic volumes are at the core the so-called Brunn-Minkowski theory~\cite{schneider2013convex} of classical convex geometry (as well as stochastic geometry~\cite{schneider2008stochastic,klain1997introduction}, through the kinematic formulae), which often involves an exact ``isometric'' viewpoint.
In contrast, quantities such as %
Gaussian widths and covering numbers play an important role in the %
asymptotic (high-dimensional) theory~\cite{artstein2015asymptotic,pisier1999volume}, %
which is often of approximate ``isomorphic'' character.

Corollary~\ref{cor:metric-intrinsic-volumes} gives an isomorphic characterization (that is, a characterization up to constant factors) of the Wills functional of a convex body.
That the logarithm is the adequate normalization of the Wills functional to approximate it up to constant factors can be seen from the fact that $W (\frac{\theta}{\sqrt{n}} B_2^n) \to e^{\sqrt{2\pi} \theta}$ for every $\theta \geq 0$ as $n \to \infty$ (see~\eqref{eq:wills-balls}).
Of course, this normalization is also suggested by
the statistical interpretation of $\mmr$ in connection with universal coding.

The characterization~\eqref{eq:isomorphic-intrinsic} provides precise quantitative information on the sequence of intrinsic volumes of a convex body.
Indeed, by replacing $K$ by a dilation $\lambda K$ with $\lambda \geq 0$, it characterizes up to constants the log-Laplace transform $\lambda \mapsto \log (\sum_{j=0}^n V_j (K) \lambda^j)$ of this sequence.

With the background of Corollary~\ref{cor:metric-intrinsic-volumes} in mind, a possible interpretation of Corollary~\ref{cor:majorizing-measure-regret}
(characterizing $\mmr (A)$ through truncated generic chaining)
is that it constitutes  a natural %
variation on the Majorizing Measure characterization of Gaussian width~\eqref{eq:majorizing-measure}, %
as we argue now.

First, a tempting interpretation of the Majorizing Measure theorem is that it
characterizes the first intrinsic volume $V_1 (K)$ of convex bodies $K$ up to constants.
The rationale for this interpretation is that for a compact set $A \subset \R^n$, the convex body $K = \conv (A)$ satisfies $w (A) = w (K) = V_1 (K)/\sqrt{2\pi}$, so one may in principle restrict the study of the Gaussian width
of general sets to that of convex bodies.
Through this lens, a natural question would be to extend this characterization to higher-order intrinsic volumes of convex bodies; that is, to estimate $V_j (K)^{1/j}$ up to constant factors for $j \geq 2$.

However, this interpretation is not fully satisfactory.
Indeed, the Majorizing Measure theorem~\eqref{eq:majorizing-measure} relates $w (A)$ to metric properties of the set $A$ itself, rather than of its convex hull $\conv (A)$.
In particular, it contains the nontrivial statement that $\gamma_2 (A) \asymp \gamma_2 (\conv (A))$ for any $A \subset \ell^2$ (which remains to be explained from a purely geometric perspective, see~\cite[pp.~64-65 and Research Problems~2.11.2--4]{talagrand2021upper}); such a statement would be lost by restricting to convex bodies.
In addition, convexity does not play any special role in the proof of the Majorizing Measure theorem, or even in the definition of the $\gamma_2$ functional which is purely metric.
Finally, the quantity $V_j (K)^{1/j}$ is highly sensitive to the exact dimension of $K$, and can be affected by ``mild'' operations such as addition of sets: indeed, one may have $V_j (K)^{1/j} = 0$, but $V_j (K \times K) > 0$ (where $K \times K = (K \times \{0\})+(\{0\}\times K)$).
This suggests that a characterization of $V_j(K)^{1/j}$ in terms of objects such as admissible sequences of partitions is unlikely to be at all possible.

In contrast, the Wills functional appears to be a more suitable and regular quantity than individual intrinsic volumes
for a metric study of higher-order intrinsic volumes.
As noted above, it naturally extends to nonconvex sets (the extension to nonconvex sets is not gratuitous, since $\mmr (\conv (A)) \neq \mmr (A)$ in general), and as will be seen in Section~\ref{sec:conc-prop-noise} it interacts well with the additive structure of the space.
From this perspective, instead of individual normalized intrinsic volumes $V_j (K)^{1/j}$ (as a function of $j$), one considers the regularized quantity $\log (\sum_{j=0}^n V_j (K) \lambda^j)$ as a function of $\lambda \in \R^+$.
Corollary~\ref{cor:majorizing-measure-regret} can therefore be understood as a suitable ``generic chaining''-type characterization applying to higher-order intrinsic volumes, %
expressed in terms of this
regularized quantity.

\section{Interplay with additive structure and noise dependence
}
\label{sec:conc-prop-noise}

The quantitative estimates from the previous section are of isomorphic nature, in that they hold up to universal constant factors.
In this section, we complement these results by studying some exact ``isometric'' inequalities satisfied by the functional $\mmr (A)$.
While versions of those inequalities with additional constant factors could be deduced from the results of Section~\ref{sec:comp-metric-estimates}, our aim is to obtain statements with optimal constants, allowing to identify some exact properties such as monotonicity or convexity.

Specifically, this section is dedicated to the study of the interaction between the Wills functional and the affine (additive) structure of $\ell^2$.
We will successively consider the behavior of $\mmr$ under convex combinations%
, sums and differences%
, and dilations %
of sets.
The last operation has a natural statistical interpretation, which is discussed in Section~\ref{sec:depend-noise-sample}.

\subsection{Convex combinations
}
\label{sec:concavity-property}

We first consider the behavior of $\mmr$ under convex combinations of sets.
The main ``concavity'' property of the Wills functional was identified by Alonso-Gutiérrez, Hernández Cifre, and Yepes Nicolás~\cite[Theorem~1.4]{alonso2021further}, who showed that for any convex bodies $K, L \subset \R^n$ and $\lambda \in (0, 1)$,
\begin{equation}
  \label{eq:brunn-minkowski-wills}
  W (\lambda K + (1-\lambda) L)
  \geq W (K)^\lambda W (L)^{1-\lambda}
  \, .
\end{equation}
The proof of this inequality only relies on the integral expression of the Wills functional (rather than its connection with intrinsic volumes), and does not require convexity.
In order to confirm this fact and for the reader's convenience, we recall the short
proof of this inequality due to~\cite{alonso2021further}, which relies on the Prékopa-Leindler inequality.

\begin{theorem}
  \label{thm:concavity}
  For every nonempty subsets $A, B \subset \R^n$ and $\lambda \in [0, 1]$, one has
  \begin{equation}
    \label{eq:concavity}
    \mmr \big( \lambda A + (1- \lambda) B \big) \geq \lambda \mmr (A) + (1-\lambda) \mmr (B)
    \, .
  \end{equation}
\end{theorem}

\begin{proof}%
  Define the functions $f, g, h : \R^n \to \R^+$ by
  $f (x) = e^{-\dist^2 (x, A)/2}$, $g (y) = e^{-\dist^2 (y, B)/2}$ and $h (z) = e^{-\dist^2 (z, \lambda A + (1-\lambda) B)/2}$.
  For any $x, y \in \R^n$, by convexity of the squared norm one has
  \begin{equation*}
    \dist^2 (\lambda x + (1-\lambda) y, \lambda A + (1-\lambda) B)
    \leq \lambda \, \dist^2 (x, A) + (1 - \lambda) \, \dist^2 (y, B)
    \, ,
  \end{equation*}
  namely $h ( \lambda x + (1-\lambda) y ) \geq f (x)^\lambda g (y)^{1-\lambda}$.
  By the Prékopa-Leindler inequality (\eg, \cite[Theorem~7.1]{gardner2002brunn}), 
  this gives
  \begin{equation*}
    \int_{\R^n} h
    \geq \bigg( \int_{\R^n} f \bigg)^{\lambda} \bigg( \int_{\R^n} g \bigg)^{1-\lambda}
    \, .
  \end{equation*}
  Taking logarithms and recalling the expression~\eqref{eq:minimax-wills-integral} of $\mmr$,
  this establishes~\eqref{eq:concavity}.
\end{proof}

An application of the concavity property
is provided %
in Appendix~\ref{sec:short-proof-mcmull}.

\subsection{Sum and difference bodies}
\label{sec:sum-diff-bodi}

We now turn to the behavior of $\mmr$ under Minkowski sums of sets.

The comparison inequality (Theorem~\ref{thm:comparison-wills}) and the concavity property (Theorem~\ref{thm:concavity}, due to \cite{alonso2021further}) are complementary properties of the Wills functional and minimax regret.
Combining them gives the following bounds for sum and difference bodies, allowing one to restrict to origin-symmetric sets (up to constants).

\begin{proposition}
  \label{prop:sum-difference-bodies}
  For every nonempty subsets $A, B \subset \ell^2$, we have that
  \begin{equation}
    \label{eq:inequality-sum-bodies}
    \mmr (A + B)
    \leq \sqrt{2} \big( \mmr (A) + \mmr (B) \big)
    \, .
  \end{equation}
  In particular, for any $A \subset \ell^2$,
  \begin{equation}
    \label{eq:inequality-difference-bodies}
    \mmr (A - A)
    \leq 2 \sqrt{2} \, \mmr (A)
    \, .
  \end{equation}
\end{proposition}

\begin{proof}
  Note that $\ell^2 \times \ell^2$ with the natural scalar product is a Hilbert space which identifies with $\ell^2$.
  In addition, the linear map $\ell^2 \times \ell^2 \to \ell^2$ given by $(x, y) \mapsto x + y$ is Lipschitz with constant $\sqrt{2}$, since for every $x, y \in \ell^2$ one has
  \begin{equation*}
    \norm{x+y}
    \leq \norm{x} + \norm{y}
    \leq \sqrt{2} \big( \norm{x}^2 + \norm{y}^2 \big)^{1/2}
    = \sqrt{2} \, \norm{(x, y)}
    \, .
  \end{equation*}
  Hence, the map $\alpha (x, y) = (x+y)/\sqrt{2}$ is $1$-Lipschitz, and by the comparison inequality (Theorem~\ref{thm:comparison-wills}) it follows that
  \begin{align*}
    \mmr (A + B)
    &= \mmr \big( \alpha ( \sqrt{2} A \times \sqrt{2} B ) \big)
    \leq \mmr \big( \sqrt{2} A \times \sqrt{2} B \big)
    = \mmr (\sqrt{2} A) + \mmr (\sqrt{2} B) \\
    &\leq \sqrt{2} \big( \mmr (A) + \mmr (B) \big)
    \, ,
  \end{align*}
  where we used the additivity of $\mmr$ over products (Proposition~\ref{prop:properties-minimax}) and the fact that $\mmr (\lambda A) \leq \lambda \mmr (A)$ for $\lambda \geq 1$, by concavity of $\mmr$ (Theorem~\ref{thm:concavity}).
\end{proof}

An interesting
aspect of Proposition~\ref{prop:sum-difference-bodies} is that it both generalizes (to the case of nonconvex $A$ and $B \neq - A$) and improves (in terms of constants) an inequality deduced from the classical Rogers-Shephard inequality~\cite{rogers1957difference}
on the volume of difference bodies.
This inequality~\cite[Theorem~10.1.4 p.~530]{schneider2013convex} states that for any convex body $K \subset \R^n$,
\begin{equation*}
  \vol_n (K - K)
  \leq \binom{2n}{n} \vol_n (K)
  \leq 4^n \vol_n (K)
  \, .
\end{equation*}
Plugging this into Kubota's formula~\eqref{eq:def-kubota} gives $V_j (K - K) \leq 4^j V_j (K)$ for $0 \leq j \leq n$,
so that
\begin{equation*}
  \mmr (K - K)
  \leq \log \bigg( \sum_{j=0}^n 4^j V_j (K/\sqrt{2\pi}) \bigg)
  = \log \bigg( \sum_{j=0}^n V_j (4 K/\sqrt{2\pi}) \bigg)
  = \mmr (4 K)
  \leq 4 \, \mmr (K)
\end{equation*}
where the last inequality uses concavity of $\mmr$.
This is similar (in the special case of convex bodies) to inequality~\eqref{eq:inequality-difference-bodies}, but with a worse constant of $4$ instead of $2 \sqrt{2}$.

However, in the spirit of this section, one may ask whether the $\sqrt{2}$ factor in Proposition~\ref{prop:sum-difference-bodies} is necessary.
As it turns out, this factor can be removed by a more careful analysis.
The corresponding sub-additivity inequality is another %
``exact'' %
property
of the functional $\mmr$, together with the comparison inequality and concavity.

\begin{theorem}
  \label{thm:sub-additivity}
  For any nonempty subsets $A, B \subset \ell^2$, one has
  \begin{equation}
    \label{eq:sub-additivity}
    \mmr (A + B)
    \leq \mmr (A) + \mmr (B)
    \, .
  \end{equation}
  In particular, $\mmr (A - A) \leq 2 \,\mmr (A)$.
\end{theorem}

The constant $2$ in the bound $\mmr (A-A) \leq 2 \,\mmr (A)$ is best possible, since $\mmr (t A) \sim t\, w (A)$ as $t \to 0$ and $w(A-A) = 2 \, w(A)$.
Theorem~\ref{thm:sub-additivity} can be equivalently stated in terms of the Wills functional, as $W (A + B) \leq W (A) W (B)$ for any $A, B \subset \ell^2$.

Unlike Proposition~\ref{prop:sum-difference-bodies}, the proof of Theorem~\ref{thm:sub-additivity} does not use either the comparison inequality %
or concavity of $\mmr$.
Instead, it proceeds by showing that $\mmr (A + B) \leq \mmr (A \times B)$, which is achieved by interpolating between the two sets. %
The proof
leverages computations from that
of Theorem~\ref{thm:comparison-wills}, but concludes by exploiting the sum structure rather than distance inequalities.

\begin{remark}[Link with concavity]
  \label{rem:link-subadd-conc}
  A concave function $f : \R^+ \to \R$ such that $f (0) = 0$ is sub-additive.
  As such, one may ask whether
  the same holds for set functionals, allowing to %
  deduce Theorem~\ref{thm:sub-additivity} from Theorem~\ref{thm:concavity}.
  This is not the case.
  Theorem~\ref{thm:sub-additivity} does imply that the function $f : \R_+^2 \to \R$ defined by $f (s, t) = \mmr (s A + t B)$ is concave, and clearly $f (0, 0) = 0$;
  however, these properties do not imply sub-additivity in the bivariate case,
  as shown by the counter-example $g : (x, y) \mapsto \min (x,y)%
  $
  (\eg, $g(1, 0) + g(0, 1) = 0 < g(1, 1) = 1$).  
  For instance,
  the Brunn-Minkowski inequality~\cite[Theorem~4.1]{gardner2002brunn} implies that the functional $A \mapsto \vol_n (A)^{1/n}$ (defined over compact sets $A\subset \R^n$)
  is both concave and \emph{super}-additive; in fact, the two are equivalent for $1$-homogeneous functionals.
  This suggests that concavity and sub-additivity are complementary (rather than similar)
  additive properties, which may also explain why the proof of Theorem~\ref{thm:sub-additivity} has more in common with that of Theorem~\ref{thm:comparison-wills} than with that of Theorem~\ref{thm:concavity}.
\end{remark}

\begin{proof}
  By a similar approximation argument as in %
  Theorem~\ref{thm:comparison-wills}, it suffices to show~\eqref{eq:sub-additivity} when $A = \{ x_i : 1 \leq i \leq N \}$ and $B = \{ y_k : 1 \leq k \leq M \}$ are finite subsets of $\R^n$ for some $M, N, n \geq 1$.

  We identify $A + B$ with the subset $(A+B) \times \{ 0\}$ of $\R^{2n}$, and denote by $A \oplus B = A \times B \subset \R^{2n}$ the orthogonal sum of $A,B$.
  By Proposition~\ref{prop:properties-minimax}, one has $\mmr (A \oplus B) = \mmr (A) + \mmr (B)$, therefore it suffices to show that $\mmr (A + B) \leq \mmr (A \oplus B)$.

  We use the interpolation and smoothing arguments and computations from the proof of Theorem~\ref{thm:comparison-wills}, replacing respectively $A,B$ therein by $A+B$ and $A\oplus B$.
  Let $U, V, W$ be three independent random vectors with distribution $\gaussdist (0, I_n)$.
  From the Gaussian representation~\eqref{eq:gaussian-representation-wills} and the same arguments as in the proof of Theorem~\ref{thm:comparison-wills}, defining for $t \in [0, 1]$ and $\beta > 0$:
  \begin{align*}
    F_\beta (t)
    &= \E \bigg[ \bigg\{ \sum_{i=1}^N \sum_{k=1}^M \exp \Big\{ \beta \big[ X_{ik} (t) - \E X_{ik} (t)^2/2 \big] \Big\} \bigg\}^{1/\beta} \bigg]
    \\
    \text{where} \quad
    X_{ik} (t)
       &= \sqrt{1-t} \, \innerp{x_i + y_k}{U} + \sqrt{t} \, \big( \innerp{x_i}{V} + \innerp{y_k}{W} \big)
         \, ,
  \end{align*}
  one has $\lim_{\beta \to \infty} F_\beta (0) = W ( (A+B)/\sqrt{2\pi} )$ and $\lim_{\beta \to \infty} F_\beta (1) = W ( (A \oplus B)/\sqrt{2\pi} )$.
  It therefore suffices to show that $F_\beta : [0, 1] \to \R$ is non-decreasing.
  In addition, defining respectively
  \begin{align*}
    \pr_{ik} (t)
    &= \pr_{ik} (t; \beta)
    = \frac{e^{\beta [ X_{ik} (t) - \E X_{ik} (t)^2/2 ]}}{\sum_{j=1}^N \sum_{l=1}^M e^{\beta [ X_{jl} (t) - \E X_{jl} (t)^2/2 ]}} \, ;
    \ \
    \fs (t)
    = \bigg\{ \sum_{j=1}^N \sum_{l=1}^M e^{\beta [ X_{jl} (t) - \E X_{jl} (t)^2/2 ]} \bigg\}^{1/\beta} ;
    \\
    \Delta_{ik,jl} &= \E X_{ik} (1) X_{jl} (1) - \E X_{ik} (0) X_{jl} (0) \, ,
  \end{align*}
  it follows from~\eqref{eq:proof-derivative} (with appropriate changes in notation) that
  \begin{equation}
    \label{eq:proof-derivative-sum}
    F_\beta' (t)
    = \frac{\beta - 1}{4} \sum_{1 \leq i, j \leq N} \sum_{1 \leq k,l \leq M} \E \big[ \fs (t) \pr_{ik} (t) \pr_{jl} (t) \big] \big\{ \Delta_{ik,ik} + \Delta_{jl,jl} - 2 \Delta_{ik,jl} \big\}
    \, .
  \end{equation}
  
  Now, one has successively:
  \begin{align}
    \E X_{ik} (t) X_{jl} (t)
    &= (1-t) \innerp{x_i + y_k}{x_j + y_l} + t \big( \innerp{x_i}{x_j} + \innerp{y_k}{y_l} \big) \,; \nonumber \\
    \Delta_{ik,jl}
    &= \big( \innerp{x_i}{x_j} + \innerp{y_k}{y_l} \big) - \innerp{x_i + y_k}{x_j + y_l} \nonumber \\
    &= - \innerp{x_i}{y_l} - \innerp{x_j}{y_k} \, ; \nonumber
  \end{align} %
  \begin{align}
    \Delta_{ik,ik} + \Delta_{jl,jl} - 2 \Delta_{ik,jl}
    &= - 2 \innerp{x_i}{y_k} - 2 \innerp{x_j}{y_l} + 2 \big( \innerp{x_i}{y_l} + \innerp{x_j}{y_k} \big) \nonumber \\
    &= - 2 \innerp{x_i - x_j}{y_k - y_l}
      \, . \label{eq:proof-diff-delta}
  \end{align}
  A consequence of~\eqref{eq:proof-diff-delta} is that swapping $k$ and $l$ turns the last quantity into its opposite.
  Plugging~\eqref{eq:proof-diff-delta} into the expression~\eqref{eq:proof-derivative-sum} and using this property gives:
  \begin{align}
    F_\beta' (t)
    &= - \frac{\beta - 1}{2} \sum_{1 \leq i, j \leq N} \sum_{1 \leq k,l \leq M} \E \big[ \fs (t) \pr_{ik} (t) \pr_{jl} (t) \big] \innerp{x_i - x_j}{y_k - y_l} \nonumber \\
    &= - \frac{\beta - 1}{4} \cdot \E \bigg[ \fs (t) \cdot \sum_{ij,kl} \big\{ \pr_{ik} (t) \pr_{jl} (t) - \pr_{il} (t) \pr_{jk} (t) \big\} \langle x_i - x_j, y_k - y_l \rangle \bigg]
      \label{eq:proof-deriv-prod}
      \, .
  \end{align}

  Now, since
  \begin{equation*}
    \pr_{ik} (t)
    = \fs (t)^{-\beta} \cdot \exp \big\{ \beta \big[ X_{ik} (t) - \E X_{ik} (t)^2/2 \big] \big\} \, ,
  \end{equation*}
  one has
  \begin{equation*}
    \pr_{ik} (t) \pr_{jl} (t)
    = \fs (t)^{-2\beta} \cdot \exp \big\{ \beta \big[ X_{ik} (t) + X_{jl} (t) - \E X_{ik} (t)^2/2 - \E X_{jl} (t)^2/2 \big] \big\} \, .
  \end{equation*}
  In addition,
  \begin{align*}
    X_{ik} (t) + X_{jl} (t)
    &= \sqrt{1-t} \, \innerp{x_i + x_j + y_k + y_l}{U} + \sqrt{t} \, \big( \innerp{x_i + x_j}{V} + \innerp{y_k + y_l}{W} \big) \\
    &= X_{il} (t) + X_{jk} (t) \\
    \E \big[ X_{ik} (t)^2 + X_{jl} (t)^2 \big]
    &= (1-t) \big( \norm{x_i + y_k}^2 + \norm{x_j + y_l}^2 \big) + t \big( \norm{x_i}^2 + \norm{x_j}^2 + \norm{y_k}^2 + \norm{y_l}^2 \big) \\
    &= \big( \norm{x_i}^2 + \norm{x_j}^2 + \norm{y_k}^2 + \norm{y_l}^2 \big) + 2 (1-t) \big( \innerp{x_i}{y_k} + \innerp{x_j}{y_l} \big) %
      \, .
  \end{align*}
  As a result, defining
  \begin{equation*}
    \fs_{ijkl} (t) = \fs(t)^{1-2\beta} \exp \big\{ \beta \big[ X_{ik} (t) + X_{jl} (t) - (\norm{x_i}^2 + \norm{x_j}^2 + \norm{y_k}^2 + \norm{y_l}^2)/2 \big] \big\}
    = \fs_{ijlk} (t) \geq 0 \, ,
  \end{equation*}
  one has 
  \begin{equation*}
    \fs (t) \pr_{ik} (t) \pr_{jl} (t)
    = \fs_{ijkl} (t) \cdot \exp \big\{ - {\beta} (1-t) \big( \innerp{x_i}{y_k} + \innerp{x_j}{y_l} \big) \big\} \, .
  \end{equation*}
  Plugging this expression into~\eqref{eq:proof-deriv-prod} and using that $\fs_{ijkl} (t) = \fs_{ijlk} (t)$ yields
  \begin{align*}
    F_\beta' (t)
    = - \frac{\beta - 1}{4} \sum_{ij,kl} \E [\fs_{ijkl} (t)] \big\{ e^{- \beta (1-t) (\innerp{x_i}{y_k} + \innerp{x_j}{y_l})} - e^{- \beta (1-t) (\innerp{x_i}{y_l} + \innerp{x_j}{y_k})} \big\} \langle x_i - x_j, y_k - y_l \rangle
    .
  \end{align*}
  Now for $\beta > 0$ and $t \in (0,1)$, since $\innerp{x_i - x_j}{y_k - y_l} = (\innerp{x_i}{y_k}+\innerp{x_j}{y_l}) - (\innerp{x_i}{y_l} + \innerp{x_j}{y_k})$, one has $\innerp{x_i - x_j}{y_k - y_l} \geq 0$ if and only if $\innerp{x_i}{y_k}+\innerp{x_j}{y_l} \geq \innerp{x_i}{y_l} + \innerp{x_j}{y_k}$, if and only if $e^{- \beta (1-t) (\innerp{x_i}{y_k} + \innerp{x_j}{y_l})} \leq e^{- \beta (1-t) (\innerp{x_i}{y_l} + \innerp{x_j}{y_k})}$.
  Therefore
  \begin{equation*}
    \big\{ e^{- \beta (1-t) (\innerp{x_i}{y_k} + \innerp{x_j}{y_l})} - e^{- \beta (1-t) (\innerp{x_i}{y_l} + \innerp{x_j}{y_k})} \big\} \langle x_i - x_j, y_k - y_l \rangle
    \leq 0
  \end{equation*}
  and thus $F_\beta'(t) \geq 0$, which concludes the proof.
\end{proof}

\subsection{Dilations and dependence on noise and sample size}
\label{sec:depend-noise-sample}

We conclude this section by considering dilations of sets.
This operation has a natural statistical interpretation, in terms of varying the ``noise level'' of ``sample size'', as described below.

Previously, we have considered the class of Gaussians with unit variance, namely distributions $\gaussdist (\theta, I_n)$ with $\theta \in \R^n$.
One may also consider Gaussians with general variance $\sigma^2 > 0$.
Specifically, for $\theta \in \R^n$, let $p_{\theta, \sigma^2} (x) = (2\pi \sigma^2)^{-n/2} \exp (- \norm{x - \theta}^2/(2\sigma^2))$ %
be the density of the distribution $\gaussdist (\theta, \sigma^2I_n)$.
Likewise, for a subset $A \subset \R^n$, denote $\probas_{A, \sigma^2} = \{ p_{\theta, \sigma^2} : \theta \in A \}$ the associated subset of the Gaussian model with variance $\sigma^2$, and define
\begin{equation}
  \label{eq:def-regret-variance}
  \mmr (A, \sigma^2)
  = \mmr (\probas_{A, \sigma^2})
\end{equation}
the associated minimax regret (with respect to $\probas_{A, \sigma^2}$) for sequential probability assignment.
Now, since rescaling by $\sigma^{-1}$ changes the model $\probas_{A, \sigma^2}$ into $\probas_{A/\sigma,1} = \probas_{A/\sigma}$, we have that
\begin{equation}
  \label{eq:regret-variance-scale}
  \mmr (A, \sigma^2)
  = \mmr (A/\sigma)
  \, .
\end{equation}
Hence, dilations of $A$ can be understood as varying the \emph{noise level} $\sigma^2$.

Another interpretation comes from the ``{repeated}'' variant of the coding problem.
Here, instead of predicting or coding a single vector $x = (x_1, \dots, x_d) \in \R^d$, one aims to predict a sequence of vectors $x^{(1)}, \dots, x^{(n)} \in \R^d$ for some ``{sample size}'' $n \geq 1$.
In addition, given a set $A \subset \R^d$ (which induces the model $\probas_{A}$ on $\R^d$ for a single vector), the model on sequences of vectors is obtained by modeling the $n$ vectors as \iid random vectors distributed as $p_\theta$ for some $\theta \in A$.
This amounts to sequential probability assignment over $(\R^d)^n \simeq \R^{n d}$, with model
\begin{equation*}
  \big\{ p_{\theta}^{\otimes n} : \theta \in A \big\}
  = \probas_{A_n}
\end{equation*}
where $A_n = \{ (\theta, \dots, \theta) : \theta \in A \} \subset \R^{nd}$.
Denote $\mmr_n (A) = \mmr (A_n)$ the associated minimax regret. %
Note that $A_n = \pi_n (\sqrt{n} A)$ where $\pi_n : \R^d \to \R^{nd}$ is the linear isometry defined by $\pi_n (\theta) = (\theta/\sqrt{n}, \dots, \theta/\sqrt{n})$,
hence by Proposition~\ref{prop:properties-minimax} one has
\begin{equation}
  \label{eq:regret-repeated}
  \mmr_n (A)
  = \mmr (A_n)
  = \mmr (\sqrt{n} A)
  \, .
\end{equation}
In particular, $\mmr_n (A) = \mmr (A, \sigma^2)$ with $\sigma = 1/\sqrt{n}$.

The quantities $\mmr (A, \sigma^2)$ and $\mmr_n (A)$ are normalized by total code-length (log-likelihood), but alternative normalizations are also meaningful. %
Since $\ell (p_{\theta, \sigma^2}, x) = \norm{\theta - x}^2/(2\sigma^2)$ for $\theta, x \in \R^n$,
multiplying log-loss by $\sigma^2$ makes it proportional to the standard squared error
for densities in $\probas_{A, \sigma^2}$,
suggesting the alternative normalization $\sigma^2 \cdot \mmr (A, \sigma^2)$.
Likewise, $\mmr_n (A)$ measures regret for the sum of code-lengths over the vectors $x^{(1)}, \dots, x^{(n)}$, another natural quantity being the average (per-vector) regret $\mmr_n (A)/n$---%
equal to $\sigma^2 \mmr (A, \sigma^2)$ for $\sigma = 1/\sqrt{n}$.

It is natural to expect that the total regret $\mmr_n (A)$ increases with $n$, namely that $\mmr (A, \sigma^2)$ decreases with $\sigma$ (as $\sigma^2$ increases, distributions $p_{\theta, \sigma^2}$ with $\theta \in A$ become ``closer'' to each other in an information-theoretic sense, so the model $\probas_{A, \sigma^2}$ becomes smaller).
Conversely, one may expect that the average regret $\mmr_n (A) / n$ decreases with $n$, namely that $\sigma^2 \cdot \mmr (A, \sigma^2)$ increases with $\sigma$.
The following proposition confirms that both of these monotonicity properties hold, and shows that the second one can in fact be strengthened.

\begin{proposition}
  \label{prop:dilations-monotone}
  For any subset $A \subset \ell^2$, the function $t \mapsto \mmr (t A)$ is increasing on $\R_+^*$, while the function $t \mapsto \mmr (t A) / t$ is decreasing.

  In particular, $\mmr_n (A)$ increases with $n$ \textup(resp.~$\mmr (A, \sigma^2)$ decreases with $\sigma$\textup), while $\mmr_n (A)/\sqrt{n}$ decreases with $n$ \textup(resp.~$\sigma \cdot \mmr(A, \sigma^2)$ increases with $\sigma$\textup).
\end{proposition}

\begin{proof}
  In the case where $A = K$ is a convex body, the fact that $\mmr (t K)$ increases with $t$ follows directly from either the expression in terms of intrinsic volumes, or from monotonicity with respect to inclusion (Proposition~\ref{prop:properties-minimax}) as one can assume that $0 \in K$ up to translating.
  However, both of these arguments are specific to the convex case.

  In the general case, the fact that $\mmr (t A)$ increases with $t$ follows from the comparison inequality (Theorem~\ref{thm:comparison-wills}), since for $0 < s \leq t$ the contraction $x \mapsto (s/t) x$ maps $t A$ onto $s A$.

  We now turn to the second claim.
  By concavity of $\mmr$ (Theorem~\ref{thm:concavity}) applied to dilations of $A$, the map $t \mapsto \mmr (t A)$ is concave on $\R^+$, and equals $0$ at $0$.
  As a result, the slope $\mmr (t A) / t %
  $ decreases with $t$.
  Alternatively, by the sub-additivity property (Theorem~\ref{thm:sub-additivity}), the map $t \mapsto \mmr (t A)$ is sub-additive, and since it is also increasing it is continuous, so $\mmr (t A) / t$ decreases.
\end{proof}

\section{Comparison with classical asymptotics of coding
}
\label{sec:non-asympt-param}

In this section, we
relate and contrast our results to
classical asymptotic results on universal coding in information theory.

\paragraph{Classical asymptotics.}

A seminal result of Rissanen~\cite{rissanen1996fisher} provides a precise asymptotic expansion of the minimax regret
in the fixed-dimensional regime.
It implies that for a fixed compact and sufficiently smooth subset $A \subset \R^d$ with non-empty interior, as $n \to \infty$,
\begin{equation}
  \label{eq:rissanen}
  \mmr_n (A)
  = \frac{d}{2} \log \Big( \frac{n}{2\pi} \Big) + \log \vol_d (A) + o (1)
  \, .
\end{equation}
In fact, the asymptotic expansion~\eqref{eq:rissanen} extends to more general families of distributions smoothly parameterized by a compact subset of $\R^d$ or manifold, where the ``volume'' in~\eqref{eq:rissanen} is now defined with respect to the volume form associated to the Fisher information metric (see~\cite{rissanen1996fisher} and~\cite[Chapter~7]{grunwald2007mdl}).
This general result---which is based on
a quadratic approximation of the log-likelihood,
\ie a Laplace approximation of Shtarkov's integral~\eqref{eq:regret-shtarkov}---is
an instance of the ``local asymptotic normality'' phenomenon in asymptotic parametric statistics~\cite{lecam2000asymptotics}.
Within the Gaussian setting, Proposition~\ref{prop:large-scale} (together with the identity $\mmr_n (A) = \mmr (\sqrt{n} A)$) shows that the expansion~\eqref{eq:rissanen} holds for any compact set $A \subset \R^d$ with $\vol_d (A) > 0$, without additional regularity assumptions.

Rissanen's expansion had a deep influence on the theory of universal coding, data compression and model selection.
It plays a key role within the ``Minimum Description Length'' (MDL) paradigm to coding and estimation, proposed by Rissanen and expounded %
in~\cite{rissanen1985mdl,barron1998minimum,grunwald2007mdl}.
In particular, the quantity~\eqref{eq:rissanen}, sometimes called ``stochastic complexity'', has been suggested as a notion of model complexity and used as a complexity penalty in model selection.

An interesting feature of the expansion~\eqref{eq:rissanen} is that it features the volume of the parameter set $A \subset \R^d$, albeit as a constant second-order term dominated by the $\frac{d}{2} \log n$ term.
Hence, in the asymptotic regime, the minimax regret $\mmr_n (A) \sim \frac{d}{2} \log n$ is essentially determined by the dimension, with a lower-order term depending on the volume.

A limitation of this result is that it is purely asymptotic, in that it only holds for fixed $A$ (and thus $d$) as $n \to \infty$.
This rules out the modern ``high-dimensional'' regime, where the dimension may be comparable to, or larger than, the sample size $n$; as such, it does not describe what happens for a general sample size $n$.
In addition, it does not specify how large $n$ should be for the expansion~\eqref{eq:rissanen} to be accurate.

\paragraph{Comparison with estimates for convex bodies.}

When $A = K \subset \R^d$ is a convex body and for a general sample size $n$, Proposition~\ref{prop:regret-max-intrinsic} implies that, whenever $n \geq 4/w (K)^2$ (arguably the statistically meaningful regime, since $w (K) \gtrsim 1$ when the model is ``non-trivial''),
\begin{equation}
  \label{eq:refined-rissanen}
  \mmr_n (K)
  \asymp \max_{1 \leq j \leq d} \log V_j \Big( \sqrt{\frac{n}{2\pi}} K \Big)
  = \max_{1 \leq j \leq d} \Big\{ \frac{j}{2} \log \Big( \frac{n}{2\pi} \Big) + \log V_j (K) \Big\}
  \, .
\end{equation}
This expression is reminiscent of~\eqref{eq:rissanen}, but with the \emph{volume} replaced by the maximum \emph{intrinsic volume}.
The estimate~\eqref{eq:refined-rissanen} is also non-asymptotic, namely it holds for a general sample size $n$.

This result addresses the two aforementioned limitations of
the asymptotic expansion.
First, it allows one to quantify for which sample sizes $n$ the asymptotic expansion is accurate: it suffices that the $d^{\text{th}}$ intrinsic volume in the maximum of~\eqref{eq:refined-rissanen} dominates, which amounts (by the unimodality properties
of the intrinsic volume sequence,
see Section~\ref{sec:proof-max-intrinsic})
to $V_d (\sqrt{n/2\pi} K) \geq V_{d-1} (\sqrt{n/2\pi} K)$.
By homogeneity of intrinsic volumes and the link with volume and surface, assuming that $\vol_{d} (K) > 0$ this is equivalent to 
\begin{equation}
  \label{eq:sample-size-asymptotic}
  n \geq 8 \pi \cdot \Big( \frac{\mathrm{surface}_{d-1} (\partial K)}{\vol_d (K)} \Big)^2
  \, .
\end{equation}

Second, the estimate~\eqref{eq:refined-rissanen} shows that, for a general sample size $n$ that does not satisfy the %
``large sample'' condition~\eqref{eq:sample-size-asymptotic}, ``stochastic complexity'' is characterized by a lower-dimensional intrinsic volume rather than the volume.

\section{Examples}
\label{sec:examples}

In this section, we leverage the general results of Section~\ref{sec:comp-metric-estimates} to derive explicit characterizations of the minimax redundancy and regret in
some natural examples.

\subsection{Ellipsoids
}
\label{sec:case-ellipsoids}

We now consider the case of general ellipsoids in $\R^n$.
Ellipsoids arise naturally in the regression context described at the beginning of Section~\ref{sec:regret-intrinsic}, when the underlying class of functions $\F$ is the ball of a reproducing kernel Hilbert space~\cite{aronszajn1950theory,steinwart2008svm}. %
By invariance under affine isometries, one may restrict to non-degenerate axis-aligned ellipsoids, of the following form.
For $a = (a_1, \dots, a_n)$ with $a_1 \geq \dots \geq a_n > 0$, define
the ellipsoid $E_a \subset \R^n$ %
by
\begin{equation}
  \label{eq:def-ellipsoid}
  E_a
  = \bigg\{ (\theta_1, \dots, \theta_n) \in \R^n : \sum_{i=1}^n \frac{\theta_i^2}{a_i^2} \leq 1 \bigg\}
  \, .
\end{equation}

First, recall the standard Kahane-Khintchine estimate on the Gaussian width of ellipsoids; it also follows (with different constants) from the concentration property of
the Gaussian measure.
\begin{lemma}
  \label{lem:width-ellipsoid}
  For any $a \in \R_+^n$, one has
  \begin{equation}
    \label{eq:gaussian-width-ellipsoid}
    \frac{1}{\sqrt{3}} \bigg( \sum_{i=1}^n a_i^2 \bigg)^{1/2}
    \leq w (E_a)
    \leq \bigg( \sum_{i=1}^n a_i^2 \bigg)^{1/2}
    \, .
  \end{equation}
\end{lemma}

\begin{proof}
  Let $X = (X_1, \dots, X_n)
  \sim \gaussdist (0, I_n)$. %
  Letting $Z = \big( \sum_{i=1}^n a_i^2 X_i^2 \big)^{1/2}$
  and $\norm{Z}_p = \E [Z^p]^{1/p}$,
  \begin{equation*}
    w (E_a)
    = \E \sup_{\theta \in E_a} \langle \theta, X\rangle
    = \E \bigg[ \bigg( \sum_{i=1}^n a_i^2 X_i^2 \bigg)^{1/2} \bigg]
    = \norm{Z}_1
    \leq \norm{Z}_2
    = \bigg( \sum_{i=1}^n a_i^2 \bigg)^{1/2}
    \, .
  \end{equation*}
  On the other hand, since $\E X_i^2 X_j^2 \leq 3$ for $1 \leq i,j\leq n$ (with equality if $i=j$), we have
  \begin{equation*}
    \norm{Z}_4
    = \E \bigg[ \bigg( \sum_{i=1}^n a_i^2 X_i^2 \bigg)^2 \bigg]^{1/4}
    \leq \bigg[ 3 \sum_{1 \leq i,j \leq n} a_i^2 a_j^2 \bigg]^{1/4}
    = 3^{1/4} \bigg( \sum_{i=1}^n a_i^2 \bigg)^{1/2}
    \, ,
  \end{equation*}
  namely $\norm{Z}_4 \leq 3^{1/4} \norm{Z}_2$.
  Finally, Hölder's inequality implies that $\norm{Z}_2 \leq \norm{Z}_1^{1/3} \norm{Z}_4^{2/3}$, so
  that $\norm{Z}_1/\norm{Z}_2 \geq \big( \norm{Z}_2 / \norm{Z_4} \big)^2 \geq 1/\sqrt{3}$.
\end{proof}

In the case of $E_a = B_2^n$ (namely $a_1 = \dots = a_n = 1$), the proof above with the exact value $\E Z^4 = n^2 + 2 n$
gives
\begin{equation}
  \label{eq:width-ball}
  \Big( 1 + \frac{2}{\sqrt{n}} \Big)^{-1/2} \sqrt{n}
  \leq w (B_2^n)
  \leq \sqrt{n}
  \, .
\end{equation}

The following estimate for local widths of ellipsoids is standard, see \eg \cite{mendelson2003performance} for related results on local Rademacher averages of ellipsoids.

\begin{lemma}
  \label{lem:local-widths-ellipsoid}
  For every $r > 0$,
  \begin{equation}
    \label{eq:local-widths-ellipsoid}
    \frac{1}{\sqrt{3}} \bigg( \sum_{i=1}^n \min (a_i^2, r^2) \bigg)^{1/2}
    \leq w_{E_a} (r)
    \leq \sqrt{2} \, \bigg( \sum_{i=1}^n \min (a_i^2, r^2) \bigg)^{1/2}
    \, .
  \end{equation}
\end{lemma}

\begin{proof}
  Since $E_a$ is convex and symmetric, by Lemma~\ref{lem:local-gaussian-width} one has $w_{E_a} (r) = w (E_a \cap r B_2)$.
  Let $L \subset \R^n$ denote the centered axis-aligned ellipsoid with half-lengths $\min (a_i, r)$, $1 \leq i \leq n$.
  From the inequalities, for every $\theta \in \R^n$,
  \begin{align*}
    \max \bigg( \sum_{i=1}^n \frac{\theta_i^2}{a_i^2}, \sum_{i=1}^n \frac{\theta_i^2}{r^2} \bigg)
    \leq \sum_{i=1}^n \frac{\theta_i^2}{\min (a_i^2, r^2)}
    \leq \sum_{i=1}^n \bigg( \frac{\theta_i^2}{a_i^2} + \frac{\theta_i^2}{r^2} \bigg)
    \leq 2 \max \bigg( \sum_{i=1}^n \frac{\theta_i^2}{a_i^2}, \sum_{i=1}^n \frac{\theta_i^2}{r^2} \bigg)
    \, ,
  \end{align*}
  one deduces the inclusions $L \subset E_a \cap r B_2 \subset \sqrt{2} L$, so that $w (L) \leq w_{E_a} (r) \leq \sqrt{2} w (L)$.
  Applying Lemma~\ref{lem:width-ellipsoid} concludes the proof.
\end{proof}

We will also need the following covering number estimates for ellipsoids.
These estimates are %
classical; see \eg \cite[Theorem~3.4.2 p.~119]{carl1990entropy} for an essentially equivalent result, stated in terms of entropy numbers of operators between Hilbert spaces.

\begin{lemma}
  \label{lem:covering-ellipsoid}
  For any $r > 0$,
  \begin{equation}
    \label{eq:covering-ellipsoid}
    \log N (E_a, 5 r)
    \leq \sum_{i \pp a_i \geq 2 r} \log \Big( \frac{a_i}{r} \Big)
    \leq \log N (E_a, r)
    \, .
  \end{equation}
\end{lemma}

\begin{proof}
  Recall that $a_1 \geq \dots \geq a_n$.
  Let $d = \max \{ i : a_i \geq 2 r \}$ (with the convention $d=0$ if $\max_i a_i < 2 r$), and let $\Pi : \R^n \to \R^d$ the orthogonal projection on the first $d$ coordinates.
    
  For the second inequality, note that if $\theta_1, \dots, \theta_N$ form an $r$-cover of $E_a$, then $\Pi \theta_1, \dots, \Pi \theta_N$ form an $r$-cover of $E' = \Pi (E_a)$.
  The claimed inequality is then obtained by noting that
  \begin{equation*}
    \kappa_d \prod_{i=1}^d a_i
    = \vol_d (E')
    \leq \sum_{j=1}^N \vol_d (B_d (\Pi \theta_j, r))
    = N \kappa_d r^d
    \, .%
  \end{equation*}

  For the first inequality, denote %
  $E'= %
  \Pi (E_a)$, and identify $\R^d$ with $\R^d \times \{ 0 \} \subset \R^n$.
  Note that if $\theta \in E_a$, then $\sum_{i>d} \theta_i^2/r^2 \leq \sum_{i>d} \theta_i^2/a_i^2 \leq 1$ namely $\norm{\theta - \Pi (\theta)} \leq r$.
  Thus $E_a \subset E' + r B_2^n$, which implies that $N (E_a, (K+1) r) \leq N (E', K r)$ for any $K > 1$.
  Recalling that $N_p$ refers to packing numbers, inequality~\eqref{eq:packing-covering} provides $N (E', K r)\leq N_p (E', K r)$.
  Let $N = N_p (E', Kr)$ and $\theta_1, \dots, \theta_N$ form a $K r$-packing of $E' \subset \R^d$.
  Then, the balls $B (\theta_j, K r /2)$, $1 \leq j \leq N$, are pairwise disjoint and contained in $E ' + K r / 2 B_2^d \subset (1 + K/4) E'$, so that
  \begin{equation*}
    \vol_d \big( (1 + K/4) E' \big)
    = \kappa_d \prod_{i=1}^d \big( (1+K/4) a_i \big)
    \geq \sum_{j=1}^N \vol_d (B(\theta_j, r/2))
    = N \kappa_d (K r/2)^d
    \, .
  \end{equation*}
  This implies, for $K = 4$,
  \begin{equation*}
    \log N (E_a, 5 r)
    \leq \log N_p (E', 4r)
    \leq \sum_{i=1}^d \log \bigg( \frac{(1+K/4) a_i}{(K/2) r} \bigg)
    = \sum_{i \pp a_i \geq 2 r} \log \Big( \frac{a_i}{r} \Big)
    \, . \qedhere
  \end{equation*}  
\end{proof}

With these estimates at hand,
together with the consequence~\eqref{eq:fixed-points-isomorphic} of Lemma~\ref{lem:prop-fixed-points}, we can characterize the two fixed points associated to $E_a$:
\begin{align*}
  \frac{1}{6} \inf_{r > 0} \Big\{ \sum_{i=1}^n \min \Big( \frac{a_i^2}{r^2}, 1 \Big) + r^2 \Big\}
  &\leq r_* (E_a)^2
    \leq 2 \inf_{r > 0} \Big\{ \sum_{i=1}^n \min \Big( \frac{a_i^2}{r^2}, 1 \Big) + r^2 \Big\} \\
  \frac{1}{2} \inf_{r > 0} \bigg\{ \sum_{i \pp a_i \geq 2 r} \log \Big( \frac{a_i}{r} \Big) + r^2 \bigg\}
  &\leq \wt r (E_a)^2
    \leq 5 \inf_{r > 0} \bigg\{ \sum_{i \pp a_i \geq 2 r} \log \Big( \frac{a_i}{r} \Big) + r^2 \bigg\} 
  \, .
\end{align*}
(For the second fixed point, we used the equivalent representation in terms of $w (E_a \cap rB_2)^2/r^2 = w( r^{-1} E_a \cap B_2 )^2$ described in Section~\ref{sec:altern-expr-minim}%
.)
The minimax regret for ellipsoids is then characterized as follows:

\begin{proposition}
  \label{prop:regret-ellipsoids}
  One has:
  \begin{equation}
    \label{eq:regret-ellipsoid}
    \frac{1}{6000} \inf_{r > 0} \Big\{ \sum_{i=1}^n \log \Big( 1 + \frac{a_i^2}{r^2} \Big) + r^2 \Big\}
    \leq \mmr (E_a)
    \leq 10 \inf_{r > 0} \Big\{ \sum_{i=1}^n \log \Big( 1 + \frac{a_i^2}{r^2} \Big) + r^2 \Big\}
    \, .
  \end{equation}
  Furthermore, there exists $\lambda \geq 0$ such that the density $q_\lambda$ of the Gaussian distribution
  \begin{equation*}
    \gaussdist \Big( 0, \sum_{i=1}^n \Big( 1 + \frac{a_i^2}{\lambda} \Big) e_i e_i^\top \Big)
  \end{equation*}
  where $(e_1, \dots, e_n)$ is the canonical basis of $\R^n$, satisfies, for any sequence $y = (y_1, \dots, y_n) \in \R^n$,
  \begin{equation}
    \label{eq:regret-ridge}
    \ell (q_\lambda, y) - \inf_{\theta \in E_a} \ell (p_{\theta}, y)
    \leq 3000\, \mmr (E_a)
    \, .
  \end{equation}
\end{proposition}

Proposition~\ref{prop:regret-ellipsoids} characterizes (up to constants) the magnitude of the minimax regret, and thus the Wills functional, on ellipsoids.
The second part of this statement asserts that a suitable Gaussian predictive distribution achieves %
minimax-optimal regret over any ellipsoid.
Gaussian distributions are particularly convenient in practice from the perspective of sequential prediction (described in the introduction) %
due to the fact that conditioning a Gaussian measure on a linear projection reduces to simple matrix computations. 
Since our focus is on the functional $\mmr$ rather than practical implementation of prediction schemes, we do not expand further
on this point.

\begin{proof}[Proof of Proposition~\ref{prop:regret-ellipsoids}]
  From the characterization~\eqref{eq:minimax-regret-fixed} of Theorem~\ref{thm:minimax-regret-metric} as well as the previous fixed points estimates for ellipsoids, after re-arranging one has  
  \begin{equation}
    \label{eq:proof-ellipsoid}
    \begin{split}
      \frac{1}{3600} \inf_{r>0} \bigg\{ \sum_{i=1}^n \min \Big( \frac{a_i^2}{r^2}, 1 \Big) + \sum_{i \pp a_i \geq 2 r} 2 \log \Big( \frac{a_i}{r} \Big) + r^2 \bigg\}
    \leq \mmr (E_a) \\
      \leq 5 \inf_{r>0} \bigg\{ \sum_{i=1}^n \min \Big( \frac{a_i^2}{r^2}, 1 \Big) + \sum_{i \pp a_i \geq 2 r} 2 \log \Big( \frac{a_i}{r} \Big) + r^2 \bigg\}
      \, .
    \end{split}
  \end{equation}
  Now, for $a_i < 2 r$,
  \begin{equation*}
    ({\log 5})^{-1} \log \Big( 1 + \frac{a_i^2}{r^2} \Big)
    \leq \min \Big( \frac{a_i^2}{r^2}, 1 \Big)
    \leq \frac{a_i^2}{r^2}
    \leq (\log 2)^{-1} \log \Big( 1 + \frac{a_i^2}{r^2} \Big)
    \, ,
  \end{equation*}
  while for $a_i \geq 2 r$ (so that $\min (a_i^2/r^2, 1) = 1$)
  \begin{equation*}
    \log \Big( 1+ \frac{a_i^2}{r^2} \Big) \leq
    2 \log \Big( \frac{a_i}{r} \Big) + 1
    \leq 2 \log \Big( 1 + \frac{a_i^2}{r^2} \Big)
    \, .
  \end{equation*}
  Summing these two inequalities gives
  \begin{equation*}
    (\log 5)^{-1} \sum_{i=1}^n \log \Big( 1 + \frac{a_i^2}{r^2} \Big)
    \leq \sum_{i=1}^n \min \Big( \frac{a_i^2}{r^2}, 1 \Big) + \sum_{i \pp a_i \geq 2 r} 2 \log \Big( \frac{a_i}{r} \Big)
    \leq 2 \sum_{i=1}^n \log \Big( 1 + \frac{a_i^2}{r^2} \Big)
    \, .
  \end{equation*}
  Plugging this into~\eqref{eq:proof-ellipsoid} (and using that $3600 \log 5 \leq 6000$) gives~\eqref{eq:regret-ellipsoid}.

  For the second statement, we use the following identity shown (with different but equivalent ``kernel'' notation) by~\cite{seeger2008information}, which is also implicitly used in~\cite{vovk2001competitive}: for every $y \in \R^n$,
  \begin{equation}
    \label{eq:regret-gaussian-mixture}
    \ell (q_\lambda, y) - \inf_{\theta \in \R^n} \Big\{ \ell (p_\theta, y) + \frac{\lambda}{2} \sum_{i=1}^n \frac{\theta_i^2}{a_i^2} \Big\}
    = \frac{1}{2} \sum_{i=1}^n \log \Big( 1 + \frac{a_i^2}{\lambda} \Big)
    \, .
  \end{equation}
  This expression for the regret can be obtained through the identity
  \begin{equation*}
    \inf_{\theta \in \R^n} \Big\{ \ell (p_\theta, y) + \frac{\lambda}{2} \sum_{i=1}^n \frac{\theta_i^2}{a_i^2} \Big\}
    = \inf_{\theta \in \R^n} \Big\{ \frac{1}{2} \norm{\theta - y}^2 + \frac{\lambda}{2} \sum_{i=1}^n \frac{\theta_i^2}{a_i^2} \Big\}
    = \frac{1}{2} \sum_{i=1}^n \frac{y_i^2}{1+a_i^2/\lambda}
    \, ,
  \end{equation*}
  which corresponds to $\ell (q_\lambda, y)$ up to the %
  constant in the r.h.s.~of~\eqref{eq:regret-gaussian-mixture}.
  It follows from~\eqref{eq:regret-gaussian-mixture} and the definition of $E_a$ that, for every $y \in \R^n$,
  \begin{equation*}
    \ell (q_\lambda, y)
    - \inf_{\theta \in E_a} \ell (p_\theta, y)
    \leq \frac{1}{2} \sum_{i=1}^n \log \Big( 1 + \frac{a_i^2}{\lambda} \Big) + \frac{\lambda}{2}
    \, .
  \end{equation*}
  This upper bound together with the minimax regret lower bound in~\eqref{eq:regret-ellipsoid} imply~\eqref{eq:regret-ridge}.
\end{proof}

A similar characterization and optimality result hold for the minimax redundancy:

\begin{proposition}
  \label{prop:redundancy-ellipsoid}
  The minimax redundancy on the ellipsoid $E_a$ satisfies:
  \begin{equation}
    \label{eq:redundancy-ellipsoid}
    \frac{1}{600} \inf_{r > 0} \bigg\{ \sum_{i \pp a_i \geq 2 r} \log \Big( \frac{a_i}{r} \Big) + r^2 \bigg\}
    \leq \mred (E_a)
    \leq 5 \inf_{r > 0} \bigg\{ \sum_{i \pp a_i \geq 2 r} \log \Big( \frac{a_i}{r} \Big) + r^2 \bigg\}
    \, .
  \end{equation}
  In addition, there exists $\lambda \geq 0$ such that the density $\wt q_\lambda$ of the Gaussian distribution
  \begin{equation*}
    \gaussdist \Big( 0, \sum_{i=1}^n \Big( 1 + \frac{a_i^2}{\lambda} \indic{a_i \geq 2 \sqrt{\lambda}} \Big) e_i e_i^\top \Big)
  \end{equation*}
  satisfies %
  \begin{equation}
    \label{eq:gaussian-redundancy-ellipsoid}
    \sup_{\theta \in E_a} \kll{p_\theta}{\wt q_\lambda}
    \leq 1200\, \mred (E_a)
    \, .
  \end{equation}
\end{proposition}

\begin{proof}
  The upper and lower bounds in~\eqref{eq:redundancy-ellipsoid} are obtained by combining Theorem~\ref{thm:minimax-redundancy} with the covering number upper bound of Lemma~\ref{lem:covering-ellipsoid} and the lower bound on $\wt r (E_a)^2$ above.

  To bound the maximal redundancy of $\wt q_\lambda$, we use the following identity:
  if $q$ denotes the density of the distribution $\gaussdist (0, \Sigma)$ (for a positive semi-definite $n \times n$ matrix $\Sigma$) then
  \begin{equation*}
    \kll{p_\theta}{q}
    = \frac{1}{2} \Big( \log \det \Sigma - n + \tr (\Sigma^{-1}) \Big) + \frac{1}{2} \langle \Sigma^{-1} \theta, \theta \rangle
    \, .
  \end{equation*}
  In particular, for every $\theta \in E_a$, using that $1 + ({a_i^2}/{\lambda}) \indic{a_i \geq 2 \sqrt{\lambda}} \geq a_i^2/(4\lambda)$ for all $i$,
  \begin{align*}
    \kll{p_\theta}{\wt q_\lambda}
    &= \frac{1}{2} \sum_{i \pp a_i \geq 2 \sqrt{\lambda}} \bigg\{ \log \Big( 1 + \frac{a_i^2}{\lambda} \Big) - 1 + \frac{\lambda}{a_i^2} \bigg\} + \frac{1}{2} \sum_{i=1}^n \frac{\theta_i^2}{1 + (a_i^2/\lambda) \indic{a_i \geq 2 \sqrt{\lambda}}} \\
    &\leq \frac{1}{2} \sum_{i \pp a_i \geq 2 \sqrt{\lambda}} \log \Big( 1 + \frac{a_i^2}{\lambda} \Big) + 2 \lambda \sum_{i=1}^n \frac{\theta_i^2}{a_i^2} \\
    &\leq \frac{1}{2} \sum_{i \pp a_i \geq 2 r} \log \Big( 1 + \frac{a_i^2}{r^2} \Big) + 2 r^2
  \end{align*}
  where we let $r = \sqrt{\lambda}$.
  Inequality~\eqref{eq:gaussian-redundancy-ellipsoid}
  is then a consequence of this upper bound and of the lower bound in~\eqref{eq:redundancy-ellipsoid}.
\end{proof}

\subsection{The $\ell^1$ ball
}
\label{sec:case-ell1-ball}

Another ``canonical'' example is given by the $\ell^1$ ball.
This class plays an important role in modern high-dimensional statistics and signal processing~\cite{foucart2013compressive,buhlmann2011statistics,tsybakov2009nonparametric,koltchinskii2011oracle,wainwright2019high}.

Here, we consider the set $\alpha B_1^d = \big\{ \theta \in \R^d : \sum_{i=1}^d |\theta_i| \leq \alpha \big\}$ for some $\alpha > 0$ and $d \geq 2$.
The minimax error for Gaussian statistical estimation over this set was characterized in~\cite{donoho1994minimax}, while here we consider the hardness of universal coding.

In order to characterize the fixed point $\wt r (\alpha B_1^d)$, we first recall the notion of \emph{entropy numbers}, which are closely related to covering numbers.
Specifically, for any integer $k \geq 0$ and set $A \subset \R^n$, the $k$-th entropy number $e_k (A)$ is defined as
\begin{equation*}
  e_k (A)
  = \inf \bigg\{ r > 0 : \exists x_1, \dots, x_{2^k} \in \R^n, A \subset \bigcup_{i=1}^{2^k} B (x_i, r) \bigg\}
  = \inf \big\{ r > 0 : N (A, r) \leq 2^k \big\}
  \, .
\end{equation*}
Using the connection with covering numbers in the last expression, one readily verifies that
\begin{equation*}
  \wt r (A)^2
  \asymp \inf_{k \geq 0} \big\{ k + e_k (A)^2 \big\}
  \, .
\end{equation*}
Now, it was shown in~\cite{schutt1984entropy} that
\begin{equation*}
  e_k (B_1^d)
  \asymp \left\{
    \begin{array}{cl}
      1 & \text{ if } k \leq \log d \\
      \sqrt{\log (e d / k) / k}
        & \text{ if } \log d \leq k \leq d \\
      2^{-k/d} d^{-1/2} & \text{ if } k \geq d
                               \, .
    \end{array}
  \right.
\end{equation*}
From these estimates (together with Theorem~\ref{thm:minimax-redundancy}) and by distinguishing cases, one deduces that 
\begin{equation}
  \label{eq:redundancy-l1}
  \mred (\alpha B_1^d)
  \asymp \wt r (\alpha B_1^d)^2
  \asymp \left\{
    \begin{array}{cl}
      \alpha^2 & \text{ if } \alpha \leq \sqrt{\log d} \\
      \alpha \sqrt{\log (e d/\alpha)} & \text{ if } \sqrt{\log d} \leq \alpha \leq d \\
      d \log (e \alpha/d) & \text{ if } \alpha \geq d
                               \, .
    \end{array}
  \right.
\end{equation}

In addition, it is known from~\cite{gordon2007gaussian} that
\begin{equation*}
  w \big( \alpha B_1^d \cap r B_2^d \big)
  \asymp \left\{
    \begin{array}{cl}
      \alpha \sqrt{\log (e d \min (r^2/\alpha^2, 1))} & \text{ if } r \geq \alpha / \sqrt{d} \\
      r \sqrt{d} & \text{ if } r \leq \alpha/\sqrt{d}
                   \, .
    \end{array}
    \right.
\end{equation*}
This implies (see \eg \cite[p.~19]{lecue2013subgaussian}) that
\begin{equation}
  \label{eq:local-fixed-l1}
  r_* (\alpha B_1^d)^2
  \asymp \left\{
    \begin{array}{cl}
      \alpha \sqrt{\log (e d \min(1/\alpha, 1) )} & \text{ if } \alpha \leq d \\
      d & \text{ if } \alpha \geq d
                               \, .
    \end{array}
  \right.
\end{equation}
The fixed points estimates~\eqref{eq:redundancy-l1} and~\eqref{eq:local-fixed-l1}, together with Theorem~\ref{thm:minimax-regret-metric},
allow to conclude that
\begin{equation}
  \label{eq:regret-l1}
  \mmr (\alpha B_1^d)
  \asymp \left\{
    \begin{array}{cl}
      \alpha \sqrt{\log d} & \text{ if } \alpha \leq 1
      \\
      \alpha \sqrt{\log (e d/\alpha)} & \text{ if } 1 \leq \alpha \leq d
      \\
      d \log (e \alpha / d) & \text{ if } \alpha \geq d
                               \, .
    \end{array}
  \right.
\end{equation}
An upper bound on the minimax regret essentially corresponding to the second regime in~\eqref{eq:regret-l1} (up to a non-explicit dependence on the norm inside the logarithm, possibly due to the asymptotic regime implicitly considered there) was obtained in~\cite{miyaguchi2019adaptive}.
This upper bound matches (in the regime $\sqrt{\log d} \leq \alpha \leq d^{1-\eps}$) the order of the minimax estimation error from~\cite{donoho1994minimax}.

\section{Proof of Theorems~\ref{thm:minimax-regret-metric} and~\ref{thm:minimax-redundancy} and related results}
\label{sec:proof-theorems}

\subsection{Fixed points and local widths} %
\label{sec:prop-fixed-points}

For convenience, we gather here some basic properties of the ``fixed points'' we encounter.

\begin{lemma}
  \label{lem:prop-fixed-points}
  For any function $g : \R^+ \to \R^+ \cup \{ + \infty \}$ which is not identically $+\infty$, define
  \begin{equation*}
    \phi (g)
    = \sup \big\{ r \geq 0 : g (r) \geq r^2 \big\}
    \in [0, + \infty]
    \, .
  \end{equation*}
  \begin{enumerate}
  \item Let $g, h : \R^+ \to \R^+ \cup \{ + \infty \}$.
    Then:
    \begin{enumerate}
    \item[$(a)$] if $g \leq h$, then $\phi (g) \leq \phi (h)$;
    \item[$(b)$] $\phi (\max(g, h)) = \max \{ \phi (g), \phi (h) \}$.
    \end{enumerate}
  \item Assume now that $r \mapsto g (r)/r$ is non-increasing on $\R_+^*$.
    Then,
    \begin{enumerate}
    \item[$(a)$] $\phi (g) < \infty$, and $r^2 < g (r)$ for $r < \phi (g)$ while $r^2 > g (r)$ for $r > \phi (g)$;
    \item[$(b)$] for any $\lambda \geq 1$, one has $\phi (\lambda g) \leq \lambda \phi (g)$;
    \item[$(c)$] for any $\lambda \geq 1$, let $g_\lambda : r \mapsto g (x/\lambda)$; then, one has $\phi (g_\lambda) \leq \lambda \phi (g)$.
    \end{enumerate}
  \item Finally, assume that $g$ is non-increasing.
    Then,
    \begin{equation*}
      \phi (g)^2
      \leq \inf_{r \geq 0} \big\{ g (r) + r^2 \big\}
      \leq 2 \, \phi (g)^2
      \, .
    \end{equation*}
  \end{enumerate}
\end{lemma}

In particular, if $g, h$ are such that $g(r)/r$ and $h(r)/r$ are non-increasing in $r$, and if there are constants $C_1, C_2 \geq 1$ and $c_1, c_2 \leq 1$ such that, for every $r \geq 0$,
\begin{equation*}
  c_1 h (r/c_2)
  \leq g (r)
  \leq C_1 h (r/C_2)
  \, ,
\end{equation*}
then
\begin{equation}
  \label{eq:fixed-points-isomorphic}
  c_1 c_2 \, \phi (h)
  \leq \phi (g)
  \leq C_1 C_2 \, \phi (h)
  \, .
\end{equation}

\begin{proof}
  Properties 1(a-b) follow directly from the definition of $\phi$.
  Property~2(a) comes from the fact that $\phi (g) = \sup \{ r > 0 :
  r \leq g(r) /r
  \}$ (with the convention that $\sup \varnothing = 0$), and that $r \mapsto g (r)/r$ is non-increasing (and thus bounded on $[x, + \infty)$ where $g(x) < \infty$), while $r$ is increasing and unbounded.
  To prove property~2(b), let $r < \phi (\lambda g)$ be arbitrary.
  By 2(a) %
  and since $g (r) / r \leq g(r/\lambda)/(r/\lambda)$, 
  \begin{equation*}
    r^2
    < \lambda \, g ( r )
    \leq \lambda^2 g \Big( \frac{r}{\lambda} \Big)
  \end{equation*}
  namely $(r/\lambda)^2 < g (r/\lambda)$ and thus $r/ \lambda \leq \phi (g)$.
  Letting $r \to \phi (g)$ we get $\phi (\lambda g) \leq \lambda \phi (g)$.
  Likewise, to prove 2(c), let $r < \phi (g_\lambda)$ be arbitrary.
  By 2(a),
  \begin{equation*}
    \Big( \frac{r}{\lambda} \Big)^2
    \leq r^2
    < g \Big( \frac{r}{\lambda} \Big)
  \end{equation*}
  so that $r/\lambda \leq \phi (g)$ which concludes the proof.
  Finally, let us prove property~3.
  Recall that $\phi(g) < \infty$ by 2(a).
  To prove the second inequality, let $r > \phi (g)$ be arbitrary, so that $g(r) < r$ and thus $g (r) + r^2 \leq 2 r^2$, and let $r \to \phi (g)$.
  We now prove the first inequality, namely $\phi (r) + r^2 \geq \phi (g)^2$ for any $r \geq 0$.
  If $r \geq \phi (g)$, then $\phi (r) + r^2 \geq r^2 \geq \phi (g)^2$.
  If $r < \phi (g)$, let $r'\in (r, \phi (g))$ be arbitrary.
  Since $\phi$ is non-increasing and $r < r' < \phi (r)$, one has
  $\phi (r) + r^2 \geq \phi (r) \geq \phi (r') \geq (r')^2$; letting $r' \to \phi(g)$ concludes.  
\end{proof}

In the convex case, local Gaussian widths have the following convenient properties:
\begin{lemma}
  \label{lem:local-gaussian-width}
  Let $K \subset \R^n$ be a convex body.
  \begin{enumerate}
  \item The function $r \mapsto w_K (r) / r$ is decreasing;
  \item If $K$ is symmetric, then $w_K (r) = w (K \cap r B_2)$.
  \end{enumerate}
\end{lemma}

\begin{proof}
  We start with the first point.
  Up to translating $K$, %
  assume that $0 \in K$.
  Let $0 < r \leq r'$ and $\lambda = r'/r \geq 1$.
  For every $\theta \in K$, %
  \begin{align*}
    w \big( K \cap B (\theta, r') \big)
    &= w \big( (K - \theta) \cap \lambda r B_2 \big) \\
    &= \lambda w \big( \lambda^{-1} (K - \theta) \cap r B_2 \big) \\
    &\leq \lambda w_{\lambda^{-1} K} (r) \\
    &\leq \lambda w_{K} (r)
      \, ,
  \end{align*}
  where the last inequality comes from the fact that $\lambda^{-1} K \subset K$ (as $K$ is convex and contains $0$).
  Taking the supremum over $\theta \in K$ and substituting for $\lambda$
  leads to $w_K (r')/r' \leq w_K (r)/r$.

  For the second point, fix $r \geq 0$ and let $g (\theta) = w (K \cap B (\theta, r))$ for $\theta \in K$, so that $w_K (r) = \sup_{\theta \in K} g (\theta)$.
  Since $K$ is symmetric, the function $g$ is even, namely $g(-\theta) = g(\theta)$.
  It is also concave since for $\theta, \theta' \in K$ and $\lambda \in [0, 1]$ one has
  \begin{equation*}
    K \cap B (\lambda \theta + (1-\lambda) \theta', r)
    \supset \lambda [K \cap B (\theta, r)] + (1-\lambda) [ K \cap B (\theta', r) ]
    \, ,
  \end{equation*}
  so that (using that $w (A+B) = w (A) + w (B)$, %
  by
  definition of the Gaussian width)
  \begin{equation*}
    g \big( \lambda \theta + (1-\lambda) \theta' \big)
    \geq \lambda g(\theta) + (1-\lambda) g(\theta')
    \, .
  \end{equation*}
  This implies, for every $\theta \in K$, that $g (\theta) = \big( g (\theta) + g (-\theta) \big)/2
  \leq g ( (\theta - \theta)/2 ) = g(0)$.
\end{proof}

\subsection{Upper bound on the redundancy}
\label{sec:upper-bound-redund}

We now include the proof of the upper bounds of Theorem~\ref{thm:minimax-redundancy} on the minimax redundancy.
The upper bound is classical and due to~\cite{yang1999information}, who used it to further bound the rate of estimation.

Fix $A \subset \R^n$.
If $N (A, r) = \infty$ for any $r > 0$ (\ie, $A$ is unbounded), then the right-hand sides of both~\eqref{eq:minimax-redundancy-fixed} and~\eqref{eq:minimax-redundancy-sum} are infinite
and the upper bounds hold.
Let $r > 0$ be arbitrary such that $N = N (A, r)$ is finite.
Let $\theta_1, \dots, \theta_N \in A$ be an $r$-cover of $A$, and define the
density
\begin{equation*}
  q_r 
  = \frac{1}{N} \sum_{i=1}^N p_{\theta_i}
  \, .
\end{equation*}

For every $\theta \in A$ and $1 \leq i \leq N$, using that $q_r \geq p_{\theta_i}/N$ one has
\begin{align*}
  \kll{p_\theta}{q_r}
  = \int_{\R^n} p_\theta \log \Big( \frac{p_\theta}{q_r} \Big)
  \leq
  \int_{\R^n} p_\theta \log \Big( \frac{N \cdot p_\theta}{p_{\theta_i}} \Big) =
  \kll{p_{\theta}}{p_{\theta_i}} + \log N
  \, .
\end{align*}
Using the %
identity $\kll{p_\theta}{p_{\vartheta}} = \norm{\theta - \vartheta}^2/2$ for $\theta, \vartheta \in \R^n$, taking the infimum over $i$, and using that $\{ \theta_1, \dots, \theta_N \}$ form an $r$-cover of $A$, this gives:
\begin{equation*}
  \kll{p_\theta}{q_r}
  \leq \frac{1}{2} \min_{1 \leq i \leq N} \norm{\theta - \theta_i}^2 + \log N
  \leq \frac{r^2}{2} + \log N (A, r)
  \, .
\end{equation*}
Since $r > 0$ was arbitrary, this gives
\begin{equation}
  \label{eq:proof-bound-redundancy}
  \mred (A)
  \leq \inf_{r > 0} \Big\{ \log N (A, r) + \frac{r^2}{2} \Big\}
  \leq 2 \, \wt r (A)^2
  \, ,
\end{equation}
where the second inequality uses point~3 of Lemma~\ref{lem:prop-fixed-points}, applied to
$g (r) = \log N (A, r)$.

\subsection{Lower bound on the redundancy}
\label{sec:lower-bound-redund}

The proof starts by lower-bounding the minimax redundancy by a mutual information.

\begin{lemma}
  \label{lem:lower-mutual}
  For any $N \geq 1$ and $\theta_1, \dots, \theta_N \in A$, denoting $\ol p_N = N^{-1} \sum_{i=1}^N p_{\theta_i}$ one has:
  \begin{equation*}
    \mred (A)
    \geq \frac{1}{N} \sum_{i=1}^N \kll{p_{\theta_i}}{\ol p_N}
    \, .
  \end{equation*}
\end{lemma}

\begin{proof}%
  For any density $q$ on $\R^n$, one has
  \begin{align}
    \sup_{\theta \in A} \kll{p_\theta}{q}
    &\geq \frac{1}{N} \sum_{i=1}^N \kll{p_{\theta_i}}{q} %
    = \frac{1}{N} \sum_{i=1}^N \kll{p_{\theta_i}}{\ol p_N} + \kll{\ol p_N}{q}
    \, ,
  \end{align}
  where the last identity %
  can be verified directly.
  Lower-bounding $\kll{\ol p_N}{q} \geq 0$ for any $q$ and taking the infimum of the l.h.s.~over $q$ concludes the proof of Lemma~\ref{lem:lower-mutual}.
\end{proof}

The following lemma is an important ingredient of the proof.
It states that if $A$ contains a sufficiently large and separated subset, then its minimax redundancy is large.

\begin{lemma}
  \label{lem:lower-separated}
  Assume that $\theta_1, \dots, \theta_N \in \R^n$ are such that $\norm{\theta_i - \theta_j} \geq r$ for $i \neq j$, with $r^2 \geq 32 \log N$.
  Then,
  \begin{equation}
    \label{eq:lower-separated}
    \frac{1}{N} \sum_{i=1}^N \kll{p_{\theta_i}}{\ol p_N}
    \geq \log N - 3
    \, .
  \end{equation}
\end{lemma}

\begin{proof}%
  In order to prove~\eqref{eq:lower-separated}, it suffices to prove that $\kll{p_{\theta_i}}{\ol p_N} \geq \log N - 3$ for each $i = 1, \dots, N$.
  By symmetry, it suffices to prove it for $i = 1$; and up to translating $A$, one may assume that $\theta_1 = 0$.
  In particular, $\norm{\theta_i} \geq r$ for $2 \leq i \leq N$.
  We then have, letting $X \sim \gaussdist (0, I_n)$,
  \begin{align*}
    \kll{p_0}{\ol p_N}
    &= \int_{\R^n} p_0 \log \bigg( \frac{p_0}{N^{-1} \sum_{i=1}^N p_{\theta_i}} \bigg) 
    = \log N - \int_{\R^n} p_0 \log \bigg( \sum_{i=1}^N \frac{p_{\theta_i}}{p_0} \bigg) \\
    &= \log N - \E \bigg[ \log \bigg( 1 + \sum_{i=2}^N e^{\langle \theta_i, X\rangle - \norm{\theta_i}^2/2} \bigg) \bigg]
      \, .
  \end{align*}
  Hence, $\kll{p_0}{\ol p_N} = \log N - \E [Z]$, where
  \begin{equation*}
    Z = \log \bigg( 1 + \sum_{i=2}^N e^{\langle \theta_i, X\rangle - \norm{\theta_i}^2/2} \bigg)
    \, .
  \end{equation*}
  It therefore suffices to show that $\E [Z] \leq 3$.
  We write, using that $\norm{\theta_i} \geq r$ for $2 \leq i \leq N$,
  \begin{align*}
    Z
    &\leq \log \bigg\{ 1 + N \max_{2 \leq i \leq N} \exp \Big[ \langle \theta_i, X\rangle - \frac{\norm{\theta_i}^2}{2} \Big] \bigg\} \\
    &= \log \bigg\{ 1 + \exp \bigg[ \Big( \log N - \frac{r^2}{32} \Big) + \max_{2 \leq i \leq N} \Big( \langle \theta_i, X\rangle - \frac{15 \norm{\theta_i}^2}{32} \Big) \bigg] \bigg\} \\
    &\leq \log \bigg\{ 1 + \exp \bigg[ \max_{2 \leq i \leq N} \Big( \langle \theta_i, X\rangle - \frac{15 \norm{\theta_i}^2}{32} \Big) \bigg] \bigg\}
  \end{align*}
  where the last inequality is due to the assumption that $r^2 \geq 32 \log N$.
  Now, for $2 \leq i \leq N$, one has $\langle \theta_i, X\rangle \sim \gaussdist (0, \norm{\theta_i}^2)$, so by a %
  Gaussian tail bound (\eg, \cite[p.~22]{boucheron2013concentration}), for any $\delta \in (0,1)$:
  \begin{equation*}
    \P \big( \langle \theta_i, X\rangle \geq \norm{\theta_i} \sqrt{2 \log (1/\delta)} \big)
    \leq \delta
    \, .
  \end{equation*}
  By a union bound, if $\delta \in (0,1)$, with probability at least $1-\delta$, for $i=2, \dots, N$,
  \begin{align*}
    \langle \theta_i, X\rangle
    &< \norm{\theta_i} \sqrt{2 \log (N/\delta)} 
    \leq \norm{\theta_i} \sqrt{2 \log N} + \norm{\theta_i} \sqrt{2 \log (1/\delta)} 
    \leq \frac{\norm{\theta_i}^2}{4} + \norm{\theta_i} \sqrt{2 \log (1/\delta)}
      \, ,
  \end{align*}
  where the last inequality uses that $\log N \leq r^2/32 \leq \norm{\theta_i}^2/32$.
  Hence, with probability %
  $1-\delta$,
  \begin{align*}
    \max_{2 \leq i \leq N} \Big( \langle \theta_i, X\rangle - \frac{15 \norm{\theta_i}^2}{32} \Big)
    &\leq \max_{2 \leq i \leq N} \Big( \norm{\theta_i} \sqrt{2 \log (1/\delta)} + \frac{\norm{\theta_i}^2}{4} - \frac{15 \norm{\theta_i}^2}{32} \Big) \\
    &\leq \sup_{\rho \geq 0} \Big( \rho \sqrt{2 \log (1/\delta)} - \frac{7 \rho^2}{32} \Big) 
    = \frac{16}{7} \log (1/\delta)      
      \, .
  \end{align*}
  Putting the previous inequalities together, with probability at least $1-\delta$,
  \begin{equation*}
    Z
    \leq %
    \log \big( 1 + \delta^{-16/7} \big)
    \leq \log \big( 2 \delta^{-16/7} \big)
    = \log 2 + \frac{16}{7} \log (1/\delta)
    \, .
  \end{equation*}
  Integrating the tails gives:
  \begin{equation*}
    \E [Z]
    \leq \int_{0}^1 \Big( \log 2 + \frac{16}{7} \log (1/\delta) \Big) \di \delta
    = \log 2 + \frac{16}{7}
    \leq 3
    \, .
  \end{equation*}
  This concludes the proof.
\end{proof}

We note that Lemma~\ref{lem:lower-separated} could also be proved using a lower bound on the mutual information in terms of a log-Laplace transform of the Rényi divergence from Haussler and Opper~\cite[Theorem~2]{haussler1997mutual}; this lower bound implies in particular
a lower bound on the minimax redundancy in terms of Hellinger covering numbers~\cite[Theorem~32.5 p.~530]{polyanskiy2022information}.
The Hellinger distance is particularly convenient by virtue of being a distance (allowing, for instance, to relate packing and covering numbers); %
on the other hand, in the Gaussian case relaxing the Rényi divergence into a Hellinger distance
leads to suboptimal lower bounds
in general.
We thank Yanjun Han for bringing Haussler and Opper's inequality to our attention.

With these lemmas at hand, we conclude the proof of the lower bound in Theorem~\ref{thm:minimax-redundancy}.
Let $\wt r = \wt r (A) \geq 0$ (which may be infinite, in case $A \subset \R^n$ is unbounded).

\paragraph{First case: $\wt r \gtrsim 1$.}

We first assume that $\wt r$ is large enough, namely
$\wt r \geq 12$.

For any $r > 0$, let $N_p (A, r)$ denote the packing number of $A$ at scale $r$, defined as the supremum of all $N \geq 1$ for which there exist $\theta_1, \dots, \theta_N \in A$ such that $\norm{\theta_i - \theta_j} > r$ for $1 \leq i, j \leq N$, $i \neq j$.
Packing and covering numbers are related through the following standard inequalities (\eg, \cite[Lemma~4.2.8 p.~83]{vershynin2018high}):
\begin{equation}
  \label{eq:packing-covering}
  N_p (A, 2 r)
  \leq N (A, r)
  \leq N_p (A, r)
  \, .
\end{equation}
Define
\begin{equation*}
  \wt r_p
  = \sup \Big\{ r \geq 0 : \log N_p (A, r) \geq r^2 \Big\}
  \, .
\end{equation*}
By inequality~\eqref{eq:packing-covering} and the consequence~\eqref{eq:fixed-points-isomorphic} of Lemma~\ref{lem:prop-fixed-points}, we have that
\begin{equation*}
  \wt r_p / 2
  \leq \wt r
  \leq \wt r_p
  \, .
\end{equation*}

This implies that $\wt r_p > 12$.
Now, let $12 < r < \wt r_p$ be arbitrary, and
let $N \geq 1$ be the largest integer such that $\log N \leq r^2/32$.
Since $r \geq 12$, one has $N \geq 90$ and $r^2 / 32 \leq \log (N+1) \leq \frac{\log 91}{\log 90} \log N$
so that $\log N \geq r^2/33$.

In addition, since $r < \wt r_p$ and $\log N_p (A, r)$ is non-increasing, one has $\log N_p (A, r) \geq r^2 \geq r^2/32 \geq \log N$.
Hence, $N \leq N_p (A, r)$ and one may find $\theta_1, \dots, \theta_N \in A$ such that $\norm{\theta_i - \theta_j} \geq r$ for $i \neq j$.
By Lemmas~\ref{lem:lower-mutual} and~\ref{lem:lower-separated}, it follows that
\begin{equation*}
  \mred (A)
  \geq \log N - 3
  \geq \frac{r^2}{33} - 3
  \geq \frac{r^2}{110}
\end{equation*}
since $r \geq 12$.
Letting $r\to \wt r_p$, we obtain $\mred (A) \geq \wt r_p^2/110 \geq \wt r^2/110$.

\paragraph{Second case: $\wt r \lesssim 1$.}

We now handle the complementary case $0 \leq \wt r < 12$.
(The derivation below is only needed to lower-bound the redundancy; for the minimax regret $\mmr (A)$ and Wills functional, it is implied by the stronger lower bound of Section~\ref{sec:proof-lower-regret}.)

First, note that since $\log N (A, r)$ is non-increasing, if $\wt r = 0$ then for any $r > 0$ one has $\log N (A, r) = 0$ namely $N (A, r) = 1$.
This implies that $A$ has a single element $\theta \in \R^n$ and thus (taking $q = p_\theta$ in the expression~\eqref{eq:def-redundancy} of $\mred (A)$) $\mred (A) = 0$.

Now, assume that $\wt r > 0$.
For $0 < r < \wt r$, one has $\log N (A, r) \geq r^2 > 0$ so $N (A, r) \geq 2$.
It follows that $A$ contains two elements $\theta_1, \theta_2$ such that $\rho = \norm{\theta_1 - \theta_2} > r$.
By Lemma~\ref{lem:lower-mutual}, 
\begin{align*}
  \mred (A)
  \geq \frac{1}{2} \Big\{ \kl \Big( {p_{\theta_1}}, {\frac{p_{\theta_1} + p_{\theta_2}}{2}} \Big) + \kl \Big( {p_{\theta_2}}, {\frac{p_{\theta_1} + p_{\theta_2}}{2}} \Big) \Big\}
  = \kl \Big( {p_{0}}, {\frac{p_{0} + p_{\rho}}{2}} \Big)
  \, ,
\end{align*}
where the reduction to $(\theta_1, \theta_2) = (0, \rho) \in \R \times \R$ uses the invariance and independence properties of Gaussian measures.
Let $\tv{p - q} = \frac{1}{2} \int |p - q|$ denote the total variation distance between two densities $p,q$.
By Pinsker's inequality (\eg, \cite[Lemma~2.5 p.~88]{tsybakov2009nonparametric}),
\begin{equation*}
  \kl \Big( p_0, \frac{p_0 + p_\rho}{2} \Big)
  \geq 2 \cdot \Big\| p_0 - \frac{p_0 + p_\rho}{2} \Big\|_{\mathrm{TV}}^2
  = \frac{1}{2} \cdot \tv{p_0 - p_\rho}^2
  \, .
\end{equation*}
Now, a direct computation gives, 
denoting $p (x) = p_0 (x) = e^{-x^2/2} / \sqrt{2\pi}$,
\begin{equation*}
  \tv{p_0 - p_\rho}
  = \frac{1}{2} \Big( \int_{-\infty}^{\rho/2} \big( p (x) - p (x - \rho) \big) \di x + \int_{\rho/2}^{\infty} \big( p (x-\rho) - p (x) \big) \di x \Big)
  = 2 \int_{0}^{\rho/2} p (x) \di x
\end{equation*}
Since $p$ is positive and decreasing on $\R^+$, $\tau (\rho) = \tv{p_0 - p_\rho}$ is increasing and $\tau(\rho)%
/\rho$ decreasing with respect to $\rho$.
Recalling that $\rho \geq r$ and $r < \wt r \leq 12$, this gives
\begin{equation*}
  \tv{p_0 - p_\rho}
  \geq \tau (r)
  \geq \frac{\tau (12)}{12} \cdot r
  \geq \frac{r}{12.1}
  \, .
\end{equation*}
Combining the previous inequalities and recalling that $0 < r < \wt r$ was arbitrary, we deduce that
\begin{equation*}
  \mred (A)
  \geq \sup_{0 < r < \wt r} \frac{1}{2} \Big( \frac{r}{12.1} \Big)^2
  \geq \frac{\wt r^2}{300}
  \, .
\end{equation*}

We have thus proved that, in all cases, $\mred (A) \geq \wt r (A)^2 /300$.
The alternative lower bound \eqref{eq:minimax-redundancy-sum} comes from the fact that, by Lemma~\ref{lem:prop-fixed-points},
\begin{equation*}
  \wt r (A)^2 \geq \frac{1}{2} \,
  \inf_{r > 0} \big\{ \log N (A, r) + r^2 \big\}
  \, .
\end{equation*}

\subsection{First upper bound on the regret}
\label{sec:proof-upper-regret}

We now turn to the study of the minimax regret $\mmr (A)$, starting from the upper bound
\begin{equation}
  \label{eq:upper-regret-sum}
  \mmr (A)
  \leq \inf_{r > 0} \Big\{ w_A (r) + \log N (A, r) \Big\}
  \, .
\end{equation}

This upper bound is obtained by naturally combining two distinct basic bounds, namely McMullen's ``isoperimetric'' bound~\eqref{eq:wills-width} and another ``mixture'' argument already used in the upper bound on the redundancy in Section~\ref{sec:upper-bound-redund}.

We start with the mixture bound: if $A \subset \bigcup_{i=1}^N A_i$ for some $N \geq 1$ and $A_1, \dots, A_N \subset \R^n$, %
\begin{align*}
  \mmr (A)
  &\leq \log \bigg( \frac{1}{(2\pi)^{n/2}} \int_{\R^n} e^{- \min_{1 \leq i \leq N} \dist^2 (x, A_i)/2} \di x \bigg) \\
  &\leq \log \bigg( \frac{1}{(2\pi)^{n/2}} \int_{\R^n} \sum_{i=1}^N e^{- \dist^2 (x, A_i)/2} \di x \bigg) 
  = \log \bigg( \sum_{i=1}^N \exp \big( \mmr (A_i) \big) \bigg) 
\end{align*}
so that
\begin{equation}
  \label{eq:mixture-regret}
  \mmr (A)
  \leq \max_{1 \leq i \leq N} \mmr (A_i) + \log N
  \, .
\end{equation}

Now, for any $r > 0$, let $N = N (A, r)$ and $\{ \theta_1, \dots, \theta_N \} \subset A$ an $r$-cover of $A$.
One may write
\begin{equation*}
  A
  = \bigcup_{i=1}^N \big[ A \cap B (\theta_i, r) \big]
  \, .
\end{equation*}
Applying the mixture bound~\eqref{eq:mixture-regret} followed by McMullen's upper bound (Proposition~\ref{prop:finiteness}) gives
\begin{align*}
  \mmr (A)
  &\leq \max_{1 \leq i \leq N} \mmr \big( A \cap B (\theta_i, r) \big) + \log N %
\leq \sup_{\theta \in A} w \big( A \cap B (\theta, r) \big) + \log N (A, r)
    \, ,
\end{align*}
which is precisely~\eqref{eq:upper-regret-sum}, since $r > 0$ is arbitrary.

\subsection{Lower bound on the regret}
\label{sec:proof-lower-regret}

We now prove the following lower bound: for any $A \subset \R^n$
\begin{equation}
  \label{eq:proof-lower-regret-fixed}
  \mmr (A)
  \geq \max \bigg( \frac{r_* (A)^2}{2}, \frac{\wt r (A)^2}{300} \bigg)
  \, .
\end{equation}

First, the lower bound in terms of $\wt r (A)$ follows from the lower bound for the redundancy proved in Section~\ref{sec:lower-bound-redund}: by~\eqref{eq:redundancy-regret} and~\eqref{eq:minimax-redundancy-fixed},
\begin{equation*}
  \mmr (A)
  \geq \mred (A)
  \geq \frac{\wt r (A)^2}{300}
  \, .
\end{equation*}

It remains to show the lower bound in terms of $r_* (A)$.
Let $r > 0$ and $\theta \in A$ be arbitrary.
By monotonicity of $\mmr (A)$ under inclusion (Proposition~\ref{prop:properties-minimax}), one has:
\begin{equation*}
  \mmr (A)
  \geq \mmr \big( A \cap B (\theta, r) \big)
  \, .
\end{equation*}
In addition, since $A \cap B (\theta, r) \subset B (\theta, r)$, the lower bound~\eqref{eq:reverse-isometry-wills} implies that
\begin{equation*}
  \mmr \big( A \cap B (\theta, r) \big)
  \geq w \big( A \cap B (\theta, r) \big) - \frac{r^2}{2}
  \, .
\end{equation*}
Combining the previous inequalities and taking the supremum over $\theta \in A$ yields, for every $r > 0$,
\begin{equation}
  \label{eq:proof-lower-diff}
  \mmr (A)
  \geq \sup_{\theta \in A} w \big( A \cap B (\theta, r) \big) - \frac{r^2}{2}
  = w_A (r) - \frac{r^2}{2}
  \, .
\end{equation}

Now, let $r < r_* (A)$ be arbitrary.
By definition~\eqref{eq:fixed-point-local} of $r_* (A)$, there exists $r' \geq r$ such that $w_A (r') \geq r'^2$.
By~\eqref{eq:proof-lower-diff} applied to $r'$, one has
\begin{equation*}
  \mmr (A)
  \geq w_A (r') - \frac{r'^2}{2}
  \geq \frac{r'^2}{2}
  \geq \frac{r^2}{2}
  \, .
\end{equation*}
Taking the supremum over $r < r_* (A)$ gives $\mmr (A) \geq r_* (A)^2/2$, concluding the proof of~\eqref{eq:proof-lower-regret-fixed}.

\subsection{Relating the upper and lower bounds}
\label{sec:proof-relating-upper-lower}

At this point, we have shown that
\begin{equation*}
  \max \bigg( \frac{r_* (A)^2}{2}, \frac{\wt r (A)^2}{300} \bigg)
  \leq \mmr (A)
  \leq \inf_{r > 0} \Big\{ w_A (r) + \log N (A, r) \Big\}
  \, .
\end{equation*}

To show that both upper and lower bounds above are sharp and conclude the proof of Theorem~\ref{thm:minimax-regret-metric}, it remains to bound (up to constants) the rightmost term by the leftmost term.

In order to lighten notation, we keep dependence on $A$ implicit, and let $r_*  = r_* (A)$, $\wt r = \wt r (A)$, and for $r > 0$ denote $N (r) = N (A, r)$ and $w (r) = w_A (r)$.
It remains to show that $\inf_{r > 0} \big\{ w (r) + \log N (r) \big\} \lesssim \max (r_*^2, \wt r^2)$.

If $\max (r_*, \wt r) = +\infty$, the inequality clearly holds, so we assume that $\max (r_*, \wt r) < \infty$.
Let $r > \max (r_*, \wt r)$ be arbitrary.
By definition of $r_*$ and $\wt r$, one has
\begin{align*}
  w (r)
  < r^2
  \quad \text{ and } \quad
  \log N (r)
  < r^2
  \, ,
\end{align*}
so that $w (r) + \log N (r) < 2 r^2$.
Taking the infimum over $r > \max (r_*, \wt r)$ leads to
\begin{equation}
  \label{eq:proof-relate-fix-sum}
  \inf_{r > 0} \Big\{ w_A (r) + \log N (A, r) \Big\}
  \leq 2 \max \big( r_* (A)^2, \wt r (A)^2 \big)
  \, .
\end{equation}
Combining inequalities~\eqref{eq:upper-regret-sum},~\eqref{eq:proof-lower-regret-fixed} and~\eqref{eq:proof-relate-fix-sum} concludes the proof of Theorem~\ref{thm:minimax-regret-metric}.

\subsection{Generic chaining characterization (Corollary~\ref{cor:majorizing-measure-regret})}
\label{sec:proof-chain-char}

We now establish the alternative characterizations of Corollary~\ref{cor:majorizing-measure-regret} in terms of admissible sequences and truncated $\gamma_2$ functionals.

We start with the proof the upper bound in~\eqref{eq:mm-regret-sum}.
Let $\A = (\A_j)_{j \geq 0}$ be an admissible sequence of partitions of $A$, and let $p \in \N$ be arbitrary.
By the upper bound~\eqref{eq:mixture-regret} applied to the partition $\A_p$ of $A$ (with at most $N_p$ elements) combined with the Gaussian width upper bound of Proposition~\ref{prop:finiteness}, one has
\begin{equation*}
  \mmr (A)
  \leq \log N_p + \max_{A' \in \A_p} \mmr (A')
  \leq \log N_p + \max_{A' \in \A_p} w (A')
  \, .
\end{equation*}
Note that $\log N_p = 0$ if $p=0$ and $\log N_p = 2^p \log 2$ if $p \geq 1$, so that $\log N_p \leq 2 (2^p - 1)$ for all $p \geq 0$.
In addition, for every $A' \in \A_p$, the sequence $(\A^{A'}_j)_{j \geq 0}$ of partitions of $A'$ defined by $\A^{A'}_j = \{ A' \}$ for $j \leq p-1$, and $\A^{A'}_j = \{ B \in \A_j : B \subset A'\}$ for $j \geq p$, is an admissible sequence.
Also, for $\theta \in A' \subset A$, the element of $\A^{A'}_j$ containing $\theta$ is $A' = A_p (\theta)$ if $j \leq p-1$ and $A_j (\theta)$ if $j \geq p$.
By the upper bound in~\eqref{eq:majorizing-measure}, one has
\begin{align*}
  w (A')
  &\leq c \sup_{\theta \in A'} \bigg( \sum_{j=0}^{p-1} 2^{j/2} \diam A_p (\theta) + \sum_{j \geq p} 2^{j/2} \diam A_j (\theta) \bigg) \\
  &= c \sup_{\theta \in A'} \bigg( \frac{2^{p/2} - 1}{\sqrt{2} - 1} \diam A_p (\theta) + \sum_{j \geq p} 2^{j/2} \diam A_j (\theta) \bigg) \\
  &\leq \frac{2 c}{\sqrt{2}-1} \sup_{\theta \in A'} \sum_{j \geq p} 2^{j/2} \diam A_j (\theta)
    \, .
\end{align*}
Putting these bounds together gives
\begin{align*}
  \mmr (A)
  &\leq 2 (2^p - 1) + \frac{2 c}{\sqrt{2} - 1} \max_{A' \in \A_p} \sup_{\theta \in A'} \sum_{j \geq p} 2^{j/2} \diam A_j (\theta) \\
  &\leq 2 (\sqrt{2} + 1) c \, \bigg\{ (2^p - 1) + \sup_{\theta \in A} \sum_{j \geq p} 2^{j/2} \diam A_j (\theta) \bigg\}
    \, .
\end{align*}
Since $\A$ and $p \in \N$ were arbitrary, taking the infimum over both gives the upper bound in~\eqref{eq:mm-regret-sum}.

We now turn to the lower bound in~\eqref{eq:mm-regret-sum}.
It follows from the lower bound in Theorem~\ref{thm:minimax-regret-metric} together with the lower bound in the Majorizing Measure theorem (see~\eqref{eq:majorizing-measure}).

By the lower bound~\eqref{eq:minimax-regret-sum} of Theorem~\ref{thm:minimax-regret-metric}, it suffices to show the following: there exists an absolute constant $C'$ such that, for every $r > 0$,
there exists $p \in \N$ and an admissible sequence $(\A_j)_{j \geq 0}$ such that
\begin{equation}
  \label{eq:proof-mm-lower}
  (2^p - 1) + \sup_{\theta \in A} \sum_{j \geq p} 2^{j/2} \diam A_j (\theta)
  \leq C' \Big\{ \log N (A, r) + w_A (r) \Big\}
  \, .
\end{equation}
We now fix $r > 0$.

If $N (A, r) = 1$, we let $p = 0$ and then~\eqref{eq:proof-mm-lower} reduces to the lower bound~\eqref{eq:majorizing-measure} in the Majorizing Measure theorem.

Now, assume that $N = N (A, r) \geq 2$ and let $p \geq 1$ be the smallest integer such that $N_p \geq N$, so that (by definition of $N_p$) $\log N \leq \log N_p \leq 2 \log N$.
Let $\{ \theta_1, \dots, \theta_N \} \subset A$ be an $r$-cover of $A$, so that $A = \bigcup_{i=1}^N A_i$ where $A_i = A \cap B (\theta_i, r)$.
Note that $w (A_i) \leq w_A (r)$ by definition of $w_A$.

For every $i = 1, \dots, N$, by the lower bound~\eqref{eq:majorizing-measure} in the Majorizing Measure theorem, there exists an admissible sequence $(\A_{ij})_{j \geq 0}$ of partitions of $A_i$ such that
\begin{equation*}
  \sup_{\theta \in A_i} \sum_{j \geq p} 2^{j/2} \diam A_{ij} (\theta)
  \leq \sup_{\theta \in A_i} \sum_{j \geq 0} 2^{j/2} \diam A_{ij} (\theta)
  \leq c \, w (A_i)
  \leq c \, w_A (r)
  \, .
\end{equation*}
Now, define the following admissible sequence $\A = (\A_j)_{j \geq 0}$ of partitions of $A$:
for $j \geq p$, let $\A_j = \bigcup_{i=1}^N \A_{i,j-p}$, and define $\A_j$ arbitrarily
for $0 \leq j \leq p-1$, since the left-hand side of~\eqref{eq:proof-mm-lower} only depends on $(\A_j)_{j \geq p}$.
This is indeed an admissible sequence:
note that $\A_p = \{ A_1, \dots, A_N \}$, so $|\A_p| = N \leq N_p$; while for $j \geq p+1$, one has $|\A_j | = \sum_{i=1}^N |\A_{i,j-p}| \leq N_p \cdot N_{j-p} \leq N_{j-1}^2 = N_j$ (since $j - 1 \geq p \geq 1$ and $N_{k}^2 = N_{k+1}$ for $k \geq 1$).
In addition, $\A_{j+1}$ is a refinement of $\A_{j}$ since each sequence $(\A_{ij})_{j \geq 0}$ is admissible.

Combining the previous bounds and using that $\log N_p = 2^p \log 2 \geq (2^p - 1)\log 2$,
\begin{align*}
  (2^p - 1) + \sup_{\theta \in A} \sum_{j \geq p} 2^{j/2} \diam A_j (\theta)
  &\leq (\log 2)^{-1} \log N_p + \max_{1 \leq i \leq N} \sup_{\theta \in A_i} \sum_{j \geq p} 2^{j/2} \diam A_{ij} (\theta) \\
  &\leq \max \Big\{ \frac{2}{\log 2}, c \Big\} \cdot \Big[ \log N (A, r) + w_A (r) \Big]
    \, ,
\end{align*}
which proves the upper bound in~\eqref{eq:mm-regret-sum}.

The second characterization~\eqref{eq:mm-regret-fixed} in Corollary~\ref{cor:majorizing-measure-regret} is deduced from the first.
Indeed, if $p^* (A) = - \infty$, then $\gamma^* (A) = \gamma_2 (A) = \gamma_2^{(0)} (A) < 1$, so that for any $p \geq 1$, $2^p-1 + \gamma_2^{(p)} (A) \geq 2^p-1 \geq 1 \geq 2^0 - 1 + \gamma_2^{(0)} = \gamma_2 (A)$ and therefore
\begin{equation*}
  \inf_{p \geq 0} \big\{ (2^p-1) + \gamma_2^{(p)} (A) \big\}
  = \gamma_2 (A)
  = \gamma^* (A)
  \, .
\end{equation*}
Now, if $p^* = p^* (A) \geq 0$, then (as $\gamma_2^{(p)} (A)$ is non-increasing in $p$) one has for $p \leq p^*$, $\gamma_2^{(p)} (A) \geq \gamma_2^{(p^*)} (A) \geq 2^{p^*}$, while for $p \geq p^*+1$ one has $2^p-1 \geq 2^{p-1} \geq 2^{p^*}$.
We deduce that
\begin{equation*}
  2^{p^*}
  \leq \inf_{p \geq 0} \big\{ (2^p-1) + \gamma_2^{(p)} (A) \big\}
  \leq (2^{p^*+1}-1) + \gamma_2^{(p^*+1)} (A)
  \leq 4 \cdot 2^{p^*}
\end{equation*}
where the last inequality uses that $\gamma_2^{(p^*+1)} (A) \leq 2^{p^*+1}$.
Thus, in all cases, one has $\gamma^* (A) \leq \inf_{p \geq 0} \{ (2^p-1) + \gamma_2^{(p)} (A) \} \leq 4 \gamma^* (A)$ so that~\eqref{eq:mm-regret-fixed} follows from~\eqref{eq:mm-regret-sum}.

\subsection{Proof of Proposition~\ref{prop:correlation-local}}
\label{sec:proof-correlation-local}

We start with the following simple lemma.

\begin{lemma}
  \label{lem:local-global-widths}
  For any nonempty subset $A \subset \R^n$, one has $r_* (A)^2 \leq w (A)$.
  In addition, if $\diam^2 (A) \leq w (A)$ \textup(which holds in particular if $\diam (A) \leq 1/\sqrt{2\pi}$\textup)
  then $r_* (A)^2 = w (A)$.  
\end{lemma}

\begin{proof}
  For the first inequality, note that for any $r > \sqrt{w (A)}$ and $\theta \in A$ one has
  $w ( A \cap B (\theta, r)) \leq w (A) < r^2$, so that $r_* (A) \leq \sqrt{w (A)}$.

  Now, assume that $\diam^2 (A) \leq w (A)$ and let $r = \sqrt{w (A)}$.
  Then by assumption one has $r \geq \diam (A)$, so for any $\theta \in A$ one has $w (A \cap B(\theta, r)) = w (A) = r^2$ and thus $r_* (A) \geq r = \sqrt{w (A)}$.  
  The fact that $\diam (A) \leq 1/\sqrt{2\pi}$ implies $\diam^2 (A) \leq w (A)$ is pointed out in Proposition~\ref{prop:small-scale}.
\end{proof}

We now proceed with the proof of Proposition~\ref{prop:correlation-local}.
By definition, we have $\corr (A) = \sup_{\theta \in A} \corr (\theta, A)$ and $r_* (A) = \sup_{\theta \in A} r_* (\theta, A)$, where for $\theta \in A$ we let
\begin{align*}
  r_* (\theta, A)
  &= \sup \Big\{ r \geq 0 : w \big( A \cap B (\theta, r) \big) \geq r^2 \Big\} \, ; \\
  \corr (\theta, A)
  &= \corr (p_\theta, \probas_A)
    \, .
\end{align*}
Therefore, in order to prove Proposition~\ref{prop:correlation-local} it suffices to prove the corresponding inequalities for $\corr (\theta, A)$ and $r_* (\theta, A)$.
In addition, without loss of generality one may assume that $\theta = 0$, so in particular $0 \in A$.
Let $r_* = r_* (0, A)$ and $\corr = \corr (0, A)$.
We have, letting $X \sim \gaussdist (0, I_d)$,
\begin{align}
  \corr
  &= \E \Big[ \ell (p_0, X) - \inf_{\theta \in A} \ell (p_\theta, X) \Big] \nonumber \\
  &= \E \Big[ \frac{\norm{X}^2}{2} - \inf_{\theta \in A} \frac{\norm{\theta - X}^2}{2} \Big] 
    = \E \Big[ \sup_{\theta \in A} \Big\{ \innerp{\theta}{X} - \frac{\norm{\theta}^2}{2} \Big\} \Big]
    \label{eq:corr-zero-offset}
    \, .
\end{align}
This quantity is known as the ``offset Gaussian average'' of the set $A$~\cite[\S18.3]{rakhlin2022mathematical}.
Inequalities of the form~\eqref{eq:corr-upper-lower} between $r_*$ and $\corr$ can be found in~\cite[Lemma 33]{rakhlin2022mathematical} in the case where $A = K$ is a convex set, and in what follows we extend this result to the nonconvex case.

For the lower bound, let $r < r_*$ be arbitrary.
By definition of $r_*$, there exists $r' > r$ such that $w (A \cap r' B_2) \geq r'^2$, and then
\begin{align*}
  \corr
  &\geq \E \Big[ \sup_{\theta \in A \cap r' B_2} \Big\{ \innerp{\theta}{X} - \frac{\norm{\theta}^2}{2} \Big\} \Big]
    \geq \E \Big[ \sup_{\theta \in A \cap r' B_2} \innerp{\theta}{X} - \frac{r'^2}{2} \Big] \\
  &= w \big( A \cap r' B_2 \big) - \frac{r'^2}{2}
    \geq \frac{r'^2}{2}
    \geq \frac{r^2}{2}
    \, ,
\end{align*}
and since $r < r_*$ was arbitrary we obtain that $\corr \geq r_*^2/2$.

We now turn to the upper bound.
We start by noting that
\begin{equation}
  \label{eq:sup-double-sup}
  \sup_{\theta \in A} \Big\{ \innerp{\theta}{X} - \frac{\norm{\theta}^2}{2} \Big\}
  = \sup_{r \geq 0} \Big\{ \sup_{\theta \in A \cap r B_2} \innerp{\theta}{X} - \frac{r^2}{2} \Big\}
  \, .
\end{equation}
Indeed, the lower bound in~\eqref{eq:sup-double-sup} is shown as above, by restricting to $A \cap r B_2$ for any $r \geq 0$; while the upper bound is obtained by taking $r = \norm{\theta}$ for any $\theta \in A$.

In what follows, we let
\begin{equation}
  \label{eq:condition-r0}
  r_0 > \max \Big( {12 \sqrt{2}} e^{-1/2}, \frac{12}{\sqrt{2}} \, r_* (A/12) \Big)
\end{equation}
be arbitrary, and set $r_k = r_0 2^{k/2}$ for any integer $k \geq 1$.
In addition, for any $r \geq 0$ and $x \in \R^n$ we define $h_r (x) = \sup_{A \cap r B_2} \innerp{\theta}{x}$.
We want to control, with high probability, the quantity
\begin{equation*}
  \sup_{r \geq 0} \Big[ h_r (X) - \frac{r^2}{2} \Big]
  = \max \Big\{ \sup_{0 \leq r < r_0} \Big[ h_r (X) - \frac{r^2}{2} \Big], \sup_{r \geq r_0} \Big[ h_r (X) - \frac{r^2}{2} \Big] \Big\}
  \, .
\end{equation*}
Using that the function $r \mapsto h_r (X)$ is non-decreasing, we have that
\begin{equation*}
  \sup_{0 \leq r < r_0} \Big[ h_r (X) - \frac{r^2}{2} \Big]
  \leq \sup_{0 \leq r < r_0} h_r (X)
  \leq h_{r_0} (X)
  \, .
\end{equation*}
In addition, for any $r \geq r_0$, let $k \geq 1$ be the largest integer such that $r_{k-1} \leq r$, so that $r_{k-1} = r_k/\sqrt{2} \leq r < r_k$.
We then have
\begin{equation*}
  h_r (X) - \frac{r^2}{2}
  \leq h_{r_k} (X) - \frac{(r_{k-1}/\sqrt{2})^2}{2}
  = h_{r_k} (X) - \frac{r_k^2}{4}
  \leq \sup_{k' \geq 1} \Big[ h_{r_{k'}} (X) - \frac{r_{k'}^2}{4} \Big]
  \, .
\end{equation*}
This holds in particular for $r = r_0$, so $h_{r_0} (X) \leq r_0^2/2 + \sup_{k \geq 1} \big[ h_{r_k} (X) - r_k^2/4 \big]$.
Putting these inequalities together, we obtain that
\begin{align}
  \label{eq:bound-sup-r}
  \sup_{r \geq 0} \Big[ h_r (X) - \frac{r^2}{2} \Big]
  &\leq \max \Big\{ h_{r_0} (X), \sup_{k \geq 1} \Big[ h_{r_k} (X) - \frac{r_k^2}{4} \Big] \Big\} \nonumber \\
  &\leq \frac{r_0^2}{2} + \sup_{k \geq 1} \Big[ h_{r_k} (X) - \frac{r_k^2}{4} \Big]
    \, .
\end{align}
We therefore need to control the right-hand side of~\eqref{eq:bound-sup-r} with high probability.

Fix $k \geq 1$.
The function $h_{r_k} : \R^n \to \R$ is Lipschitz with constant $r_k$, so the Gaussian concentration inequality (\eg, \cite[Theorem~5.6 p.~126]{boucheron2013concentration}) implies that, for any $\delta > 0$,
\begin{equation*}
  \P \Big( h_{r_k} (X) \geq \E [h_{r_k} (X)] + \sqrt{2 r_k^2 \log (2^k/\delta)} \Big)
  \leq \delta/2^k
  \, .
\end{equation*}
Now, note that $\E [h_{r_k} (X)] = w (A \cap r_k B_2)$.
In addition, for any $c \geq 1$ and any $r > c \, r_* (A/c)$, one has by definition that
$w (A/c \cap (r/c) B_2) < (r/c)^2$, namely $w (A \cap r B_2) < r^2/c$.
In the case of $r = r_k$, one has $r_k = r_0 2^{k/2} \geq \sqrt{2} r_0 > 12 \, r_* (A/12)$ due to the condition~\eqref{eq:condition-r0} on $r_0$, so that $\E [h_{r_k} (X)] = w (A \cap r_k B_2) < r_k^2/12$.

We may also bound
\begin{align*}
  \sqrt{2 r_k^2 \log (2^k/\delta)}
  &\leq \sqrt{2 r_k^2 \log (1/\delta)} + r_k \sqrt{2 \log (2^k)} \\
  &\leq \frac{r_k^2}{12} + 12 \log (1/\delta) + \frac{r_k^2}{12} \cdot \frac{12 \sqrt{2}}{r_0} \sqrt{\frac{\log (2^k)}{2^k}} \\
  &\leq \frac{r_k^2}{12} + 12 \log (1/\delta) + \frac{r_k^2}{12} 
    \, ,
\end{align*}
where we used that $\log (2^k)/2^k \leq e^{-1}$ and that $r_0 \geq 12 \sqrt{2} e^{-1/2}$ from the condition~\eqref{eq:condition-r0}.
Putting the previous inequalities together, we obtain that
\begin{equation*}
  \P \Big( h_{r_k} (X) \geq \frac{r_k^2}{4} + 12 \log (1/\delta) \Big)
  \leq \frac{\delta}{2^k}
  \, .
\end{equation*}
By a union bound over $k \geq 1$, the last inequality entails that, for any $\delta \in (0,1)$,
\begin{equation*}
  \P \Big( \sup_{k \geq 1} \Big\{ h_{r_k} (X) - \frac{r_k^2}{4} \Big\} \geq 12 \log (1/\delta) \Big)
  \leq \sum_{k\geq 1} \frac{\delta}{2^k}
  = \delta
  \, ,
\end{equation*}
so that
\begin{equation*}
  \E \Big[ \sup_{k \geq 1} \Big\{ h_{r_k} (X) - \frac{r_k^2}{4} \Big\} \Big]
  \leq \int_0^1 12 \log (1/\delta) \di \delta
  = 12
  \, .
\end{equation*}
Plugging this bound into~\eqref{eq:bound-sup-r}, and taking the infimum over values of $r_0$ satisfying~\eqref{eq:condition-r0} gives
\begin{align*}
  \corr
  = \E \Big[ \sup_{r \geq 0} \Big\{ h_r (X) - \frac{r^2}{2} \Big\} \Big]
  &\leq \frac{1}{2} \max \Big( {12 \sqrt{2}} e^{-1/2}, \frac{12}{\sqrt{2}} \, r_* (A/12) \Big)^2 + 12 \\
  &\leq 36 \, r_* (A/12)^2 + 144 e^{-1} + 12 \\
  &\leq 36 \, r_* (A/12)^2 + 65 \, ,
\end{align*}
which proves~\eqref{eq:corr-upper-lower}.

The bound $\corr (A) \leq r_*^2 (A)$ when $\diam (A) \leq 1/\sqrt{2\pi}$ follows from the fact that in this case $r_*^2 (A) = w (A)$ (by Lemma~\ref{lem:local-global-widths}), while $\corr (A) \leq w (A)$ (by dropping the $-\norm{\theta}^2/2$ term in~\eqref{eq:corr-zero-offset}).

Finally, let $K \subset \R^n$ be convex.
If $\diam (K) > 2\sqrt{2/\pi}$, then there exist $\theta_1, \theta_2 \in K$ such that $\norm{\theta_1-\theta_2} \geq 2\sqrt{2/\pi}$, and by convexity of $K$ the segment $[\theta_1,\theta_2]$ is contained in $K$.
Hence $K$ contains a segment of length $2 \sqrt{2/\pi}$, thus $r_* (K) \geq r_* ([-\sqrt{2/\pi}, \sqrt{2/\pi}])$, and a direct computation shows that $r_* ([- \sqrt{2/\pi}, \sqrt{2/\pi}]) = \sqrt{2/\pi}$.
In addition, one has $r_* (K/12) \geq r_* (K)$: indeed, up to translating we may assume that $0 \in K$, in which case $K/12 \subset K$ by convexity of $K$.
We then have from~\eqref{eq:corr-upper-lower} that
\begin{equation*}
  \corr (K)
  \leq 36 \, r_* (K/12)^2 + 65
  \leq \Big( 36 + \frac{65 \pi}{2} \Big) r_* (K)^2
  \leq 140 \, r_* (K)^2
  \, .
\end{equation*}
Now if $\diam (K) \leq 2/\sqrt{2/\pi}$, then $\diam (K/4) \leq 1/\sqrt{2\pi}$ so by the previous point one has $\corr (K/4) \leq r_* (K/4)^2 \leq r_* (K)^2$ by convexity of $K$.
Finally note that for any $A \subset \R^n$ and constant $\lambda \geq 1$, one has
$\corr (\lambda A) \leq \lambda \corr (A)$.
Indeed, up to translating it suffices to show that $\corr (0, \lambda A) \leq \lambda \corr (0, A)$ when $0 \in A$.
By the expression~\eqref{eq:corr-zero-offset},
\begin{equation*}
  \corr (0, \lambda A)
  = \E \Big[ \sup_{\theta \in A} \Big\{ \lambda \innerp{\theta}{X} - \frac{\lambda^2 \norm{\theta}^2}{2} \Big\} \Big]
  \leq \E \Big[ \sup_{\theta \in A} \Big\{ \lambda \innerp{\theta}{X} - \frac{\lambda \norm{\theta}^2}{2} \Big\} \Big]
  = \lambda \corr (A)
  \, .
\end{equation*}
Thus $\corr (K) \leq 4 \corr (K/4) \leq 4 \, r_*^2 (K)$ in this case.

\appendix

\section{Remaining %
  proofs}
\label{sec:other-results}

\subsection{Proof of Proposition~\ref{prop:regret-max-intrinsic}}
\label{sec:proof-max-intrinsic}

Proposition~\ref{prop:regret-max-intrinsic} is a consequence of Theorem~\ref{thm:minimax-wills-volumes} (recalling also that $V_1 (K/\sqrt{2\pi}) = w (K)$), of the Poisson-log-concave inequality~\eqref{eq:alexandrov-fenchel}, and of the following lemma.

\begin{lemma}
  \label{lem:max-sum-ulc}
  Let $(v_j)_{j \geq 0}\in \R_+^\N$ be a finitely supported sequence such that for every $j \geq 1$,
  \begin{equation}
    \label{eq:ulc-ineq}
    (j+1) \, v_{j+1} v_{j-1}
    \leq j \, v_j^2
    \, .
  \end{equation}
  Assume also that $v_0 = 1$ and that $v_1 \geq 2$.
  Then, one has
  \begin{equation*}
    \log \bigg( \sum_{j \geq 0} v_j \bigg)
    \leq 4 \log \Big( \max_{j \geq 1} v_j \Big)
    \leq 8 \log \Big( \max_{k \geq 0} v_{2^k} \Big)
    \, .
  \end{equation*}
\end{lemma}

\begin{proof}
  Inequality~\eqref{eq:ulc-ineq} implies that if $n = \inf \{ j \geq 1 : v_j = 0 \} \in \N$, then $v_j = 0$ for $j \geq n$, and that the sequence $(\alpha_j)_{1 \leq j \leq n}$ defined by
  $\alpha_j
  = j v_j / v_{j-1} %
  $
  is non-increasing.
  Also let $\alpha_0 = + \infty$, and define
  \begin{equation*}
    j^*
    = \max \big\{ j \geq 0 : \alpha_{j} \geq j \big\}
    = \max \big\{ j \geq 0 : v_{j} \geq v_{j-1} \big\}
    \, .  
  \end{equation*}
  By definition, one has $v_j < v_{j-1}$ for $j > j^*$, and since $\alpha_j$ is non-increasing, $\alpha_j /j = v_j / v_{j-1}$ is decreasing so that for $j \leq j^*$, one has $v_j / v_{j-1} \geq v_{j^*} / v_{j^*-1} \geq 1$, with strict inequality if $j < j^*$.
  In other words,
  the sequence $(v_j)_{j \geq 0}$ increases on $\{ 0, \dots, j^* - 1\}$, reaches its maximum at $j^*$ (possibly also at ${j^*-1}$), and decreases on $\{ j^*, j^*+1 ,\dots \}$.

  Now since $\alpha_1 = v_1 \geq 1$, one has $j^* \geq 1$; in addition, $v_1 = \alpha_1 \geq \alpha_{j^*} \geq j^*$, so in particular $j^* \leq v_1 \leq v_{j^*}$ and
  \begin{equation*}  
    \sum_{j=0}^{2 j^*} v_j
    \leq (2 j^* + 1) v_{j^*}
    \leq 2 v_{j^*}^2 + v_{j^*}
    \, .
  \end{equation*}
  In addition, since $(\alpha_j)$ is non-increasing, for $j \geq 2 j^* + 2$ one has $\alpha_j \leq \alpha_{j^*+1} < j^* + 1$ so $v_j / v_{j-1} \leq (j^*+1)/j \leq 1/2$.
  Hence for $k \geq 0$, one has $v_{2j^* + 1 + k} \leq 2^{-k} v_{j^*}$ and thus
  \begin{equation*}
    \sum_{j \geq 2 j^* + 1} v_j
    \leq \sum_{k \geq 0} 2^{-k} v_{j^*}
    \leq 2 v_{j^*} 
    \, .
  \end{equation*}
  Combining the previous two bounds gives (recalling that $v_{j^*} \geq v_1 \geq 2$)
  \begin{equation*}
    \log \bigg( \sum_{j \geq 0} v_j \bigg)
    \leq \log \big( 2 v_{j^*}^2 + 3 v_{j^*} \big)
    \leq \log (3.5 \cdot v_{j^*}^2)
    \leq 4 \log v_{j^*}
    \, .
  \end{equation*}
  This corresponds to the first inequality in Lemma~\ref{lem:max-sum-ulc}.
  To prove the second one, note that~\eqref{eq:ulc-ineq} implies in particular the log-concavity inequality $v_{j+1} v_{j-1} \leq v_j^2$ for $j \geq 1$, which in turns implies that the sequence $(v_j^{1/j})_{j \geq 1}$ is non-increasing.
  As a result, let $1 \leq j \leq j^*$ be such that $j^* \leq 2 j$.
  Recalling that $v_{j^*} \geq 1$, one has:
  \begin{equation*}
    v_{j}^{1/j}
    \geq v_{j^*}^{1/j^*}
    \geq v_{j^*}^{1/(2j)}
    \, ,
  \end{equation*}
  namely $v_{j^*} \leq v_j^2$ \ie $\log v_{j^*} \leq 2 \log v_j$.
  Letting $j = 2^k$ where $k \geq 0$ is the maximal integer such that $2^k \leq j^*$ concludes the proof.
\end{proof}

\subsection{Proof of Proposition~\ref{prop:large-scale}}
\label{sec:proof-large-scale}

We use the following representation of the Wills functional,
which corresponds (up to normalization by $\sqrt{2\pi}$) to~\eqref{eq:other-integral-steiner} in the proof of Theorem~\ref{thm:minimax-wills-volumes}
or alternatively~\cite[Eq.~(1.4)]{hadwiger1975wills}: 
\begin{equation*}
  W (A)
  = 2 \pi \int_0^\infty \vol_n (A + r B_2^n) r e^{-\pi r^2} \di r
  \, .
\end{equation*}
In particular, for every $t > 0$ one has:
\begin{align*}
  W (t A) / t^n
  &= 2 \pi \int_0^\infty t^{-n} \vol_n (t A + r B_2^n) r e^{-\pi r^2} \di r \\
  &= 2 \pi \int_0^\infty \vol_n \Big( A + \frac{r}{t} B_2^n \Big) r e^{-\pi r^2} \di r
    \, .
\end{align*}
Now, the family of sets $A + \eps B_2^n$ is increasing (in the sense of inclusion) in $\eps > 0$, with intersection $A$ (since $A$ is closed) and contained for $0 < \eps \leq 1$ in the bounded
set $A + B_2^n$, so by dominated convergence one has $\vol_n (A + \eps B_2^n) \to \vol_n (A)$ as $\eps \to 0$.
In particular, for any fixed $r \geq 0$ one has $\vol_n \big( A + \frac{r}{t} B_2^n \big) r e^{-\pi r^2} \to \vol_n (A) r e^{-\pi r^2}$ as $t \to \infty$, and additionally for every $t \geq 1$ and $r \geq 0$ one has $| \vol_n ( A + \frac{r}{t} B_2^n ) r e^{-\pi r^2} | \leq \kappa_n (\diam (A) + r)^n r e^{-\pi r^2}$ which is integrable over $\R^+$, so by dominated convergence, as $t \to \infty$:
\begin{equation*}
  W (t A) / t^{n}
  \to 2 \pi \int_0^\infty \vol_n (A) r e^{-\pi r^2} \di r
  = \vol_n (A)
  \, .
  \qedhere
\end{equation*}

\subsection{Simplified
  form of minimax regret for convex bodies
}
\label{sec:altern-expr-minim}

Theorem~\ref{thm:minimax-regret-metric} characterizes the minimax regret in terms of two fixed points for general sets $A \subset \R^n$.
We now discuss another characterization in the case where $A = K$ is a convex body.

Since $\mmr (K) \asymp \mmr (K - K)$ by Proposition~\ref{prop:sum-difference-bodies}, one may assume that $K$ is a symmetric convex body.
In this case, Lemma~\ref{lem:local-gaussian-width} ensures that $w_K (r) = w (K \cap r B_2)$ and that $w_K (r) / r$ is decreasing in $r > 0$.
Now, one may write:
\begin{equation*}
  r_* (K)
  = \sup \Big\{ r \geq 0 : w_K (r) \geq r^2 \Big\}
  = \sup \Big\{ r \geq 0 : \Big( \frac{w_K (r)}{r} \Big)^2 \geq r^2 \Big\}
  \, .
\end{equation*}
Then, using that $w_K (r)/r = w ((r^{-1} K) \cap B_2)$ is decreasing in $r$ and Lemma~\ref{lem:prop-fixed-points},
\begin{equation}
  \label{eq:fixed-local-dec}
  \frac{1}{2} \inf_{r > 0} \Big\{ w \big( r^{-1} K \cap B_2 \big) + r^2 \Big\}
  \leq r_*^2 (K)
  \leq \inf_{r > 0} \Big\{ w \big( r^{-1} K \cap B_2 \big)^2 + r^2 \Big\}
  \, .
\end{equation}
We emphasize that the variational representation~\eqref{eq:fixed-local-dec} %
is specific to the convex case, since for general $A$ the map $r \mapsto w_A (r)/r$ may not be decreasing.
Now, for any decreasing functions $g, h : \R^+ \to \R^+$,
\begin{equation*}
  \frac{1}{2} \Big( \inf_{r > 0} \big\{ g (r) + r^2 \big\} + \inf_{r > 0} \big\{ h (r) + r^2 \big\} \Big)
  \leq \inf_{r > 0} \big\{ g (r) + h (r) + r^2 \big\}
  \leq \inf_{r > 0} \big\{ g (r) + r^2 \big\} + \inf_{r > 0} \big\{ h (r) + r^2 \big\}
  .
\end{equation*}
Combining the above leads to the following alternative characterization:
\begin{corollary}
  \label{cor:regret-sum-symmetric}
  Let $K \subset \R^n$ be an origin-symmetric convex body.
  Then,
  \begin{equation}
    \label{eq:regret-sum-symmetric}
    \begin{split}
      \frac{1}{600} \inf_{r > 0} \Big\{ \log N (K, r) + w \big( &r^{-1} K \cap B_2 \big)^2 + r^2 \Big\}
    \leq \mmr (K) \\
    &\leq 4 \inf_{r > 0} \Big\{ \log N (K, r) + w \big( r^{-1} K \cap B_2 \big)^2 + r^2 \Big\}
    \, .
    \end{split}
  \end{equation}
\end{corollary}

The interest of this representation is that the two ``complexity'' parameters, here $\log N (K, r)$ and $w (r^{-1} K \cap B_2)^2$, are both homogeneous to ``dimensions'' rather than to (powers of) distances.
The parameter $w (r^{-1} K \cap B_2)^2 = ( w (K \cap r B_2) / r)^2$ has (for $r < \diam (K) /2$) a natural interpretation: it corresponds to the so-called \emph{Dvoretzky-Milman dimension} (see~\cite[Chapter~5]{artstein2015asymptotic}) of the localized set $K \cap r B_2$.

\subsection{Short proof of McMullen's inequality}
\label{sec:short-proof-mcmull}

In this section, we provide a short proof of McMullen's inequality~\eqref{eq:wills-first-iv} that $\log W (K) \leq V_1 (K)$ for any convex body $K \subset \R^n$.

\begin{proof}
  Let $K \subset \R^n$ be a convex body.
  One has, as $t \to 0$,
  \begin{equation*}
    \log W (t K)
    = \log \bigg( \sum_{j=0}^n t^j V_j (K) \bigg)
    = t \cdot V_1 (K) + O (t^2)
    \, ,
  \end{equation*}
  so that %
  $\log W (t K) / t \to V_1 (K)$
  as $t \to 0$.
  On the other hand, by concavity of $\log W$
  (\cite{alonso2021further}, see Theorem~\ref{thm:concavity}),
  the function $t \mapsto \log W (t K)/t$ is non-increasing on $\R_+^*$.
  This implies that
  \begin{equation*}
   \log W (K)
   = \frac{\log W (1 \cdot K)}{1}
   \leq \lim_{t \to 0} \frac{\log W (t K)}{t}
   = V_1 (K)
    \, . \qedhere
  \end{equation*}
\end{proof}

\paragraph{Acknowledgements.}

The author warmly thanks Guillaume Lecué for several stimulating discussions, Yanjun Han, Yury Polyanskiy, Alexander Rakhlin and Philippe Rigollet
for helpful comments, as well as Geoffrey Chinot and Nikita Zhivotovskiy for a careful reading of a first version of this manuscript.
This research is supported by a grant of the French National Research Agency (ANR), “Investissements d’Avenir” (LabEx Ecodec/ANR-11-LABX-0047).

{\small

\newcommand{\etalchar}[1]{$^{#1}$}
  
}

\end{document}